\documentclass[12pt,a4paper]{article}

\RequirePackage{fullpage}

\title{%
  Partial Orders for Precise and Efficient Dynamic Deadlock Prediction~\footnote{See corrigendum after appendix.}
}

\author{Bas van den Heuvel, Martin Sulzmann, and Peter Thiemann}

\usepackage{amsfonts}

\RequirePackage{amsthm,thmtools}
\newtheorem{definition}{Definition}

\newtheorem{lemma}{Lemma}
\newtheorem{theorem}{Theorem}

\usepackage{xparse}

\usepackage{listings}
\usepackage{comment}

\usepackage{wrapfig}

\usepackage{multirow}

\usepackage{pgf}
\usepackage{tikz}
\usepackage{subcaption}

\usepackage{adjustbox}
\usepackage{color}
\definecolor{GrayBgColor}{rgb}{0.9, 0.9, 0.9}
\definecolor{GrayBgBColor}{rgb}{0.7, 0.7, 0.7}

\newcommand*{\backref}[1]{}
\newcommand*{\backrefalt}[4]{%
    \ifcase #1%
          \or Cited on page~#2.%
          \else Cited on pages~#2.%
    \fi%
    }

\makeatletter
\newcommand{\xRightarrow}[2][]{\ext@arrow 0359\Rightarrowfill@{#1}{#2}}
\makeatother

\newcommand{\bi}{\begin{array}[t]{@{}l@{}}}
\newcommand{\ei}{\end{array}}
\newcommand{\ba}{\begin{array}}
\newcommand{\ea}{\end{array}}
\newcommand{\bda}[1]{\begin{displaymath}\ba{#1}}
\newcommand{\eda}{\ea\end{displaymath}}
\newcommand{\bp}{\begin{quote}\tt\begin{tabbing}}
\newcommand{\ep}{\end{tabbing}\end{quote}}

\newcommand{\ignore}[1]{}

\newcommand{\mathem}{\sf}

\newcommand{\figurebox}[1]
        {\fbox{\begin{minipage}{\textwidth} #1 \end{minipage}}}

\newcommand{\turns}{\, \vdash \,}

\newcommand\Angle[1]{\langle#1\rangle}

\newcommand{\conc}{\cdot}

\newcommand{\thread}[1]{\ensuremath{\tau_{#1}}}

\newcommand{\Vacq}{V_{acq}}
\newcommand{\Vrel}{V_{rel}}

\newcommand{\evt}[2]{#1_{#2}}

\newcommand{\eventE}[1]{\ensuremath{e_{#1}}}

\newcommand{\mycolorbox}[2]{\adjustbox{margin=.6\fboxsep,bgcolor=#1,margin=-.6\fboxsep}{#2}}

\newcommand{\HIGHLIGHT}[1]{\mycolorbox{GrayBgColor}{\ensuremath{#1}}}

\newcommand{\HIGHLIGHTB}[1]{\mycolorbox{GrayBgBColor}{\ensuremath{#1}}}

\newcommand{\uline}[1]{\rule[0pt]{#1}{0.4pt}}
\newcommand{\dontCare}{\uline{.15cm}}

\newcommand\evtAA{e}
\newcommand\evtBB{f}
\newcommand\evtCC{g}

\newcommand\evtAcc{a}
\newcommand\evtAccB{b}

\newcommand\evtAccA{\evtAcc'}
\newcommand\evtRel{r}
\newcommand\evtRelA{\evtRel'}

\newcommand{\incC}[2]{{\mathem inc}(#1,#2)}

\newcommand{\maxN}[2]{{\mathem max}(#1,#2)}

\newcommand{\MAKECHAN}[1]{\mbox{\mathem make}({\mathem chan} \ #1)}

\newcommand{\semP}[3]{#1 \xRightarrow{#2} #3}

\newcommand{\LKA}{\LK1}
\newcommand{\LKB}{\LK2}
\newcommand{\LKC}{\LK3}
\newcommand{\LKD}{\LK4}
\newcommand{\LKE}{\LK5}

\newcommand{\VA}{x}
\newcommand{\VB}{y}
\newcommand{\VC}{z}

\newcommand{\lockE}[1]{\textit{acq}(#1)} 
\let\acqE\lockE
\newcommand{\reqE}[1]{\textit{req}(#1)}
\newcommand{\reqLockE}[1]{\reqE{#1}} 
\newcommand{\unlockE}[1]{\textit{rel}(#1)} 
\let\relE\unlockE
\newcommand{\readE}[1]{rd(#1)}
\newcommand{\writeE}[1]{wr(#1)}
\newcommand{\forkE}[1]{\textit{fork}(#1)}

\newcommand\TrSymbol{\mbox{\tiny Tr}}
\newcommand{\TrLt}[1][]{\mathrel{<_{\TrSymbol}^{#1}}}
\newcommand\POSymbol{\mbox{\tiny PO}}
\newcommand\POLt[1][]{\mathrel{<_{\POSymbol}^{#1}}}
\newcommand\POLeq[1][]{\mathrel{\leq_{\POSymbol}^{#1}}}

\newcommand{\hbSymbol}{\mbox{\tiny HB}}

\newcommand{\wcpSymbol}{\mbox{\tiny WCP}}

\newcommand{\HBLt}[1][]{\mathrel{<_{\hbSymbol}^{#1}}}
\newcommand{\WCPLt}[1][]{\mathrel{<_{\wcpSymbol}^{#1}}}

\newcommand{\mhbSymbol}{\mbox{\tiny MHB}}
\newcommand{\cmhbNoSymbol}{\mbox{\tiny NoCMHB}}
\newcommand{\cmhbFullSymbol}{\mbox{\tiny FullCMHB}}
\newcommand{\dpSymbol}{\mbox{\tiny DP}}

\newcommand{\CMHBLt}[2][]{\mathrel{<_{\mhbSymbol}^{#1(#2)}}}
\newcommand{\DPLt}[1][]{\mathrel{<_{\dpSymbol}^{#1}}}

\newcommand{\CMHBConc}[2][]{\mathrel{{||}_{\mhbSymbol}^{#1}}}

\RequirePackage{amsmath}
\DeclareMathOperator\crp{crp}
\DeclareMathOperator\TRWcrp{TRWcrp}

\newcommand{\ACQ}[1]{\evtAcc_{#1}}
\newcommand{\REL}[1]{\evtRel_{#1}}
\newcommand{\ReqSymb}{q}
\newcommand{\REQ}[1]{\ReqSymb_{#1}}

\newcommand{\REQB}[1]{\ReqSymb'_{#1}}

\newcommand{\LK}[1]{l_{#1}}
\let\THD\thread

\newcommand{\posP}[2]{\textit{pos}_{{\scriptstyle #1}}(#2)}

\newcommand{\threadVC}[1]{\textit{Th}(#1)}

\newcommand{\lastWriteVC}[1]{\ensuremath{L_W}(#1)}

\newcommand{\lastReadVC}[1]{\ensuremath{L_R}(#1)}

\newcommand{\accVC}[2]{#1[ #2 ]}

\newcommand\conf{\mathrel{\bowtie}}

\newcommand{\LD}[3]{{\langle #1, #2, #3 \rangle}}

\newcommand{\LDt}{\ensuremath{\mathcal{L}_{\mathit D}}} 

\newcommand{\LDMapSym}{\ensuremath{\mathcal{M}}}
\newcommand{\LDMap}[3]{\LDMapSym{\langle #1, #2, #3 \rangle}}

\newcommand{\GCMapSym}{\ensuremath{\mathcal{G}}}
\newcommand{\GlobalLS}{{\mathit  All_{lh}}}

\newcommand{\Ald}{F} 
\newcommand{\AldE}{E} 

\newcommand{\DDs}{{\mathcal{D}}}
\newcommand{\AAs}{{\mathcal{A}}}
\newcommand{\BBs}{{\mathcal{B}}}

\newcommand{\standard}{S}
\newcommand{\intrathread}{C}
\newcommand{\intrathreadthread}{CT}
\newcommand{\refinedintrathread}{R}

\newcommand{\pwrsymbol}{\mbox{\tiny PWR}}

\newcommand{\PWRLt}[1][]{\mathrel{<_{\pwrsymbol}^{#1}}}
\newcommand{\PWRConc}[1][]{\mathrel{{||}_{\pwrsymbol}^{#1}}}

\newcommand{\lwSymbol}{\mbox{\tiny LW}}

\newcommand{\trwSymbol}{\mbox{\tiny TRW}}
\newcommand{\TRWLt}[1][]{\mathrel{<_{\trwSymbol}^{#1}}}
\newcommand{\TRWConc}[1][]{\mathrel{{||}_{\trwSymbol}^{#1}}}

\newcommand{\VCConc}{\mathrel{||}}

\makeatletter
\DeclareMathOperator\@CS{CS}
\newcommand\CS[1][]{\@CS_{#1}}
\makeatother

\makeatletter
\DeclareMathOperator\@AH{AH}
\newcommand\AH[1][]{\@AH_{#1}}
\makeatother

\makeatletter
\DeclareMathOperator\@LH{LH}
\newcommand\LH[1][]{\@LH_{#1}}
\makeatother

\newcommand{\LHSym}[1]{{LH_{#1}}}           
\newcommand{\LocksSym}[1]{\ensuremath{\mathcal{L}_{#1}}}           

\newcommand{\AcqHeld}{\mathcal{A}_{H}}

\newcommand{\Hist}[1]{\ensuremath{\mathcal{H}(#1)}} 

\newcommand{\Acq}[1]{Acq(#1)}

\newcommand{\DP}[1]{\{ #1 \}}  

\newcommand{\rwT}[1]{T^{\textit{rw}}_{#1}}

\newcommand{\relT}{T^{\textit{rel}}}

\newcommand{\wwwSym}{<^{\scriptscriptstyle W3}}

\newcommand{\prefixOf}[2]{#1 \rhd #2}

\RequirePackage{xspace}

\def\Tunc{Tun\c{c}\xspace}

\newcommand{\SPDOfflineUD}{\mbox{SPDOffline$^{*}$}\xspace}    
\newcommand{\SPDOffline}{\mbox{SPDOffline}\xspace}            

\makeatletter
\newdimen\legendxshift
\newdimen\legendyshift
\newcount\legendlines
\newcommand{\bclldist}{1mm}
\newcommand{\bclegend}[3][10mm]{%
	\legendxshift=0pt\relax
	\legendyshift=0pt\relax
	\xdef\legendnodes{}%
	\foreach \lcolor/\ltext [count=\ll from 1] in {#3}%
	{\global\legendlines\ll\pgftext{\setbox0\hbox{\bcfontstyle\ltext}\ifdim\wd0>\legendxshift\global\legendxshift\wd0\fi}}%
	\@tempdima#1\@tempdima0.5\@tempdima
	\pgftext{\bcfontstyle\global\legendxshift\dimexpr\bcwidth-\legendxshift-\bclldist-\@tempdima-0.72em}
	\legendyshift\dimexpr5mm+#2\relax
	\legendyshift\legendlines\legendyshift
	\global\legendyshift\dimexpr\bcpos-2.5mm+\bclldist+\legendyshift
	\begin{scope}[shift={(\legendxshift,\legendyshift)}]
		\coordinate (lp) at (0,0);
		\foreach \lcolor/\ltext [count=\ll from 1] in {#3}%
		{
			\node[anchor=north, minimum width=#1, minimum height=5mm,fill=\lcolor] (lb\ll) at (lp) {};
			\node[anchor=west] (l\ll) at (lb\ll.east) {\bcfontstyle\ltext};
			\coordinate (lp) at ($(lp)-(0,5mm+#2)$);
			\xdef\legendnodes{\legendnodes (lb\ll)(l\ll)}
		}
		\node[draw, inner sep=\bclldist,fit=\legendnodes] (frame) {};
	\end{scope}
}
\makeatother

\RequirePackage{algorithm}
\RequirePackage{algorithmicx}
\PassOptionsToPackage{noend}{algpseudocode}
\RequirePackage{algpseudocode}
\algnewcommand{\IfThen}[2]{\State \algorithmicif\ #1\ \algorithmicthen\ #2}
\algnewcommand{\ForDo}[2]{\State \algorithmicfor\ #1\ \algorithmicdo\ #2}

\NewDocumentEnvironment{proofsketch}{}{\begin{proof}[Proof (sketch)]}{\end{proof}}

\newcommand\cond[1]{\textsf{[#1]}}
\DeclareMathOperator\evts{evts}
\DeclareMathOperator\thd{thd}

\DeclareMathOperator\thds{thds}
\DeclareMathOperator\proj{proj}

\makeatletter
\DeclareMathOperator\@mod{mod}
\renewcommand\mod{\mathbin{\%}}
\makeatother

\RequirePackage{xparse}
\let\oldfigure\figure
\let\endoldfigure\endfigure
\RenewDocumentEnvironment{figure}{O{}}{\oldfigure[#1]\small}{\endoldfigure}

\def\ih#1{IH\textsubscript{#1}\xspace}

\RequirePackage{hyperref}
\usepackage[capitalise,noabbrev]{cleveref}
\crefname{line}{line}{lines}

\begin{document}

\maketitle

\begin{abstract}
  Deadlocks are a major source of bugs in concurrent programs.
  They are hard to predict, because they may only occur under specific
  scheduling conditions.
  Dynamic analysis attempts to identify potential deadlocks by
  examining a single execution trace of the program.
  A standard approach involves monitoring sequences of lock
  acquisitions in each thread, with the goal of identifying deadlock
  patterns.
  A deadlock pattern is characterized by a cyclic chain of lock
  acquisitions, where each lock is held by one thread while being
  requested by the next.
  However, it is well known that not all deadlock patterns identified
  in this way correspond to true deadlocks, as they may be impossible to manifest under any schedule.

  We tackle this deficiency by proposing a new method \emph{based on partial orders} to eliminate false positives:
  lock acquisitions must be unordered under a given partial order, and not preceded by other deadlock patterns.
  We prove \emph{soundness} (no falsely predicted deadlocks) for the novel TRW partial order,
  and \emph{completeness} (no deadlocks missed) for a slightly weakened variant of TRW.
  Both partial orders can be computed efficiently and report the same
  deadlocks for an extensive benchmark suite.
\end{abstract}

\section{Introduction}
\label{sec:intro}

To fully leverage today’s multicore architectures, programs must be
designed for concurrency.
Yet, concurrent programming introduces
significant challenges, as developers must use synchronization
primitives, such as locks, to ensure correct behavior.
Incorrect
use of locks can lead to critical issues, including resource
deadlocks.
Resource deadlock occurs when two or more threads are
stuck waiting for each other’s locks, resulting in a standstill that
halts program execution.
These deadlocks are particularly difficult to
diagnose, as they may only appear intermittently, often after hundreds
of successful runs.

Program analysis can help  programmers detect potential resource
deadlocks.
In particular, dynamic analysis aims at predicting a program's behavior
under different schedules by analyzing a trace of events (including
acquires and releases of locks)
generated during a
program run.
A common approach to predicting deadlocks is to search a trace for \emph{deadlock patterns},
that involve series of lock acquire events with a cyclic dependency.
Deadlock patterns are
not sufficient to characterize resource deadlocks, meaning that they may lead
to false positives: predicted deadlocks that cannot manifest under any schedule.



\paragraph{Prior work.}

To eliminate false positives,
prior work~\cite{conf/oopsla/KalhaugeP18,10.1007/978-3-319-23404-5_13,10.1145/1542476.1542489,Samak:2014:TDD:2692916.2555262,10.1007/978-3-319-23404-5_13}
attempts to find a schedule which exposes the deadlock, either symbolically via
SMT solving or by program re-execution.
Such methods can be very costly, as we might need to exhaustively explore
all alternative schedulings of the execution.
To improve efficiency, recent works~\cite{conf/fse/CaiYWQP21,conf/pldi/TuncMPV23}
employ alternative solving methods to remove false positives.
However, these methods may still be costly (the time complexity of \cite{conf/fse/CaiYWQP21} is quadratic in trace length) and leave deadlocks unpredicted (\cite{conf/pldi/TuncMPV23} reports only so-called ``sync-preserving'' deadlocks).

\paragraph{Our novel approach.}

We introduce a new approach that eliminates false-positive deadlock patterns, inspired by partial-order methods known from data-race prediction.
Particularly, our work is novel in that it requires only an analysis of the ordering of trace events, in contrast to prior works that incorporate partial orders (such as~\cite{conf/oopsla/KalhaugeP18,conf/fse/CaiYWQP21,conf/pldi/TuncMPV23}) but need further steps akin to trace exploration.

To be precise, our approach refines the notion of deadlock pattern from the
literature by introducing the following two conditions:

\begin{description}

  \item[Partially-ordered acquires]
    A deadlock pattern is a false positive if some of its acquires are ordered.
    Hence, we only consider deadlock patterns where all acquires are concurrent (i.e., pairwise unordered).

  \item[Partially-ordered deadlock patterns]
    A deadlock pattern is a false positive if it is blocked by an earlier deadlock.
    We identify such situations by defining a partial order on deadlock
    patterns, and only consider the ``earliest'' deadlock patterns.
    %

\end{description}

\noindent
This way, given an appropriate partial order, our method instantly removes false positives without requiring extra steps or limiting the search to a subclass of deadlocks.

\paragraph{A new partial order for sound and efficient deadlock prediction.}

Thus, the challenge lies in finding an appropriate partial order among acquires
that fits the purpose of deadlock prediction (the partial order among deadlock patterns is fairly straightforward).
Existing partial orders employed in the context
of data race prediction~\cite{lamport1978time,Smaragdakis:2012:SPR:2103621.2103702,conf/pldi/KiniMV17,10.1145/3360605}
are not suitable, because when applied to deadlock detection we may end up with false positives \emph{and} false negatives.

In this paper, we introduce the novel TRW partial order.
We are able to show that under TRW all deadlock patterns correspond to true deadlocks.
However, we may encounter false negatives.
Earlier complexity results~\cite{conf/pldi/TuncMPV23} suggest that it is impossible to find
an \emph{efficient} deadlock-prediction method that is \emph{sound} (no false positives) and \emph{complete} (no false negatives).
Hence, we also consider PWR~\cite{conf/mplr/SulzmannS20},
and show that under PWR deadlock patterns are complete, although we may face some false positives.
Our experiments show that PWR and TRW are efficient and report the same set of deadlock patterns
for an extensive benchmark suite.


\paragraph{Contributions and outline.}

In summary, our contributions are:
\begin{itemize}

  \item
    We define precisely our refined deadlock patterns
    using partial orders
    among events and among deadlock patterns.
    For the TRW partial order we establish soundness (\cref{{s:sound}}), and
    for the PWR partial order we establish completeness
    (%
    \cref{s:complete}).

  \item We present an implementation of our approach as an offline version of
    the UNDEAD deadlock predictor~\cite{conf/ase/ZhouSLCL17}, with versions based
    on PWR and TRW
    (\cref{sec:implementation}).

  \item We study the impact on performance and precision
    of deadlock prediction under PWR and TRW
    (\cref{sec:experiments}).
    We also compare in detail against the sync-preserving deadlock
    predictor \SPDOffline~\cite{conf/pldi/TuncMPV23}.

\end{itemize}
\Cref{sec:overview} gives an overview of our work, and \cref{sec:prelim} introduces basic definitions and notations.
\Cref{sec:related,sec:conclusions} discuss related work and draw some conclusions, respectively.
The supplement can be safely ignored; it contains detailed proofs, additional examples and experimentation results, and preliminary access to our implementation (to be submitted as an artifact).


\section{Overview}
\label{sec:overview}

In this section, we recall a standard method for dynamic deadlock
detection~\cite{10.5555/645880.672085,DBLP:conf/spin/Harrow00}, point out its shortcomings, and explain our approach to
addressing these shortcomings.

Our discussion relies on traces like $T_1$ shown in \cref{fig:ex0}.
The left part of the diagram represents a program run by a trace of events.
We briefly explain the notation; \cref{sec:prelim} includes formal definitions and details.
The diagram visualizes the interleaved execution of the program's events using
a tabular notation with a separate column for each thread and one
event per row.
The order from top to bottom reflects the observed temporal order of events.

Each event takes place in a specific thread and
represents an operation.
We write $\LKA, \LKB, \ldots$ to denote locks and  $\VA, \VB, \ldots$ to denote shared variables.
Operations $\lockE{\LKA}$ and~$\unlockE{\LKA}$ acquire and release lock~$\LKA$, respectively.
Operations $\readE{\VA}$ and $\writeE{\VA}$ are read and write operations, respectively, on shared variable~$\VA$.
The same operation may appear multiple times in a trace, so
we use indices $e_i$ to uniquely identify events in the trace.
Traces are formally expressed as lists, e.g.,
$[e_1,e_2, e_3, e_4, e_5, e_6, e_7,e_8]$ for $T_1$.

\subsection{Lock Dependencies}

\begin{figure}[t]

  \begin{minipage}[b]{.44\textwidth}
    \bda{@{}lcl}
      \begin{array}{|l|l|l||l|}
        \hline
        T_1 & \thread{1} & \thread{2} & \mbox{Lock deps} \\
        \hline
        \eventE{1} & \lockE{\LKA} &&  \\
        \eventE{2} & \HIGHLIGHT{\lockE{\LKB}} && \LD{\thread{1}}{\LKB}{\{\LKA\}} \\
        \eventE{3} & \unlockE{\LKB} && \\
        \eventE{4} & \unlockE{\LKA} && \\
        \eventE{5} && \lockE{\LKB} & \\
        \eventE{6} && \HIGHLIGHT{\lockE{\LKA}} & \LD{\thread{2}}{\LKA}{\{\LKB\}} \\
        \eventE{7} && \unlockE{\LKA} & \\
        \eventE{8} && \unlockE{\LKB} & \\
        \hline
      \end{array}

    \eda
    \subcaption{Trace with lock dependencies.}
    \label{fig:ex0}
  \end{minipage}%
  \hfill%
  \begin{minipage}[b]{.55\textwidth}
    \bda{c}

\ba{|l|l|l|l|}
\hline T_2  & \thread{1} & \thread{2} & \mbox{Lock deps} \\ \hline
\eventE{1}  & \lockE{\LKA} && \\
\eventE{2}  & \lockE{\LKB}&& \LD{\thread{1}}{\LKB}{\{\LKA\}}\\
\eventE{3}  & \unlockE{\LKB}&&\\
\eventE{4}  & \unlockE{\LKA}&&\\
\eventE{5}  & \writeE{\VA}&&\\
\eventE{6}  & &\readE{\VA} &\\
\eventE{7}  & &\lockE{\LKB} & \\
\eventE{8}  & &\lockE{\LKA} & \LD{\thread{2}}{\LKA}{\{\LKB\}}\\
\eventE{9}  & &\unlockE{\LKA} &\\
\eventE{10}  & &\unlockE{\LKB} &\\

 \hline \ea{}

    \eda
    \subcaption{False positive due to last-write dependency.
    }\label{fig:ex_3}
  \end{minipage}
  \caption{Traces with potential deadlocks.}
  \label{f:firstExamples}
\end{figure}

In trace $T_1$ in \cref{fig:ex0}, all events in thread $\thread1$ take place before the
events in thread $\thread2$. A different schedule for the same program
might give rise to a \emph{reordered trace prefix} $[e_1,e_5]$, which
indicates that $T_1$ has the potential to deadlock: the highlighted acquire events $e_2$ and $e_6$
are \emph{enabled} to be scheduled next, but either
would break the lock semantics by acquiring a lock already acquired but not yet released by the other tread.

The  standard approach
predicts such situations by constructing a \emph{lock dependency} of
the form $\LD{t}{\LK{}}{L_t}$ for every acquire event~\cite{10.1007/11678779_15,10.1145/1542476.1542489}.
Here, $t$ is the thread that acquires lock $\LK{}$ and $L_t$ is the
set of locks (aka \emph{lockset}) held (acquired but not yet released) by this thread at the point of acquiring lock~$\LK{}$.
The right part of \cref{fig:ex0} shows the lock dependencies for $T_1$.

The two lock dependencies
$
\LD{\thread{1}}{\LKB}{\{\LKA\}}$ and $ \LD{\thread{2}}{\LKA}{\{\LKB\}}
$ obtained from $T_1$ indicate a potential deadlock, because they
exhibit a \emph{cyclic  chain} of lock dependencies according to the following two conditions:
\begin{description}
\item[\cond{DP-Cycle}] The acquired lock $\LKB$ of the first dependency is in the lockset $\{ \LKB \}$
of the second dependency, and the acquired lock $\LKA$ of the second dependency
is in the lockset $\{ \LKA \}$ of the first dependency.
\item[\cond{DP-Guard}] The underlying locksets $\{\LKA \}$ and $\{ \LKB \}$ are disjoint.
\end{description}
Condition~\cond{DP-Guard} ensures that the deadlocked situation shown by
the reordered prefix~$[e_1,e_5]$ can be reached without violating lock semantics by acquiring a lock already held by another tread.
Condition~\cond{DP-Cycle} characterizes the deadlocked situation, caused by a cycle of lock dependencies.

Instead of writing out the entire cyclic lock-dependency chain,
we use \emph{deadlock patterns}, i.e.,
the sequence of acquire events that constitute the cycle.
For example, the cyclic lock-dependency chain
$\LD{\thread{1}}{\LKB}{\{\LKA\}}$ and $ \LD{\thread{2}}{\LKA}{\{\LKB\}}$
corresponds to the deadlock pattern~$\DP{e_2,e_6}$.

\subsection{False Positives}

The deadlock-pattern approach suffers from false positives.
%
Consider trace $T_2$ in \cref{fig:ex_3}.
This trace gives rise
to the same lock dependencies as $T_1$ in \cref{fig:ex0}. 
Alas, the resulting deadlock pattern is 
a \emph{false positive},
because it does not correspond to a true deadlock: there is no (correct) reordering of $T_2$ that exhibits the deadlock.
For example, the reordered prefix $[e_1,e_6,e_7]$ is deadlocked on the enabled events $e_2$ and $e_8$ enabled, but it is incorrect:
the read event~$e_6$ no longer observes the same write event~$e_5$ it does in $T_2$ (i.e., the \emph{last write} of $e_6$ is different).
Hence, $e_6$ may read a different value than before, possibly affecting the program flow such that the events $e_7$--$e_{10}$ may never happen.

A similar observation applies to lockset-based data-race prediction methods
where Condition~\cond{DP-Guard} is used to check if two conflicting
memory operations are in a data race.
Two operations are \emph{conflicting} if they refer to the same memory address and at least one of them
is a write.

We conclude that Conditions~\cond{DP-Cycle} and~\cond{DP-Guard} are not sufficient to guarantee a deadlock: additional conditions are necessary to rule out false positives.
A popular method in the area of data-race prediction is to derive a partial order on
events from the program trace.
Ideally, this partial order captures the inter-event dependencies imposed be the trace and its semantics, making it possible to reason about correctly reordered traces without exploring all possible interleavings.
This way, e.g., a partial order may order a read after its last write, or order the events within the same thread.
Unordered events are then considered concurrent: they may appear in different orders among correct reorderings.
Hence, a data race is signaled if two conflicting memory operations are concurrent.
A sufficiently strong partial order then eliminates falsely signaled data races.

\begin{figure}[t]

  \begin{minipage}[b]{.3\textwidth}
    \bda{@{}lcl}
\ba{|l|l|l|l|}
\hline T_3  & \thread{1} & \thread{2} & \thread{3}\\ \hline
\eventE{1}  & \lockE{\LKA}&&\\
\eventE{2}  & \lockE{\LKB}&&\\
\eventE{3}  & \unlockE{\LKB}&&\\
\eventE{4}  & \unlockE{\LKA}&&\\
\eventE{5}  & &\lockE{\LKB}&\\
\eventE{6}  & &\writeE{\VA}&\\
\eventE{7}  & &\unlockE{\LKB}&\\
\eventE{8}  & &&\lockE{\LKB}\\
\eventE{9}  & &&\writeE{\VA}\\
\eventE{10}  & &&\lockE{\LKA}\\
\eventE{11}  & &&\unlockE{\LKA}\\
\eventE{12}  & &&\unlockE{\LKB}\\

 \hline \ea{}

 \eda
     \subcaption{WCP false negative.} 
    \label{fig:ex3b}
  \end{minipage}
  \hfill%
  \begin{minipage}[b]{.55\textwidth}
    \bda{@{}lcl}

    \begin{array}{|l|l|l||l|l|}
      \hline
      T_4 & \thread{1} & \thread{2} & \multicolumn{2}{|l|}{\mbox{Lock deps}} \\
      \hline
      \eventE{1} & \HIGHLIGHT{\lockE{\LKA}} &&& \\
      \eventE{2} & \lockE{\LKB} && D_1 & \LD{\thread{1}}{\LKB}{\{\HIGHLIGHT{\LKA}\}} \\
      \eventE{3} & \unlockE{\LKB} &&& \\
      \eventE{4} & \lockE{\LKC} &&& \LD{\thread{1}}{\LKC}{\{\LKA\}} \\
      \eventE{5} & \lockE{\LKD} && D_2 & \LD{\thread{1}}{\LKD}{\{\HIGHLIGHT{\LKA},l_3\}} \\
      \eventE{6} & \unlockE{\LKD} &&& \\
      \eventE{7} & \unlockE{\LKC} &&& \\
      \eventE{8} & \unlockE{\LKA} &&& \\
      \eventE{9} && \HIGHLIGHTB{\lockE{\LKB}} && \\
      \eventE{10} && \lockE{\LKA} & D_3 & \LD{\thread{2}}{\LKA}{\{\HIGHLIGHTB{\LKB}\}} \\
      \eventE{11} && \unlockE{\LKA} && \\
      \eventE{12} && \lockE{\LKD} && \LD{\thread{2}}{\LKD}{\{\LKB\}} \\
      \eventE{13} && \lockE{\LKC} & D_4 & \LD{\thread{2}}{\LKC}{\{\HIGHLIGHTB{\LKB},l_4\}} \\
      \eventE{14} && \unlockE{\LKC} && \\
      \eventE{15} && \unlockE{\LKD} && \\
      \eventE{16} && \unlockE{\LKB} && \\
      \hline
    \end{array}

    \eda
    \subcaption{Trace where only the first deadlock pattern is feasible.}
    \label{fig:dp-order}
  \end{minipage}%
  \caption{Partial orders for deadlock prediction.}
  \label{f:traceBlock}
\end{figure}

\subsection{Partial-order Methods for Deadlock Prediction
            }

We apply the partial-order idea to the deadlock-prediction setting and refine the definition of a deadlock pattern
$\DP{e_1,\ldots,e_n}$ with the following additional condition:
\begin{description}
  \item[\cond{DP-P}] Events $e_i$ are pairwise concurrent under partial order~P.
\end{description}
Which partial order P to use?
Answering this question is a non-trivial task, shown by
review of a number of existing partial orders that have been applied
in the data-race setting.

\paragraph{Happens-Before is too strict.}

We first consider the Happens-Before relation (HB)~\cite{lamport1978time}.
Under HB, lock releases and acquisitions on the same lock
are ordered as they appear in the trace.
For trace $T_2$ in \cref{fig:ex0}, we
then find that $e_4 \HBLt e_5$.
Since HB is closed under trace order within the same thread (program order) and transitivity,
we find that $e_2 \HBLt e_{6}$: these lock acquires are not concurrent.
Thus, Condition~\cond{DP-HB} eliminates the (false-positive) deadlock pattern $\DP{e_2,e_{6}}$.
However, Condition~\cond{DP-HB} does not hold \emph{for any} deadlock pattern,
because
of the strict textual order among lock acquisitions and releases (critical sections).
We conclude that Condition~\cond{DP-HB} reports no false positives but too many false negatives.
Clearly, this means that DP is not useful in practice: we need a weaker partial order
that does not strictly order critical sections.

\paragraph{Weak-Causally-Precedes is not strict enough.}

The Weak-Causally-Precedes relation (WCP)~\cite{conf/pldi/KiniMV17}
considers two critical sections as unordered, unless they contain conflicting memory operations.
In case of trace $T_1$ in \cref{fig:ex0}, we find that $e_2$ and $e_{6}$ are not ordered.
Thus, Condition~\cond{DP-WCP} (correctly) predicts that $\DP{e_2,e_{6}}$ codifies a deadlock.
Unfortunately, we now encounter false negatives
and still some false positives, as shown next.

Trace $T_2$ in \cref{fig:ex_3} shows that we encounter false positives.
Events $e_5$ and $e_6$ refer to conflicting memory operations, but they are not part of a critical section, so there is no order among the events in threads $\thread{1}$ and $\thread{2}$.
Hence, Condition~\cond{DP-WCP} applies to deadlock pattern $\DP{e_2,e_8}$.
This is a false positive, because any reordering that exhibits the deadlock
will violate the last-write dependency among $e_5$ and $e_6$.

The case of false negative is explained by trace $T_3$ in \cref{fig:ex3b}.
The critical sections for lock $\LKB$ in threads $\thread{2}$ and $\thread{3}$
contain some conflicting write operations on $\VA$, and therefore we find that $e_7 \WCPLt e_9$.
WCP composes (to the left and right) with HB.
Hence, $e_2 \WCPLt e_{10}$ by composition of $e_7 \WCPLt e_9$
with $e_2 \HBLt e_7$ and $e_9 \HBLt e_{10}$.
But then Condition~\cond{DP-WCP} (wrongly) rules out $\DP{e_2,e_{10}}$: a false negative.

We conclude that neither HB nor WCP are suitable for the purpose of deadlock prediction.
Similar observations apply to other partial orders such as SDP~\cite{10.1145/3360605}
and DC~\cite{Roemer:2018:HUS:3296979.3192385-obsolete}.

\paragraph{Total-Read-Write order and deadlock-pattern order for soundness.}

Our idea is to adapt WCP to a new Total-Read-Write partial order (TRW), in the following non-trivial way.
Under TRW, if two critical sections contain TRW-ordered events,
the release of the critical section that appears earlier in the trace is ordered before
the TRW-ordered event in the later critical section.
Though similar, there are two main differences between TRW and WCP.
First, TRW does not compose with HB.
This eliminates false negatives, as seen in trace $T_3$ in \cref{fig:ex3b}.
Second, under TRW, any pair of conflicting memory operations is ordered based on their order in the trace.
This eliminates false positives, as seen in trace $T_4$ in \cref{fig:ex3}.

We can show that the additional Condition~\cond{DP-TRW} guarantees
soundness (no false positives)
if there is no \emph{earlier} deadlock reported.
To explain this issue, consider trace~$T_4$ in \cref{fig:dp-order}.
The deadlock pattern $\DP{e_5,e_{13}}$ does not codify a true deadlock, because there is
an earlier deadlock pattern $\DP{e_2, e_{10}}$.
Hence, there is no correct reordering of trace $T_3$
where events $e_5$ and $e_{13}$ are enabled.

To understand the issue in more detail,
we examine the corresponding cyclic lock-dependency chains in \cref{fig:dp-order}.
The deadlock pattern $\DP{e_2, e_{10}}$ is supported by dependencies~$D_1$ and $D_3$.
The deadlock pattern~$\DP{e_5,e_{13}}$ is supported by $D_2$ and $D_4$.
The lockset of dependency $D_1$ is a subset of the lockset of $D_3$
and the lockset of dependency~$D_2$ is a subset of the lockset of $D_4$.
In both cases, the locks held are acquired by the \emph{same event} as indicated
by the (dark and light gray) highlighting.

Hence, the cyclic lock-dependency chain between $D_2$ and $D_4$ is not reachable,
as it is blocked by the earlier cyclic lock-dependency chain between
$D_1$ and $D_3$.
To rule out such cases, we impose the following condition
that imposes an order among deadlock patterns:
%
\begin{description}
  \item[\cond{DP-Block}] A deadlock pattern is not ordered after any other deadlock patterns.
\end{description}

In \cref{s:sound}, we show that conditions \cond{DP-Guard/-Cycle/-TRW/-Block}
are sufficient to establish soundness (no false positives).

\begin{figure}[t]

  \begin{minipage}[b]{.3\textwidth}
    \bda{@{}lcl}
\ba{|l|l|l|}
\hline  T_5 & \thread{1} & \thread{2}\\ \hline
\eventE{1}  & \lockE{\LKA}&\\
\eventE{2}  & \lockE{\LKB}&\\
\eventE{3}  & \writeE{\VA}&\\
\eventE{4}  & \unlockE{\LKB}&\\
\eventE{5}  & \unlockE{\LKA}&\\
\eventE{6}  & &\lockE{\LKB}\\
\eventE{7}  & &\writeE{\VA}\\
\eventE{8}  & &\lockE{\LKA}\\
\eventE{9}  & &\unlockE{\LKA}\\
\eventE{10}  & &\unlockE{\LKB}\\

 \hline \ea{}

 \eda
 \subcaption{
             TRW false negative.}
    \label{fig:ex3}
  \end{minipage}%
  \hfill%
  \begin{minipage}[b]{.3\textwidth}
    \bda{@{}lcl}
\ba{|l|l|l|}
\hline  T_6 & \thread{1} & \thread{2}\\ \hline
\eventE{1}  & \lockE{\LKA}&\\
\eventE{2}  & \lockE{\LKB}&\\
\eventE{3}  & \unlockE{\LKB}&\\
\eventE{4}  & \unlockE{\LKA}&\\
\eventE{5}  & &\lockE{\LKA}\\
\eventE{6}  & &\unlockE{\LKA}\\
\eventE{7}  & &\lockE{\LKB}\\
\eventE{8}  & &\lockE{\LKA}\\
\eventE{9}  & &\unlockE{\LKA}\\
\eventE{10}  & &\unlockE{\LKB}\\

 \hline \ea{}
 \eda
 \subcaption{\SPDOffline false negative.}
    \label{fig:ex_4}
  \end{minipage}%
  \hfill%
  \begin{minipage}[b]{.3\textwidth}
    \bda{@{}lcl}
\ba{|l|l|l|}
\hline  T_7 & \thread{1} & \thread{2}\\ \hline
\eventE{1}  & &\writeE{\VA}\\
\eventE{2}  & \lockE{\LKB}&\\
\eventE{3}  & \lockE{\LKC}&\\
\eventE{4}  & \readE{\VA}&\\
\eventE{5}  & \lockE{\LKA}&\\
\eventE{6}  & \unlockE{\LKA}&\\
\eventE{7}  & \unlockE{\LKC}&\\
\eventE{8}  & \unlockE{\LKB}&\\
\eventE{9}  & &\lockE{\LKC}\\
\eventE{10}  & &\writeE{\VA}\\
\eventE{11}  & &\unlockE{\LKC}\\
\eventE{12}  & &\lockE{\LKA}\\
\eventE{13}  & &\lockE{\LKB}\\
\eventE{14}  & &\unlockE{\LKB}\\
\eventE{15}  & &\unlockE{\LKA}\\

 \hline \ea{}
 \eda
 \subcaption{PWR false positive.}
    \label{fig:ex_5}
  \end{minipage}%
  \caption{TRW versus \SPDOffline versus PWR.}
  \label{f:TRWvsSPDvsPWR}
\end{figure}

\paragraph{Comparison against the state of the art.}

TRW and the ordering among deadlock patterns can be computed efficiently, as our experiments confirm.

We also compare experimentally against
\SPDOffline~\cite{conf/pldi/TuncMPV23}, which represents the state of the art when it comes to efficient and sound
deadlock prediction methods.
For an extensive benchmark suite, we show that
our method covers the same deadlocks as \SPDOffline.

The similar performance and precision (absence of false positives) between \SPDOffline and our method does \emph{not} imply that they are equivalent.
We propose an entirely new direction
to achieving efficient and sound deadlock prediction.
Generally speaking, \SPDOffline and our approach are incomparable when it comes
to precision and underlying methods, as we explain
via the following examples.

Consider $T_5$ in \cref{fig:ex3}.
Deadlock pattern $\DP{e_2,e_8}$ is a true positive and reported by \SPDOffline.
However, under \cond{DP-TRW} deadlock pattern $\DP{e_2,e_8}$ is ruled out
for the following reason.
We find that $e_3 \TRWLt e_7$, because these events refer to conflicting memory operations.
Both events are protected by a common lock and therefore ${e_5 \TRWLt e_7}$.
By program order and transitivity, we conclude that $e_2 \TRWLt e_8$.

On the other hand,
\SPDOffline fails to report the deadlock pattern $\DP{e_2,e_7}$ in \cref{fig:ex_4},
whereas our method reports this true positive:
$e_2$ and $e_7$ are not ordered under TRW.
To understand these differences, we take a closer look at \SPDOffline.

\SPDOffline only considers \emph{sync-preserving} deadlocks.
A deadlock is sync preserving, if the reordered trace
that exhibits the deadlock does not reorder acquires on the same lock.
In case of \cref{fig:ex_4},
any reordering of $T_6$ that exhibits the deadlock represented by~$\DP{e_2,e_7}$ will have to reorder
the two acquires on $\LKB$ (events $e_1$ and $e_5$).
Hence, \SPDOffline rejects the deadlock pattern.

\SPDOffline discovers this non-sync-preserving reordering
via a closure construction that starts with
a deadlock pattern satisfying Conditions~\cond{DP-Cycle/-Guard}.
Initially, the closure starts with~$\DP{e_2,e_7}$.
\SPDOffline adds event $e_1$, because $e_1$ \emph{must happen before}~$e_2$ codified in a Must-Happen-Before partial order (MHB).
Based on MHB,
we also add events $e_5$ and $e_6$ to the closure.
Once \SPDOffline encounters $e_1$ and $e_5$ (two acquires on the same lock),
event $e_4$ (the matching release of~$e_1$) is added, ensuring
that the closure is sync preserving.
Via $e_4$, we reach our starting point~$e_2$.
\SPDOffline reports a cycle in the closure construction
and therefore rejects deadlock pattern $\DP{e_2,e_7}$.

Hence, there are connections between \SPDOffline and our method:
both make use of information from partial orders.
The major difference is that we
use TRW as part of
Condition~\cond{DP-TRW} to instantly reject deadlock
patterns.
In contrast, \SPDOffline
employs MHB as part of the closure construction.
This then leads to incomparable precision results,
as shown by the traces in \cref{fig:ex3,fig:ex_4}.
A similar observation applies in
comparison against SeqCheck~\cite{conf/fse/CaiYWQP21}.

\paragraph{Experimenting with other partial orders.}

TRW is carefully crafted to ensure that all false positives
are eliminated but some false negatives remain.
What if we weaken TRW, by only
ordering conflicting memory operations that are part
of a last-write dependency?
The resulting partial order is known as PWR~\cite{conf/mplr/SulzmannS20}.
PWR catches the false positive in \cref{fig:ex_3},
and correctly identifies the deadlocks
in \cref{fig:ex3,fig:ex_4}.
In \cref{s:complete}, we show that
Conditions~\cond{DP-Guard/-Cycle/-PWR/-Block}
are sufficient to establish completeness (no false negatives).

In general, PWR has false positives,
as shown by trace $T_7$ in \cref{fig:ex_5}.
Deadlock pattern $\DP{e_5,e_{13}}$ is valid under Condition~\cond{DP-PWR}.
But any reordering that exhibits the deadlock
would need to put events $[e_9,e_{10},e_{11}]$ before
event $e_3$.
This will violate the last-write dependency
between $e_1$ and~$e_4$.
For comparison, under TRW we find that $e_5 \TRWLt e_{13}$,
because all conflicting events are ordered as they appear
in the trace.

We do not encounter any such false positives in our extensive benchmark suite.
In fact, TRW and PWR report the exact same deadlocks.
This shows that TWR and PWR are excellent candidates
when it comes to sound and complete deadlock prediction.



\section{Preliminaries}
\label{sec:prelim}

We introduce executions of concurrent programs with shared variables and locks, where operations are recorded as \emph{traces} of concurrency \emph{events} (\cref{s:prelim:trace}).
Then, we define the semantics of traces in terms of \emph{well formedness}, and how traces may be \emph{reordered} without affecting the semantics of the source trace (\cref{s:prelim:wfCrp}).
Finally, we define \emph{deadlock patterns} and what it means for a trace to contain \emph{resource deadlocks} (\cref{s:prelim:deadlock})

\subsection{Events and traces}
\label{s:prelim:trace}

We use $l_1,l_2,\ldots$ and $x,y,z$ for lock and shared variables, respectively.

\begin{definition}[Events and Traces]
  \label{def:run-time-traces-events}
  \bda{@{}l@{}c@{}ll@{\hskip3em}l@{}c@{}llr@{}}
    \alpha, \beta,\delta & {} ::= {} & 1 \mid 2 \mid \ldots & \mbox{(unique event ids)}
    & e & {} ::= {} & (\alpha, t,op) & \mbox{(events)}
    \\ t,s,u & {} ::= {} & \THD1 \mid \THD2 \mid \ldots & \mbox{(thread ids)}
    & T & {} ::= {} & [] \mid e : T   & \mbox{(traces)}
    \\ op & {} ::= {} & \readE{x} \mid \writeE{x}
    \mid \reqE{l} \mid \acqE{l} \mid \relE{l}
    \span\span\span\span
    & \mbox{(operations)}
  \eda
\end{definition}

\noindent
A trace $T$ is a list of events reflecting a single execution of a
concurrent program under the sequential consistency memory
model~\cite{Adve:1996:SMC:619013.620590}.
We write $[o_1,\dots,o_n]$
for a list of objects  
and use the operator ``$\conc$'' for list concatenation.

An event $e$ is represented by a triple $(\alpha, t,op)$
where
$\alpha$ is a unique event identifier,
$op$ is an operation, and
$t$ is the thread id in which the operation took place.
The main thread has thread id~$\THD1$.
The unique event identifier unambiguously identifies events under trace reordering (cf.\ \cref{s:prelim:wfCrp}).

The operations $\readE{x}$ and $\writeE{x}$ denote reading of and writing to a shared variable~$x$, respectively.
Operations $\acqE{l}$ and $\relE{l}$ denote acquiring and releasing a lock $l$, respectively.
We also include lock requests $\reqE{l}$ to denote the (possibly unfulfillable) attempt at acquiring a lock $l$; this allows for smoother definitions and is crucial for completeness (cf.\ \cref{s:complete}).
Operations to create threads (fork) and synchronize with termination of threads (join) can be modeled using shared variables and locks.

We often omit identifier and/or thread when denoting events, and omit parentheses if only the operation remains, writing $e = (t,op)$ or $e = op$ instead of $e = (\alpha,t,op)$.
Notation $\thd(e)$ extracts the thread id from an event.

We write $e \in T$ to indicate that $T = [e_1, \dots, e_n]$ and $e = e_k$ for $1 \le k \le n$, defining also $\posP{T}{e} = k$.
The set of events in a trace is then $\evts(T) = \{ e \mid e \in T \}$, and the set of thread ids in a trace is $\thds(T) = \{ \thd(e) \mid e \in T \}$.
For trace $T$ and events $e, f \in \evts(T)$, we define trace order: $e \TrLt[T] f$ if $\posP{T}{e} < \posP{T}{f}$.
We then define program order $\POLt[T]$ as the restriction of $\TrLt[T]$ to events with the same thread id.

We often express traces in tabular notation.
Here, traces have one column per thread.
Each event is placed in row of its trace position, and in the column of its thread.
\Cref{f:firstExamples,f:traceBlock,f:TRWvsSPDvsPWR} contain numerous examples.

\subsection{Well Formedness and Correct Reorderings}
\label{s:prelim:wfCrp}

Traces must be well formed, following the standard sequential consistency conditions for concurrent objects \cite{Adve:1996:SMC:619013.620590,Said:2011:GDR:1986308.1986334,Huang:2014:MSP:2666356.2594315}.

\begin{definition}[Well Formedness]
\label{d:WF}
  A trace $T$ is \emph{well formed} if all the following conditions are satisfied:
  \begin{description}
    \item[\cond{WF-Acq}]
      For every $\evtAcc = (t,\acqE{l}), \evtAccA = (s,\acqE{l}) \in T$ where $\evtAcc \TrLt[T] \evtAccA$, there exists $\evtRel = (t,\relE{l}) \in T$ such that $\evtAcc \TrLt[T] \evtAccA$.
      We say that $a$ and $r$ \emph{match}.

   \item[\cond{WF-Rel}]
      For every $\evtRel = (t,\relE{l}) \in T$, there exists $\evtAcc = (t,\acqE{l}) \in T$ such that $a \TrLt[T] r$ and there is no $\evtRelA = (s,\relE{l}) \in T$ with $\evtAcc \TrLt[T] \evtRelA \TrLt[T] \evtRel$.
      We say that $a$ and $r$ \emph{match}.

    \item[\cond{WF-Req}]\label{cond:lock-3}
      For every $\ACQ{} = (t,\acqE{l}) \in T$, there exists $\REQ{} = (t,\reqE{l}) \in T$ such that $\REQ{} \TrLt[T] \ACQ{}$, and for every $\REQ{} = (t,\reqE{l}) \in T$ and $e = (t,op) \in T$ such that $\REQ{} \TrLt[T] e$, if there is no $f = (t,op')$ such that $\REQ{} \TrLt[T] f \TrLt[T] e$ then $e = (t,\acqE{l})$.
      We say that $\REQ{}$ \emph{requests} $e$.




  \end{description}
\end{definition}

\noindent
Condition~\cond{WF-Acq} states that a previously acquired lock can only be acquired after it has been released.
Similarly, Condition~\cond{WF-Rel} states that a lock can only be released after is has been acquired but not yet released.
Note that these conditions require matching acquires and releases to occur in the same thread.
Condition~\cond{WF-Req} states that all lock acquires must have been requested in the same thread immediately before.
We often omit requests from example traces, but we assume that they precede each acquire implicitly.
At the end of the trace, acquired locks do not have to be released and requests do not have to be fulfilled.
All traces in \cref{f:firstExamples,f:traceBlock,f:TRWvsSPDvsPWR} are well formed, with implicit requests; we discuss an example with explicit requests in the next subsection.



A trace represents only one possible interleaving of events: there can be as many interleavings as there are permutations of the original trace.
However, not all such reorderings are feasible in the sense of being reproducible by executing the program with a different schedule.
In addition to well formedness, a reordering must guarantee that (a)~program order and (b)~last writes are maintained.
A reordering maintains program order if, in any thread, the order of events is unchanged.
For guarantee~(b), note that every read observes some write on the same variable: the last preceding such write.
Last writes are then maintained if the write observed by any read is unchanged.
Guarantee~(b) is particularly important to ensure that the control flow of the traced program is unaffected by the reordering, e.g., when the read is used in a conditional statement.

Reorderings do not have to run to completion, so we consider reordered \emph{prefixes} of traces.

\begin{definition}[Correctly Reordered Prefix]
  \label{def:crp}
  The \emph{projection} of $T$ onto thread $t$, denoted $\proj(t,T)$, restricts $T$ to events $e$ with $\thd(e) = t$.
  That is,
  $\evtAA \in proj(t,T)$ if and only if $e \in T$ and $\thd(e) = t$, and
  $\evtAA \TrLt[proj(t,T)] \evtBB$ implies $\evtAA \TrLt[T] \evtBB$.

  Take $\evtAA = \readE{x},\evtBB = \writeE{x} \in T$.
  We say that $\evtBB$ is the \emph{last write} of $\evtAA$ w.r.t.\ $T$ if $\evtBB$ appears before $\evtAA$ with no other write on $x$ in between.
  That is,
  $\evtBB \TrLt[T] \evtAA$, and
  there is no $\evtCC = \writeE{x} \in T$ such that $\evtBB \TrLt[T] \evtCC \TrLt[T] \evtAA$.

  Trace $T'$ is a \emph{correctly reordered prefix} of $T$ if the following
  conditions are satisfied:
  \begin{description}
    \item[\cond{CRP-WF}] $T'$ is well formed and $\evts(T') \subseteq \evts(T)$.
    \item[\cond{CRP-PO}] For every $t \in \thds(T')$, $\proj(t,T')$ prefixes $\proj(t,T)$.
    \item[\cond{CRP-LW}] For every $\readE{x} \in T'$ with last write $f$ w.r.t.\ $T$, it must be $f$ is also the last write for $e$ w.r.t.\ $T'$.\footnote{Unique event identifiers are crucial for this condition.}
  \end{description}
  We write $\crp(T)$ to denote the set of correctly reordered prefixes of $T$.
\end{definition}

\subsection{Resource Deadlocks}
\label{s:prelim:deadlock}

We consider a standard definition of (resource) deadlock, where multiple threads are simultaneously requesting to acquire a lock, but every such lock is held by another thread: none of the requests can be fulfilled, and so the trace is stuck.
Our goal is to predict that a trace can be rescheduled to deadlock, and so we identify traces with \emph{predictable} deadlocks in terms of correctly reordered prefixes with deadlocks.
We first identify these deadlocks using \emph{deadlock patterns}.

Deadlock patterns codify cycles of requests that cannot be fulfilled because the requested lock is held by the next request.
Hence, we first define what it means for an event to hold a lock, in terms of \emph{critical sections} of events bordered by lock acquires and releases.

\begin{definition}[Critical Sections]
  \label{d:CS}
  Given $a = (t,\acqE{l}) \in T$, we define the \emph{critical section acquired by $a$} as $\CS[T](a) \subseteq \evts(T)$, where $e \in \CS[T](a)$ if all the following conditions hold:
  \begin{description}
    \item[\cond{CS-PO}]
      $e = a$ or $a \POLt[T] e$.
    \item[\cond{CS-Held}]
      There is no $r = \relE{l} \in T$ that matches $a$ with $a \POLt[T] r \POLt[T] e$.
  \end{description}

  The set of \emph{acquires held} by $e \in T$ is defined as $\AH[T](e) = \{ a \in T \mid e \in \CS[T](a) \}$.
  Accordingly, the set of \emph{locks held} is defined as $\LH[T](e) = \{ l \mid \exists \acqE{l} \in \AH[T](e) \}$.
\end{definition}

\noindent
Thus, an event is in a critical section if it is or appears after the acquire in the same thread (Condition~\cond{CS-PO}), as long as the critical section has not yet been closed (Condition~\cond{CS-Held}).
Note that Condition~\cond{CS-Held} considers a lock held if it is never released at the end of the trace.

\begin{definition}[Deadlock Patterns (DPs)]
  \label{d:dp}
  Given trace $T$, let $A = \DP{(\ACQ1,\REQ1),\ldots,(\ACQ{n},\REQ{n})} \subseteq \evts(T) \times \evts(T)$ for $n \ge 2$, where every $(\ACQ{i},\REQ{i}) \in A$ is an acquire-request tuple, and $\thd(\REQ{i}) \ne \thd(\REQ{j})$ for every $1 \le i < j \le n$.
  We define the following properties for $A$:
  \begin{description}
    \item[\cond{DP-Cycle}]
      For every $1 \le i \le n$, $q_i = \reqE{l} \in \CS[T](a_i)$ and $a_{(i \mod n) + 1} = \acqE{l}$.
    \item[\cond{DP-Guard}]
      $\LH[T](\ACQ{i}) \cap \LH[T](\ACQ{j}) = \emptyset$ for $1 \leq i < j \leq n$.
  \end{description}
  Condition~\cond{DP-Cycle} identifies $A$ as a \emph{cycle in~$T$}.
  If $A$ is a cycle in~$T$ that satisfies Condition~\cond{DP-Guard}, then $A$ is a \emph{deadlock pattern (DP)} in~$T$.
\end{definition}

\noindent
Thus, we consider sets of acquire-request pairs, where the request holds the acquire.
Condition~\cond{DP-Cycle} ensures that the set forms a cycle of requests for locks held by other requests.
DPs satisfy an additional condition to eliminate false alarms: Condition~\cond{DP-Guard} ensures that no two requests hold a common lock, which would make it impossible to schedule the deadlock.

Finally, we define predictable deadlocks in terms of DPs.

\begin{definition}[Predictable Deadlock]
  \label{d:deadlock}
  Let $A$ be a DP in well-formed trace $T$.
  We say $A$ is a \emph{predictable deadlock} if there exists $T' \in \crp(T)$ such that, for every $(\ACQ{},\REQ{}) \in A$, $\REQ{} \in T'$ and there is no $e \in T'$ with $\REQ{} \POLt[T'] f$.
\end{definition}

\begin{figure}[t]
  \begin{minipage}[b]{.32\textwidth}
    \bda{@{}c@{}}
      \begin{array}[b]{|l|l|l|}
        \hline
        T'_1 & \thread{1} & \thread{2} \\
        \hline
        \eventE{1} & \lockE{\LKA} & \\
        \eventE{2} & \reqE{\LKB} & \\
        \eventE{3} & \lockE{\LKB} & \\
        \eventE{4} & \unlockE{\LKB} & \\
        \eventE{5} & \unlockE{\LKA} & \\
        \eventE{6} && \lockE{\LKB} \\
        \eventE{7} && \reqE{\LKA} \\
        \eventE{8} && \lockE{\LKA} \\
        \eventE{9} && \unlockE{\LKA} \\
        \eventE{10} && \unlockE{\LKB} \\
        \hline
      \end{array}
      \\ \\
      \begin{array}[b]{|l|l|l|}
        \hline
        T''_1 & \thread{1} & \thread{2} \\
        \hline
        \eventE{1} & \lockE{\LKA} & \\
        \eventE{6} && \lockE{\LKB} \\
        \eventE{2} & \reqE{\LKB} & \\
        \eventE{7} && \reqE{\LKA} \\
        \hline
      \end{array}
    \eda
  \end{minipage}%
  \hfill%
  \begin{minipage}[b]{.32\textwidth}
    \bda{@{}c@{}}
      \ba{|l|l|l|}
        \hline T'_2  & \thread{1} & \thread{2} \\ \hline
        \eventE{1}  & \lockE{\LKA} & \\
        \eventE{2}  & \reqE{\LKB} & \\
        \eventE{3}  & \lockE{\LKB} & \\
        \eventE{4}  & \unlockE{\LKB} & \\
        \eventE{5}  & \unlockE{\LKA} & \\
        \eventE{6}  & \writeE{\VA} & \\
        \eventE{7}  & &\readE{\VA} \\
        \eventE{8}  & &\lockE{\LKB} \\
        \eventE{9}  & &\reqE{\LKA} \\
        \eventE{10}  & &\lockE{\LKA} \\
        \eventE{11}  & &\unlockE{\LKA} \\
        \eventE{12}  & &\unlockE{\LKB} \\
        \hline
      \ea{}
      \\ \\
      \ba{|l|l|l|}
        \hline T''_2  & \thread{1} & \thread{2} \\ \hline
        \eventE{1}  & \lockE{\LKA} & \\
        \eventE{2}  & \writeE{\VA} & \\
        \eventE{3}  & \reqE{\LKB} & \\
        \eventE{4}  & \lockE{\LKB} & \\
        \eventE{5}  & \unlockE{\LKB} & \\
        \eventE{6}  & \unlockE{\LKA} & \\
        \eventE{7}  & &\lockE{\LKB} \\
        \eventE{8}  & &\reqE{\LKA} \\
        \eventE{9}  & &\lockE{\LKA} \\
        \eventE{10}  & &\readE{\VA} \\
        \eventE{11}  & &\unlockE{\LKA} \\
        \eventE{12}  & &\unlockE{\LKB} \\
        \hline
      \ea{}
    \eda
  \end{minipage}%
  \hfill%
  \begin{minipage}[b]{.32\textwidth}
    \bda{@{}c@{}}
      \begin{array}{|l|l|l|}
        \hline
        T'_4 & \thread{1} & \thread{2} \\
        \hline
        \eventE{1} & \lockE{\LKA} & \\
        \eventE{2} & \reqE{\LKB} & \\
        \eventE{3} & \lockE{\LKB} & \\
        \eventE{4} & \unlockE{\LKB} & \\
        \eventE{5} & \lockE{\LKC} & \\
        \eventE{6} & \reqE{\LKD} & \\
        \eventE{7} & \lockE{\LKD} & \\
        \eventE{8} & \unlockE{\LKD} & \\
        \eventE{9} & \unlockE{\LKC} & \\
        \eventE{10} & \unlockE{\LKA} & \\
        \eventE{11} && \lockE{\LKB} \\
        \eventE{12} && \reqE{\LKA} \\
        \eventE{13} && \lockE{\LKA} \\
        \eventE{14} && \unlockE{\LKA} \\
        \eventE{15} && \lockE{\LKD} \\
        \eventE{16} && \reqE{\LKC} \\
        \eventE{17} && \lockE{\LKC} \\
        \eventE{18} && \unlockE{\LKC} \\
        \eventE{19} && \unlockE{\LKD} \\
        \eventE{20} && \unlockE{\LKB} \\
        \hline
      \end{array}
    \eda
  \end{minipage}%
  \caption{Example traces, with selected requests included (cf.\ \cref{f:firstExamples,f:traceBlock}).}
  \label{f:firstWithReqs}
\end{figure}

\noindent
Note that, for simplicity, \cref{sec:overview} denoted DPs as only the set of acquires requested by the unfulfillable requests.
Compare \cref{f:firstExamples,f:firstWithReqs}, where the latter explicitly includes selected requests.
DP $\{(e_1,e_2),(e_7,e_6)\}$ in trace $T'_1$ is a predictable deadlock, witnessed by $T''_1 \in \crp(T'_1)$.
On the other hand, DP $\{(e_2,e_1),(e_9,e_8)\}$ in trace $T'_2$ is not a predictable deadlock, because all reorderings in which $e_2$ is the last event in $\thread1$ omit $e_6$, which is the last write of $e_7$ and thus violates Condition~\cond{CRP-LW}.


\section{Soundness: Partial Orders to Eliminate False Positives}
\label{s:sound}

As discussed in \cref{sec:overview}, DPs as defined in \cref{s:prelim:deadlock} alone are imprecise: they often signal false alarms.
To be precise, we identified two situations in which DPs falsely identify deadlocks.
To rule out such situations, we introduce two new conditions on DPs that use partial orders, eliminating all false positives and allowing us to prove soundness (\cref{thm:TRWSound}).

\subsection{Partially-ordered Requests}
\label{s:sound:trw}

A DP is a predictable deadlock in $T$ when there is a $T' \in \crp(T)$ in which its requests are the last events in their respective threads.
Because two requests never directly depend on each other, this implies that the requests can occur in $T'$ in any order, i.e., they are concurrent.
Hence, the first situation occurs when some requests are not concurrent, making it impossible to find such a $T'$.
We rule out this situation by giving a partial order that is appropriate for determining the concurrency of the requests in a DP.

It turns out that our novel Total Read-Write order (TRW), as defined below, is up to the task (cf.\ \cref{s:sound:sound}).
TRW is derived from PWR~\cite{conf/mplr/SulzmannS20} by imposing that all reads and writes on the same variable are ordered as in the source trace.
Moreover, TRW is similar to WCP~\cite{conf/pldi/KiniMV17}, which soundly predicts deadlocks in absence of data races but is not sound in general (cf.\ \cref{sec:related} for details).

\begin{definition}[Total Read-Write Order (TRW)]
  \label{d:trw}
  We say events $e$ and $f$ are \emph{conflicting}, denoted~$e \conf f$, if $e$ and $f$ are reads/writes on the same variable and at least one is a write.

  Given trace $T$, we define the \emph{Total-Read-Write order (TRW)} as a relation $\TRWLt[T]$ on $\evts(T)$, where $e \TRWLt[T] f$ if either of the following conditions holds:
  \begin{description}
    \item[\cond{TRW-PO}]
      $e \POLt[T] f$.
    \item[\cond{TRW-Conf}]
      $e \conf f$ and $e \TrLt[T] f$.
    \item[\cond{TRW-Rel}]
      There are acquires $a_1,a_2 \in T$ on the same lock such that (i)~$e$ is the release matching $a_1$, (ii)~$a_1 \TrLt[T] a_2$, (iii)~$f \in \CS[T](a_2)$, and (iv)~$a_1 \TRWLt[T] f$.
    \item[\cond{TRW-Tr}]
      There exists $g \in T$ such that $e \TRWLt[T] g \TRWLt[T] f$.
  \end{description}
  We define $\TRWcrp(T) = \{ T' \in crp(T) \mid \forall f \in T' , e \in T \colon e \TRWLt[T] f \implies e \TrLt[T'] f \}$.

  A DP $A$ in well-formed trace $T$ satisfies Condition~\cond{DP-TRW} if neither $\REQ{} \TRWLt[T] \REQ{}'$ nor $\REQ{}' \TRWLt[T] \REQ{}$ (denoted $\REQ{} \TRWConc[T] \REQ{}'$) for every distinct $(\ACQ{},\REQ{}),(\ACQ{}',\REQ{}') \in A$.
\end{definition}

\noindent
Crucially, Condition~\cond{TRW-Conf} orders all reads/writes on the same variable by source-trace order;
note that the reads/writes have to be \emph{conflicting}, so one event must be a write meaning that reads are not ordered directly.
Condition~\cond{TRW-Rel} orders (the events within) critical sections on the same lock, if they contain TRW-ordered events;
note that events that precede the later TRW-ordered event~$f$ are not affected, such that $f$ correctly does not affect ordering in correctly ordered prefixes from which $f$ is omitted (we revisit this point in context of completeness in \cref{s:complete}).
Conditions~\cond{TRW-PO/-Tr} close TRW under program order and transitivity, respectively.

To illustrate, we reconsider trace~$T_2$ in \cref{f:firstExamples}, discussed in \cref{sec:overview} to contain a false DP because its requests are ordered.
We consider variant~$T'_2$ in \cref{f:firstWithReqs}, which explicitly includes these requests.
This trace contains DP $\{(e_1,e_2),(e_8,e_9)\}$, but ${e_2 \not\TRWConc[T'_2] e_9}$: we have $e_2 \TRWLt[T'_2] e_6$ (\cond{TRW-PO/-Tr}), $e_6 \TRWLt[T'_2] e_7$ (\cond{TRW-Conf}), and $e_7 \TRWLt[T'_2] e_9$ (\cond{TRW-PO/-Tr}), so $e_2 \TRWLt[T'_2] e_9$ (\cond{TRW-Tr}).

On the other hand, in trace~$T''_2$ in \cref{f:firstWithReqs}, a slight modification of $T'_2$ by moving the write and read events, the DP $\{(e_1,e_3),(e_7,e_8)\}$ does satisfy Condition~\cond{DP-TRW}: $e_2 \TRWLt[T''_2] e_{10}$ (\cond{TRW-Conf}) and hence $e_6 \TRWLt[T''_2] e_{11}$ (\cond{TRW-Rel}), but although $e_3 \TRWLt[T''_2] e_6$ and $e_8 \TRWLt[T''_2] e_{11}$ we have $e_3 \TRWConc[T''_2] e_8$ because their order does not matter in correctly reordered prefixes from which $e_{10}$ is omitted.

Note that the experiments in \cref{sec:experiments} show that soundness under Condition~\cond{DP-TRW} is not trivial, as it does not eliminate any of the predictable deadlocks that occur in the benchmarks.

\subsection{Partially-ordered Deadlock Patterns}

In the second situation, the candidate deadlock cannot be reached because of an ``earlier'' deadlock.
We make this formal by means of a partial order on cycles (sets of acquire-request pairs that satisfy Condition~\cond{DP-Cycle}), and rule out false positives by reporting only the earliest DPs.

\begin{definition}
  \label{def:ltDP}
  Given trace $T$, we define \emph{cycle order} as a relation $\DPLt[T]$ on cycles in $T$, where $A \DPLt[T] B$ if, for every $(\ACQ{},\REQ{}) \in A$, there exists $(\ACQ{}',\REQ{}') \in B$ with $q \TrLt[T] q'$ and $a \in \AH[T](\REQ{}')$.

  A DP $A$ in well-formed trace $T$ satisfies Condition~\cond{DP-Block} if there is no cycle $B$ in $T$ such that $B \DPLt[T] A$.
\end{definition}

\noindent
The ordering crucially requires that every acquire of a request in the earlier cycle is held by a request in the later cycle.
This way, we ensure that the earlier cycle is unavoidable when trying to achieve the later.

To illustrate, we reconsider trace~$T_4$ in \cref{f:traceBlock}, discussed in \cref{sec:overview} to contain a false DP because of an earlier deadlock.
We consider variant $T'_4$ in \cref{f:firstWithReqs}, which explicitly includes the relevant requests.
Let $A = \{(e_1,e_2),(e_{11},e_{12})\}$ and $B = \{(e_5,e_6),(e_{15},e_{16})\}$.
Both $A$ and $B$ are DPs satisfying Condition~\cond{DP-TRW}, but $B$ is a false positive.
Indeed, $e_1 \in \AH[T'_4](e_6)$ and $e_2 \TrLt[T'_4] e_6$, and $e_{11} \in \AH[T'_4] e_{16}$ and $e_{12} \TrLt[T'_4] e_{16}$.
Hence, $A \DPLt[T'_4] B$, so $B$ does not satisfy Condition~\cond{DP-Block}.

Note that deadlocks involving $n$ threads can only be blocked by prior deadlocks involving $m \leq n$ threads; if the prior deadlock involves $m' > n$ threads, at least one of the later requests would have to be in two of the earlier critical sections, which is impossible wy well formedness.
This is reflected by the forall-exists structure of the above ordering on DPs.

\subsection{Soundness}
\label{s:sound:sound}

\begin{figure}[t]
  \begin{minipage}[b]{.41\textwidth}
    \bda{c}
      \begin{array}[t]{|l|l|l|l|}
        \hline
        T_8 & \thread{1} & \thread{2} & \thread{3} \\
        \hline
        \eventE{1} & \acqE{l_3} && \\
        \eventE{2} & \writeE{x} && \\
        \eventE{3} && \readE{x} & \\
        \eventE{4} && \acqE{l_1} & \\
        \eventE{5} && \reqE{l_2} & \\
        \eventE{6} && \acqE{l_2} & \\
        \eventE{7} && \relE{l_2} & \\
        \eventE{8} && \relE{l_1} & \\
        \eventE{9} && \writeE{x} & \\
        \eventE{10} & \readE{x} && \\
        \eventE{11} & \relE{l_3} && \\
        \eventE{12} &&& \acqE{l_3} \\
        \eventE{13} &&& \acqE{l_2} \\
        \eventE{14} &&& \reqE{l_1} \\
        \eventE{15} &&& \acqE{l_1} \\
        \eventE{16} &&& \relE{l_1} \\
        \eventE{17} &&& \relE{l_2} \\
        \eventE{18} &&& \relE{l_3} \\
        \hline
      \end{array}
    \eda
    \subcaption{Trace with unbounded critical section.}
    \label{fig:guarded}
  \end{minipage}%
  \hfill%
  \begin{minipage}[b]{.56\textwidth}
    \bda{c}
      \ba{@{}ll@{}}
        \begin{array}[t]{|l|l|l|}
          \hline
          T_9 & \thread{1} & \thread{2} \\
          \hline
          \eventE{1} & \acqE{l_1} & \\
          \eventE{2} & \writeE{x} & \\
          \eventE{3} & \reqE{l_2} & \\
          \eventE{4} & \acqE{l_2} & \\
          \eventE{5} & \relE{l_2} & \\
          \eventE{6} & \relE{l_1} & \\
          \eventE{7} && \acqE{l_2} \\
          \eventE{8} && \readE{x} \\
          \eventE{9} && \reqE{l_1} \\
          \eventE{10} && \acqE{l_1} \\
          \eventE{11} && \relE{l_1} \\
          \eventE{12} && \relE{l_2} \\
          \hline
        \end{array}
        &
        \begin{array}[t]{|l|l|l|}
          \hline
          T'_9 & \thread{1} & \thread{2} \\
          \hline
          \eventE{1} & \acqE{l_1} & \\
          \eventE{2} & \writeE{x} & \\
          \eventE{7} && \acqE{l_2} \\
          \eventE{8} && \readE{x} \\
          \eventE{9} && \reqE{l_1} \\
          \eventE{3} & \reqE{l_2} & \\
          \hline
        \end{array}
      \ea
    \eda

    \subcaption{
      Trace with predictable deadlock.
    }\label{fig:requests}
  \end{minipage}%
  \caption{Traces highlighting technical aspects of soundness and completeness (\cref{s:sound,s:complete}).}
\end{figure}

We prove soundness of DPs that satisfy Conditions~\cond{DP-TRW/-Block}, i.e., those conditions rule out all false positives.
In doing so, we make two technical assumptions:
\begin{itemize}
  \item
    We only consider traces in which locks only protect requests in the same thread.
    That is, we assume that trace $T$ is \emph{TRW bounded}:
    for every request $q \in T$ and acquire $a \in T$ with matching release $r$, if $a \TRWLt[T] q$ and $r \in T$ implies $q \TRWLt[T] r$, then $\thd(q) = \thd(a)$.
    For example, trace~$T_8$ in \cref{fig:guarded} is not TRW bounded, reporting false positive $\{(e_4,e_5),(e_{14},e_{15})\}$.
  \item
    We only consider traces that are \emph{well nested}, meaning that any thread can only release an acquired lock once all further acquired locks have been released.
\end{itemize}
These are common assumptions (cf., e.g., \cite{conf/pldi/KiniMV17}).
It may be possible to lift these assumptions by additional conditions on DPs, but this is well outside the scope of this paper.

\begin{theorem}[Soundness]
  \label{thm:TRWSound}
  If well-formed trace $T$ is TRW~bounded and well nested, and DP $A$ in $T$ satisfies Conditions~\cond{DP-TRW/-Block}, then $A$ is a predictable deadlock.
\end{theorem}

\begin{proofsketch}
  We prove the theorem by showing that there is $T' \in \TRWcrp(T)$ such that every request in $A$ is the last event in its thread.
  Towards contradiction, assume there is not such witness $T'$.
  If there are witnesses in $\crp(T) \setminus \TRWcrp(T)$, we derive that some requests are TRW ordered: Condition~\cond{DP-TRW} is contradicted.
  Otherwise, there are no witnesses at all, and we derive that there must be a cycle $B$ such that $B \DPLt[T] A$: Condition~\cond{DP-Block} is contradiction.
\end{proofsketch}

\section{Completeness: A Partial Order to Eliminate False Negatives}
\label{s:complete}

In the previous section, we were able to leverage the Total Read-Write order (TRW) to non-trivially eliminate all DPs that do not correspond to predictable deadlocks.
On the other hand, TRW is slightly too strong to catch all predictable deadlocks: some deadlocks require reordering some reads and writes, which TRW does not permit.
In other words, a DP that does not satisfy Condition~\cond{DP-TRW} can still be a predictable deadlock.
Although such false negatives rarely occur in practice (cf.\ \cref{sec:experiments}), there is no theoretical guarantee that all predictable deadlocks are caught.
The question is then whether there is a variant of Condition~\cond{DP-TRW} that rules out enough false positives to be practically relevant, while guaranteeing that no predictable deadlock is missed.

\subsection{PWR: Weakening TRW}

Recall that TRW is derived from Program-Write-Release order (PWR)~\cite{conf/mplr/SulzmannS20}.
PWR is subtly weaker than TRW, meaning that it orders events less strictly: TRW orders all conflicting reads and writes, whereas PWR only orders reads and their last writes.
Interestingly, this subtle difference is enough for PWR itself to be a suitable candidate for a complete characterization of predictable deadlocks.
We define PWR and, accordingly, Condition~\cond{DP-PWR}.

\begin{definition}[Program-Write-Release Order (PWR)]
  \label{d:pwr}
  Given trace $T$, we define the Program-Write-Release order (PWR) as a relation $\PWRLt[T]$ on $\evts(T)$, where $e \PWRLt[T] f$ if either of the following conditions hold:
  \begin{description}
    \item[\cond{PWR-PO}]
      $e \POLt[T] f$.
    \item[\cond{PWR-LW}]
      $e$ is the last write of read $f$ w.r.t.\ $T$.
    \item[\cond{PWR-Rel}]
      There are acquires $a_1,a_2 \in T$ on the same lock such that (i)~$e$ is the release matching $a_1$, (ii)~$a_1 \TrLt[T] a_2$, (iii)~$f \in \CS[T](a_2)$, and (iv)~$a_1 \PWRLt[T] f$.
    \item[\cond{PWR-Tr}]
      There exists $g \in T$ such that $e \PWRLt[T] g \PWRLt[T] f$.
  \end{description}
  A DP $A$ in well-formed trace $T$ satisfies Condition~\cond{DP-PWR} if neither $q \PWRLt[T] q'$ nor $q' \PWRLt[T] q$ (denoted $q \PWRConc[T] q'$) for every distinct $(a,q),(a',q') \in A$.
\end{definition}

\noindent
Thus, PWR is defined exactly as TRW (\cref{d:trw}), except for Condition~\cond{PWR-LW}.

As mentioned before in \cref{s:sound:trw}, the design of Condition~\cond{TRW-/PWR-Rel} is essential for completeness.
Recall that in trace~$T''_2$ in \cref{f:firstWithReqs}, DP $\{(e_1,e_3),(e_7,e_8)\}$ is a predictable deadlock.
Indeed, Condition~\cond{DP-PWR} is satisfied: $e_3 \PWRConc[T''_2] e_8$.
On the other hand, suppose we simplified Condition~\cond{PWR-Rel} to order $e \PWRLt a_2$ instead of only $e \PWRLt f$.
Then $e_3 \PWRLt[T''_2] e_5 \PWRLt[T''_2] e_7 \PWRLt[T''_2] e_8$, leading to an unnecessary false negative.

\subsection{Lock Requests}

As hinted at in \cref{sec:prelim}, using lock requests not only leads to an elegant description of predictable deadlock and deadlock pattern: it is essential for completeness.

To illustrate, suppose we were to define deadlock patterns in terms of lock acquires, and predictable deadlocks by correctly reordered prefixes in which the acquires of a DP are the next event in their thread.
Consider trace~$T_9$ in \cref{fig:requests}.
As witnessed by trace~${T'_9 \in \crp(T_9)}$, DP $\{(e_1,e_4),(e_7,e_{10})\}$ is a predictable deadlock.
However, we have
(i)~$e_2 \PWRLt[T_9] e_8$ (\cond{PWR-LW}),
(ii)~$e_1 \PWRLt[T_9] e_{10}$ ((i), \cond{PWR-PO/-Tr}),
(iii)~$e_6 \PWRLt[T_9] e_{10}$ ((ii), $e_{10} \in \CS[T_9](e_{10})$, \cond{PWR-Rel/-Tr}),
and (iv)~$e_4 \PWRLt[T_9] e_{10}$ ((iii), \cond{PWR-PO/-Tr}):
Condition~\cond{DP-PWR} does not hold, so the DP is falsely rejected.
On the other hand, using requests, we have $e_3 \PWRConc[T_9] e_9$, correctly satisfying Condition~\cond{DP-PWR}.

\subsection{Completeness}

Note that our completeness result does not require a trace to be well nested or bounded, as our soundness result does.

\begin{theorem}[Completeness]
  \label{t:PWRcomplete}
  If DP $A$ in well-formed trace $T$ is a predictable deadlock, then $A$ satisfies Conditions~\cond{DP-PWR/-Block}.
\end{theorem}

\begin{proofsketch}
  By assuming that $A$ is a predictable deadlock, there is $T' \in \crp(T)$ such that every $\REQ{i}$ in $A$ is the last event in its thread.
  We prove both conditions by contradiction.
  Condition~\cond{DP-PWR} is straightforward: requests at the end of threads can be trivially reordered, contradicting any PWR ordering among them.
  Condition~\cond{DP-Block} follows by showing that a preceding cycle would lead to a cyclic trace order in~$T'$.
\end{proofsketch}

Our experiments (\cref{sec:experiments}) show that the result above is not trivial, as PWR does not report any false positives, even though it is theoretically unsound.
As such, our methodology allows a choice between near-sound completeness (PWR) and near-complete soundness~(TRW).
Ideally, we will develop a partial order that is both sound and complete, lying somewhere between TRW and PWR, but this seems to be a quest for a ``holy grail''.


\newcommand{\setalglineno}[1]{%
  \setcounter{ALG@line}{\numexpr#1-1}}

\newcommand{\all}{{\mathit all}}

\newcommand{\AllLocksSym}{\LocksSym{\all}}
\newcommand{\acqVC}[1]{{\mathit Acq(#1)}}

\newcommand{\RW}[1]{\mathit{RC}(#1)}

\newcommand{\pwrVC}[1]{\mathit{PWR}(#1)}


\newcommand{\pwrCand}{PWR}
\newcommand{\pwrCandGuarded}{PWR$^{cs}$}
\newcommand{\pwrCandGuardedRef}{PWR$^{rcs}$}
\newcommand{\undeadCand}{UD}
\newcommand{\undeadCandGuardedExtra}{UD$^{cs}$}
\newcommand{\undeadCandGuardedExtraRef}{UD$^{rcs}$}

\newcommand{\OK}{\mathit{ok}}
\newcommand{\FAIL}{\mathit{fail}}

\section{Implementation}
\label{sec:implementation}

We show how we implemented the deadlock-prediction methods
from \cref{s:sound,s:complete} as an algorithm.
Our algorithm takes as input a trace and yields a set of deadlock patterns, in two phases:
(1)~computation of the partial ordering of events and lock dependencies (\cref{sec:vcAndLd}),
followed by
(2)~detection of cyclic lock-dependency chains representing deadlock patterns (\cref{sec:impl:verify}).
The partial ordering is implemented using vector clocks.
Our implementation of Phase~(2)
closely follows the original UNDEAD depth-first algorithm
to compute cyclic lock-dependency chains,
with additional checks for Conditions~\cond{DP-P/-Block}
(where P refers to either TRW or PWR).

In our algorithms, we often write ``$\dontCare$'' in pattern matching to denote arbitrary values.

\subsection{Computation of TRW 
            and Abstract Lock Dependencies}
\label{sec:vcAndLd}

\begin{figure}[t]
  \bda{|l|l|l||l|l|}
    \hline T_{10} & \thread{1} & \thread{2} & \mbox{Standard lock dependencies} & \mbox{Abstract lock dependencies}\\ \hline
    \eventE{1}  & \lockE{\LKA}&&&\\
    \eventE{2}  & \lockE{\LKB}&& \LD{\thread{1}}{\LKB}{\{\LKA\}} & \LDMap{\thread{1}}{\LKB}{\{\LKA\}} = [(2,V_2,\{e_1\})] \\
    \eventE{3}  & \writeE{\VA}&&&\\
    \eventE{4}  & \unlockE{\LKB}&&&\\
    \eventE{5}  & \unlockE{\LKA}&&&\\
    \eventE{6}  & &\lockE{\LKB}&&\\
    \eventE{7}  & &\readE{\VA}&&\\
    \eventE{8}  & &\lockE{\LKA}& \LD{\thread{2}}{\LKA}{\{\LKB\}} & \LDMap{\thread{2}}{\LKA}{\{\LKB\}} = [(8,V_8,\{e_7\})] \\
    \eventE{9}  & &\unlockE{\LKA} &&\\
    \eventE{10}  & &\unlockE{\LKB}&&\\
    \eventE{11}  & \lockE{\LKA}&&&\\
    \eventE{12}  & \lockE{\LKB}&& \LD{\thread{1}}{\LKB}{\{\LKA\}} & \LDMap{\thread{1}}{\LKB}{\{\LKA\}} = [\smash{\begin{array}[t]{@{}l@{}} 
      (2,V_2,\{e_1\}),
      \\
      (12,V_{12},\{e_{11}\})]
    \end{array}} \\
    \eventE{13}  & \unlockE{\LKB}&&&\\
    \eventE{14}  & \unlockE{\LKA}&&&\\
    \hline
  \eda{}

  \caption{Standard versus abstract lock dependencies.}
  \label{fig:ld-abstract}
\end{figure}

\begin{algorithm*}[t!]
  \caption{Computation of TRW vector clocks and abstract lock dependencies.}
  \label{alg:lock-deps}

  {\small
    \begin{algorithmic}[1]
      \Function{computeTRWLockDeps}{$T$}
      \label{ln:cLDs}
        \State $\forall t \colon \threadVC{t} = [\bar{0}]; \incC{\threadVC{t}}{t}$
        \Comment{Vector clock $\threadVC{t}$ of thread $t$}
        \label{ln:thvc}
        \State $\forall x \colon \lastWriteVC{x} = [\bar{0}]; \lastReadVC{x} = [\bar{0}]$
        \Comment{Vector clocks $\lastWriteVC{x},\lastReadVC{x}$ of most recent $\writeE{x},\readE{x}$}
        \label{ln:lwlr}
        \State $\forall l \colon \acqVC{l} = [\bar{0}]$
        \Comment{Vector clock $\acqVC{l}$ of most recent $\acqE{l}$}
        \label{ln:acqv}
        \State $\forall l \colon \Hist{l} = []$
        \Comment{History $\Hist{l}$ of acquire-release pairs $(\Vacq,\Vrel)$ for lock $l$}
        \label{ln:hist}
        \State $\forall t \colon \AcqHeld(t) = []$
        \Comment{Sequence $\AcqHeld(t)$ of acquires held by thread $t$}
        \label{ln:acqhd}
        \State $\LDMapSym = \emptyset$
        \Comment{Map  with keys $(t,l,ls)$,
                 list values with elements $(i,V,\{a_1,\ldots,a_n\})$}
        \label{ln:ld}
        \ForDo {$e$ in $T$} {\Call{process}{$e$}}
        \State \Return \LDt
      \EndFunction
      \algstore{cld}
    \end{algorithmic}

    \begin{minipage}[t]{.52\textwidth}
      \begin{algorithmic}[1]
        \algrestore{cld}
        \Procedure{process}{$(\alpha,t,acq(l))$}
          \label{ln:procAcq}
          \If {$\AcqHeld(t) \not = []$}
          \label{ln:acq-non-empty}
          \State $ls = \{ l' \mid \acqE{l'} \in \AcqHeld(t) \}$
          \label{ln:ld-ls}
          \State $\LDMap{t}{l}{ls}.pushBack(\alpha,\threadVC{t},\AcqHeld(t))$
            \label{ln:ld-add}
          \EndIf
          \State $\AcqHeld(t) = \AcqHeld(t) \cup \{ (\alpha,t,acq(l)) \}$
          \label{ln:acq-push}
          \State $\threadVC{t} = \Call{syncCS}{\threadVC{t},\AcqHeld(t)}$
          \label{ln:acq-sync}
          \State $\acqVC{l} = \threadVC{t}$
          \label{ln:acq-store}
          \State $\incC{\threadVC{t}}{t}$
          \label{ln:acq-vc-inc}
        \EndProcedure
        \algstore{eacq}
      \end{algorithmic}

      \begin{algorithmic}[1]
        \algrestore{eacq}
        \Procedure{process}{$(\dontCare,t,rel(l))$}
          \label{ln:procRel}
          \State $\AcqHeld(t) = \{\dontCare,\dontCare,acq(l')) \in \AcqHeld(t) \mid l' \not= l \}$
          \label{ln:acq-pop}
          \State $\Hist{l} = \Hist{l} \cup \{(\Acq{l}, \threadVC{t})\}$
          \label{ln:hist-add}
          \State $\incC{\threadVC{t}}{t}$
          \label{ln:rel-vc-inc}
        \EndProcedure
        \algstore{erel}
      \end{algorithmic}


    \end{minipage}%
    \hfill%
    \begin{minipage}[t]{.46\textwidth}
      \begin{algorithmic}[1]
        \algrestore{erel}
        \Procedure{process}{$(\dontCare,t,wr(x))$}
          \label{ln:procWr}
          \State $\threadVC{t} = \threadVC{t} \sqcup \lastWriteVC{x}$
          \label{ln:ww-sync}
          \State $\threadVC{t} = \threadVC{t} \sqcup \lastReadVC{x}$
          \label{ln:rw-sync}
          \State $\threadVC{t} = \Call{syncCS}{\threadVC{t},\AcqHeld(t)}$
          \label{ln:csw-sync}
          \State $\lastWriteVC{x} = \threadVC{t}$
          \label{ln:lw-store}
          \State $\incC{\threadVC{t}}{t}$
          \label{ln:wr-vc-inc}
        \EndProcedure
        \algstore{ewr}
      \end{algorithmic}

      \begin{algorithmic}[1]
        \algrestore{ewr}
        \Procedure{process}{$(\dontCare,t,rd(x))$}
          \label{ln:procRd}
          \State $\threadVC{t} = \threadVC{t} \sqcup \lastWriteVC{x}$
          \label{ln:wr-sync}
          \State $\threadVC{t} = \Call{syncCS}{\threadVC{t},\AcqHeld(t)}$
          \label{ln:csr-sync}
            \State $\lastReadVC{x} = \threadVC{t}$ 
          \label{ln:lr-store}
          \State $\incC{\threadVC{t}}{t}$
          \label{ln:rd-vc-inc}
        \EndProcedure
        \algstore{erd}
      \end{algorithmic}

      \begin{algorithmic}[1]
        \algrestore{erd}
        \Function{syncCS}{$V, A$}
          \label{ln:syncCS}
          \For {$\acqE{l} \in A , (\Vacq, \Vrel) \in \Hist{l}$}
          \label{ln:sync-hist-entry}
            \IfThen {$\Vacq < V$} {$V = V \sqcup \Vrel$}
            \label{ln:ro-sync}
          \EndFor
        \State \Return V
        \EndFunction
      \end{algorithmic}
    \end{minipage}
  }
\end{algorithm*}


Phase~(1) computes vector clocks and lock dependencies (\cref{alg:lock-deps}).
Although our implementation supports various partial orders, our presentation details TRW (\cref{d:trw}).
Function \textsc{computeTRWLockDeps} (\cref{ln:cLDs}) takes a trace $T$, returning a set of \emph{abstract} lock dependencies
represented by the map $\LDMapSym$.

To understand the difference between standard and abstract lock dependencies,
consider trace~$T_{10}$ in \cref{fig:ld-abstract}.
For brevity, we omit explicit request events,
assuming that acquires implicitly directly preceded by a request.
Events $e_2$ and $e_{12}$ both acquire lock $\LKB$ while holding lock $\LKA$.
The resulting (standard) lock dependency
$\LD{\thread{1}}{\LKA}{\{\LKB\}}$ is sufficient to check for Conditions~\cond{DP-Cycle/-Guard}.
For the additional Conditions~\cond{DP-TRW/-Block},
we require some extra information specific to~$e_2$ and~$e_{12}$.
For Condition~\cond{DP-TRW} we need their vector clocks,
and
for Condition~\cond{DP-Block} we need their trace positions
and, for each lock held, the corresponding acquire event.
We collect this extra information as triples:
for~$e_2$ we have $(2,V_2,\{e_1\})$, and for~$e_{12}$ we have $(12,V_{12},\{e_{11}\})$.
Each triple represents a \emph{concrete} lock dependency.
Concrete lock dependencies that share the same thread, lock and locks held
are stored in an \emph{abstract} lock dependency.
An abstract lock dependency is represented as a (list) value in a map~$\LDMapSym$
with \emph{standard} lock dependencies as keys.
Initially, all lists are empty.
When processing $e_2$ we add $(2,V_{2},\{e_{1}\})$
to the empty list $\LDMap{\thread{1}}{\LKB}{\{\LKA\}}$.
After processing~$e_{12}$, we insert the associated triple, resulting in
$[(2,V_2,\{e_1\}), (12,V_{12},\{e_{11}\})]$.
New triples are added to the back of the list, to maintain the processing
order among elements.

Before detailing \cref{alg:lock-deps},
we repeat some standard definitions for vector clocks.

\begin{definition}[Vector Clocks]
  A \emph{vector clock} $V$ is a list of time stamps of the form
  $
    [i_1,\ldots,i_n]
  $.
  We assume vector clocks are of a fixed size $n$.
  \emph{Time stamps} are natural numbers, and a time stamp at position $j$ corresponds to the thread with identifier $\thread{j}$.

  We define
  $
    [i_1,\ldots, i_n] \sqcup [j_1,\ldots,j_n] = [\maxN{i_1}{j_1},\ldots,\maxN{i_n}{j_n}]
  $
  to \emph{synchronize} two vector clocks, taking pointwise maximal time stamps.
  We write $\incC{V}{j}$ to denote incrementing the vector clock $V$ at position $j$ by one,
  and $\accVC{V}{j}$ to retrieve the time stamp at position $j$.

  We write $V < V'$ if $\forall k \colon \accVC{V}{k} \leq \accVC{V'}{k} \wedge \exists k \colon \accVC{V}{k} < \accVC{V'}{k}$,
  and $V \VCConc V'$ if $V \not< V'$ and $V' \not< V$.
\end{definition}

\Cref{alg:lock-deps} processes events in trace order.
For every event $e$, we call the procedure \textsc{process}, defined differently for each event operation (\cref{ln:procAcq,ln:procRel,ln:procWr,ln:procRd}).
For brevity, we omit the standard treatment of fork/join events.

For the computation of TRW vector clocks, we maintain several state variables that
are either indexed by some thread, lock, or shared variable.
For every thread $t$, we have a vector clock~$\threadVC{t}$.
Vector clocks $\lastWriteVC{x}$ and $\lastReadVC{x}$ represent the last read and write
event on variable~$x$.
The vector clock $\acqVC{l}$ records the last acquire event on lock~$l$.
Initially, all time stamps are set to zero (\cref{ln:thvc,ln:lwlr,ln:acqv}; $\bar{0}$ denotes a sequence of zeros), and $\threadVC{t}$ is set to 1 at position $t$ (\cref{ln:thvc}).

Unlike well-known vector-clock algorithms such as
FastTrack~\cite{flanagan2010fasttrack} and
SHB~\cite{Mathur:2018:HFR:3288538:3276515},
we do not order critical sections by source trace order.
Rather, TRW only orders critical sections that contain conflicting writes and reads.
To check if critical sections are conflicting, we maintain
a critical-section history $\Hist{l}$ to track the vector clocks $\Vacq$ of acquires of lock~$l$ and $\Vrel$ of their matching releases.
In our actual implementation we use thread-local histories for efficiency reasons.
We omit details, as our use of thread-local histories
is along the lines of existing vector-clock algorithms
for WCP~\cite{DBLP:journals/corr/KiniM017} and PWR~\cite{conf/mplr/SulzmannS20}.

Additionally,
we need to know the acquire events connected to locks held.
Hence, each thread~$t$ maintains a sequence $\AcqHeld(t)$ of acquires held (\cref{ln:acqhd}).
As discussed above, the map~$\LDMapSym$~(\cref{ln:ld})
holds the set of abstract lock dependencies.

The \textsc{process} procedures update the above state variables differently for each operation, but all increment the time stamp of thread $t$ in $\threadVC{t}$ as a final step (\cref{ln:acq-vc-inc,ln:rel-vc-inc,ln:wr-vc-inc,ln:rd-vc-inc}).
We use pattern matching to distinguish operations.

For acquire event $(\alpha,t,acq(l))$ (\cref{ln:procAcq}), a new entry $(\alpha,\threadVC{t},\AcqHeld(t))$
is added to the abstract lock dependency $\LDMap{t}{l}{ls}$ (\cref{ln:ld-add}),
assuming that the sequence $\AcqHeld(t)$ of acquires held is not empty (\cref{ln:acq-non-empty}).
From $\AcqHeld(t)$ we derive the lockset~$ls$ (\cref{ln:ld-ls}).
Event id~$\alpha$ refers to the trace position of the acquire event.
Each entry uses the vector clock \threadVC{t} before synchronization with other critical sections.
Effectively, this is the vector clock of the request event that precedes the acquire,
allowing us to avoid explicitly handling requests.
$\AcqHeld(t)$ is extended with the current acquire~(\cref{ln:acq-push}).
Function call $\Call{syncCS}{\threadVC{t},\AcqHeld(t)}$ (\cref{ln:acq-sync}; discussed shortly)
enforces Condition~\cond{TRW-Rel}.
Finally, we record the vector clock of the acquire event (\cref{ln:acq-store}).

For release events $(\dontCare,t,rel(l))$ (\cref{ln:procRel}),
we remove the matching acquire from $\AcqHeld(t)$~(\cref{ln:acq-pop})
and add a new acquire-release pair to $\Hist{l}$ (\cref{ln:hist-add}).

For write events $(\dontCare,t,wr(x))$ (\cref{ln:procWr}), we synchronize $\threadVC{t}$ with the vector clock of the last write and read on $x$ (\cref{ln:ww-sync,ln:rw-sync}).
For read events $(\dontCare,t,rd(x))$ (\cref{ln:procRd}), we also synchronize $\threadVC{t}$ with the vector clock of the last write on $x$ (\cref{ln:wr-sync}).
In both cases, this enforces Condition~\cond{TRW-Conf}.
Conflicting memory operations with a common lock enforce Condition~\cond{TRW-Rel} using $\Call{syncCS}{\threadVC{t},\AcqHeld(t)}$ (\cref{ln:csw-sync,ln:csr-sync}).

Function \textsc{syncCS} (\cref{ln:syncCS}) synchronizes critical sections on the same lock.
The function cycles through every $\acqE{l} \in A$.
For each $(\Vacq,\Vrel)$ in the history $\Hist{l}$ (\cref{ln:sync-hist-entry}),
we check whether $\Vacq < V$~(\cref{ln:ro-sync}), where $V$ is the vector clock of some acquire/read/write.
If so, we synchronize $V$ and $\Vrel$, thus enforcing Conditions~\cond{TRW-Conf/-Rel}.

\paragraph{PWR vector clocks.}

To obtain PWR vector clocks, we simply drop \cref{ln:ww-sync,ln:rw-sync,ln:csw-sync,ln:lr-store} in \cref{alg:lock-deps};
all other parts remain the same.

\algdef{SE}[DOWHILE]{Do}{DoWhile}{\algorithmicdo}[1]{\algorithmicwhile\ #1}%

\begin{algorithm*}[t]
  \caption{Checking Conditions \cond{DP-Cycle/-Guard/TRW/-Block}.}
  \label{alg:dp-trw-block}

  {\small

    \begin{algorithmic}[1]
      \Function{undeadTRWBlock}{$T$}
      \label{ln:undeadTRWBlock}
      \State $\LDMapSym = \Call{computeTRWLockDeps}{T}$
      \label{ln:call-computeTRWLockDeps}
      \State $\DDs = \{ (\LDMap{\thread{1}}{l_1}{ls_1}, \ldots, \LDMap{\thread{n}}{l_n}{ls_n}) \mid \forall i \ne j \colon l_i \in ls_{(i \mod n) + 1} \wedge ls_i \cap ls_j = \emptyset \}$
      \label{ln:cyclic-chain}
      \State $\AAs = \emptyset$
      \For{$(\Ald_1,\ldots,\Ald_n) \in \DDs$}
      \label{ln:dp-trw-start}
      \If{$\Call{checkTRW}{\Ald_1,\ldots,\Ald_n} = (\OK,(\AldE_1,\ldots,\AldE_n))$}
      \label{ln:call-checkTRW}
      \State $\AAs = \AAs \cup \{ \{ \AldE_1,\ldots,\AldE_n \} \}$
      \EndIf
      \EndFor
      \label{ln:dp-trw-end}
      \State $\BBs = \emptyset$
      \For{$A \in \AAs$}
      \label{ln:dp-block-start}
      \If{$\neg \exists B \in \AAs. \Call{lessThan}{B,A}$}
      \label{ln:dp-block-check}
      \State $\BBs = \BBs \cup \{ A \}$
      \EndIf
      \EndFor
      \label{ln:dp-block-end}
      \State \Return $\BBs$
      \EndFunction
      \algstore{undeadTRWBlock}
    \end{algorithmic}

    \begin{minipage}[t]{.45\textwidth}
      \begin{algorithmic}[1]
        \algrestore{undeadTRWBlock}
        \Function{checkTRW}{$\Ald_1,\ldots,\Ald_n$}
        \label{ln:check-dp-trw}
        \State $k_i = 0$ for $i=1,\ldots,n$
        \label{ln:index-init}
        \While{$\bigwedge_{i=1}^n k_i < \Ald_i.size$}
        \label{ln:iterate-through-all}
        \State $(\dontCare,V_i,\dontCare) = \Ald_i[k_i]$ for $i=1,\ldots,n$
        \label{ln:index-access-vc}
        \If{$V_i \VCConc V_j$ for $i \not= j$}
        \label{ln:dp-trw-okay}
        \State \Return $(\OK,(\Ald_1[k_1],\ldots,\Ald_n[k_n]))$
        \EndIf
        \State $V = V_1 \sqcup \ldots \sqcup V_n$
        \label{ln:sync-candidates}
        \State $k_i = \Call{next}{V,\Ald_i,k_i}$ for $i=1,\ldots,n$
        \label{ln:next-candidates}
        \EndWhile
        \EndFunction
        \State \Return $(\FAIL,\dontCare)$
        \algstore{checkTRW}
      \end{algorithmic}
    \end{minipage}%
    \hfill%
    \begin{minipage}[t]{.50\textwidth}
      \vspace{-5.5em}
      \begin{algorithmic}[1]
        \algrestore{checkTRW}
        \Function{next}{$V,\Ald,i$}
        \Do
        \State $(\dontCare,V',\dontCare) = \Ald[i]$
        \If{$\neg (V' < V)$}
        \label{ln:next-entry}
        \State \Return $i$
        \EndIf
        \State $i = i + 1$
        \DoWhile{$i < \Ald.size$}
        \State \Return $i$
        \EndFunction
        \algstore{next}
      \end{algorithmic}

      \begin{algorithmic}[1]
        \algrestore{next}
        \Function{lessThan}{$A,B$}
        \For{$(i,\dontCare,as) \in A$}
        \If{$\neg \exists (j,\dontCare,bs) \in B \colon i < j \wedge \Call{cycLk}{as} \in bs$}
        \State \Return false
        \EndIf
        \EndFor
        \State \Return true
        \EndFunction
      \end{algorithmic}
    \end{minipage}
  }
\end{algorithm*}

\subsection{Computing and Verifying Deadlock Patterns}
\label{sec:impl:verify}

Function \textsc{undeadTRWBlock} (\cref{ln:undeadTRWBlock})
in \cref{alg:dp-trw-block}
is the main entry point of our algoritm.
We first call \textsc{computeTRWLockDeps} (in \cref{alg:lock-deps}) to compute
the set of abstract lock dependencies~(\cref{ln:call-computeTRWLockDeps} discussed in \cref{sec:vcAndLd}).

Next, we enumerate in $\DDs$ all cyclic chains of abstract lock dependencies (referred to in~\cite{conf/pldi/TuncMPV23} as ``abstract lock dependencies'') that
satisfy Conditions~\cond{DP-Cycle/-Guard} (\cref{ln:cyclic-chain}).
We refer to~\cite{conf/ase/ZhouSLCL17} for details on the computation of~$\DDs$.
Elements in $\DDs$ are sequences of the form $(\Ald_1,\ldots,\Ald_n)$,
where each $\Ald_i$ is an abstract lock dependency, i.e., a list of triples $\AldE_j = (j,V,\{a_1,\ldots,a_k\})$.
For every such~$(\Ald_1,\ldots,\Ald_n)$, we then look for a concrete
instance~$(\AldE_1,\ldots,\AldE_n)$, where every $\AldE_i \in \Ald_i$,
that satisfies Conditions~\cond{DP-TRW/-Block}.

Function \textsc{checkTRW} looks for an instance
that satisfies Condition~\cond{DP-TRW}~(\cref{ln:check-dp-trw}).
We systematically enumerate instances by starting
with the first element in $\Ald_i$~(\cref{ln:index-init}).
We use array access notation to obtain the $k_i$th element
combined with pattern matching to retrieve the vector clock~$V_i$
of that element (\cref{ln:index-access-vc}).
If the vector clocks are pairwise incomparable,
we have found an instance that satisfies~Condition~\cond{DP-TRW}~(\cref{ln:dp-trw-okay}).
Otherwise, we need to select a new instance.

To avoid naive enumeration of all instances, we exploit the property that for
${\Ald = [ \AldE_1,\ldots,\AldE_k]}$ we find that
$V_j < V_{j+1}$ where $\AldE_j = (\dontCare,V_j,\dontCare)$ for every $\AldE_j \in \Ald$.
This follows from the fact that the elements of $\Ald$ are stored in
the order in which they are generated
(cf.\ \cref{ln:ld-add} in \cref{alg:lock-deps}).
Hence, if we find that $V_i < V_j$,
we can immediately conclude that, for any $V'$
where $(\dontCare,V',\dontCare) = \Ald_i[k]$ for~$k<k_i$,
we have $V' < V_j$.
It follows that it suffices to check for ``later'' candidates $V'$
for which~$V' < V_j$ does \emph{not} hold.
This is achieved by synchronizing all candidates (\cref{ln:sync-candidates})
and calling function \textsc{next}~(\cref{ln:next-candidates}).

The resulting instances, that satisfy Condition~\cond{DP-TRW}, are stored in $\AAs$.
What remains is to select only those instances that satisfy Condition~\cond{DP-Block}.
Function~\textsc{lessThan} checks whether two instances are ordered
following the description in \cref{def:ltDP}.
For convenience, we omit the helper function~\textsc{cycLk}
that retrieves the acquire held that is part
of the cyclic dependency, as stated in Condition~\cond{DP-Cycle}.
The resulting instances stored in $\BBs$ represent
deadlocks patterns satisfying Conditions~\cond{DP-Cycle/-Guard/TRW/-Block}.

For example, for trace~$T_{10}$ in \cref{fig:ld-abstract},
we find
\[
\DDs = \{ ([(2,V_2,\{e_1\}), (12,V_{12},\{e_{11}\})],[(8,V_8,\{e_7\})]) \}
\]
This leads to the call
$
\textsc{checkTRW}(([(2,V_2,\{e_1\}), (12,V_{12},\{e_{11}\})],[(8,V_8,\{e_7\})]))
$.
The first instance to be checked
is the pair of elements $(2,V_2,\{e_1\})$ and $(8,V_8,\{e_7\})$, but it is rejected because $V_2 < V_8$.
The next pair of elements considered is
$(12,V_{12},\{e_{11}\})$ and $(8,V_8,\{e_7\})$.
This pair is successful: we find that $V_{12} \VCConc V_8$.
The resulting instance is $\{ (12,V_{12},\{e_{11}\}), (8,V_8,\{e_7\}) \}$,
which is the only deadlock pattern reported for this example.

\subsection{Time Complexity}
\label{sec:complexity}

\newcommand{\TRACELEN}{\mathcal{N}}
\newcommand{\THREADNO}{\mathcal{T}}
\newcommand{\CSNO}{\mathcal{C}}
\newcommand{\CYCLES}{\mathit{Cyc}}

We first consider the time complexity of \cref{alg:lock-deps}.
Let $\THREADNO$ be the number of threads and
$\CSNO$ be the number of critical sections (acquire-release pairs).
Then, each call to function~\textsc{syncCS} takes time $O(\THREADNO \cdot \CSNO)$,
because the number of acquires held is treated as a constant,
there are $O(\CSNO)$ entries to consider, and
for each entry the vector clock operations involved take time~$O(\THREADNO)$.
Hence, function~\textsc{computeTRWLockDeps} takes time~$O(\TRACELEN \cdot \THREADNO \cdot \CSNO)$ where $\TRACELEN$ is the length of the trace,
assuming that access and manipulation of state variables ($\threadVC{t},...$)
takes time $O(1)$ and ignoring initialization of state variables.
Instead of a global history~$\Hist{l}$ of size $O(\CSNO)$,
our actual implementation follows WCP~\cite{DBLP:journals/corr/KiniM017}
and PWR~\cite{conf/mplr/SulzmannS20} in using thread-local histories.
This way, for realistic examples
the number of entries in thread-local histories can be treated like a constant.
We therefore argue that the time complexity of \textsc{syncCS} can be simplified to $O(\THREADNO)$, and hence that \cref{alg:lock-deps} takes time $O(\TRACELEN \cdot \THREADNO)$.

What remains is to consider the time complexity of~\cref{alg:dp-trw-block}.
As shown by \Tunc et al.~\cite{conf/pldi/TuncMPV23},
the number of cyclic chains of abstract lock dependencies
can be exponential in terms of the number of threads and acquires.
We follow~\cite{conf/pldi/TuncMPV23} in writing $O(\CYCLES)$ to represent this number
as well as the time complexity of computing set $\DDs$ (\cref{ln:cyclic-chain}).
Next, we consider the complexity of~\textsc{checkTRW} (\cref{ln:call-checkTRW}).
We treat the length of the sequence $(\Ald_1,\ldots,\Ald_n)$ as a constant and
always strictly move forward through the list of candidates in $\Ald_i$.
In each step, vector clock operations take time $O(\THREADNO)$.
Hence, function~\textsc{checkTRW} takes time~$O(\TRACELEN \cdot \THREADNO)$.
This shows that \cref{ln:dp-trw-start,ln:call-checkTRW,ln:dp-trw-end}
take time $O(\CYCLES \cdot \TRACELEN \cdot \THREADNO)$.
Function~\textsc{lessThan} takes time $O(\CYCLES)$.
Thus, we find that \cref{alg:dp-trw-block} takes time
$O(\CYCLES + \CYCLES \cdot \TRACELEN \cdot \THREADNO + \CYCLES^2)$.

Overall, our method takes time $O(\TRACELEN \cdot \THREADNO + \CYCLES + \CYCLES \cdot \TRACELEN \cdot \THREADNO + \CYCLES^2)$.
In practice, $O(\CYCLES)$ is small and $O(\THREADNO)$ can be interpreted as a constant, such that the complexity reduces to~$O(\TRACELEN)$.
This is confirmed by our experiments (\cref{sec:experiments}).


\section{Experiments}
\label{sec:experiments}

\newcommand{\UNDEAD}{\mbox{UNDEAD}\xspace}

\newcommand{\UDPWRLimitH}{\mbox{UD-PWR2}\xspace}  

\newcommand{\UDPWRRandom}{\mbox{UD-PWR3}\xspace}  

\newcommand{\UDPWRRandomLimitH}{\mbox{UD-PWR4}\xspace}  

\newcommand{\Benchmark}{\mbox{\bf Benchmark}}
\newcommand{\TT}{\mbox{${\mathcal T}$}}
\newcommand{\EE}{\mbox{${\mathcal E}$}}
\newcommand{\MM}{\mbox{${\mathcal M}$}}
\newcommand{\LL}{\mbox{${\mathcal L}$}}
\newcommand{\Sum}{\mbox{$\sum$}}

\newcommand{\Cycles}{\mbox{Dlk}}      
\newcommand{\Dependencies}{\mbox{Deps}}
\newcommand{\DepsGuarded}{\mbox{Grds}}
\newcommand{\DepsExtras}{\mbox{E}}
\newcommand{\Time}{\mbox{Time}}
\newcommand{\PhaseOne}{\mbox{P1}}
\newcommand{\PhaseTwo}{\mbox{P2}}
\newcommand{\Races}{\mbox{Races}}
\newcommand{\Guards}{\mbox{Guards}}

\newcommand{\bm}[1]{\texttt{#1}}


\newcommand{\udsymbol}{UD}
\newcommand{\UD}{\mbox{\udsymbol}\xspace}
\newcommand{\UDFJ}{\mbox{\udsymbol\textsubscript{FJ}}\xspace}
\newcommand{\UDFJWRD}{\mbox{\udsymbol\textsubscript{LW}}\xspace}
\newcommand{\UDPWR}{\mbox{\udsymbol\textsubscript{PWR}}\xspace}
\newcommand{\UDTRW}{\mbox{\udsymbol\textsubscript{TRW}}\xspace}
\newcommand{\UDTRWEvict}{\mbox{\udsymbol\textsubscript{TRW$\_$R}}\xspace}
\newcommand{\UDPWRSyncP}{\mbox{\udsymbol\textsubscript{PWR}\textsubscript{+SPD}}\xspace}

\newcommand{\UDTRWGuardCheck}{\mbox{\udsymbol\textsubscript{TRW$\_${GC}}}\xspace}

\newcommand{\UDFJWRDVariant}{\mbox{\udsymbol\textsubscript{LW}$^*$}\xspace}


\newcommand{\UDPWROpt}{\mbox{UPWROpt}\xspace}
\newcommand{\UDPWRTO}{\mbox{UNDEAD\textsubscript{TO}}\xspace}
\newcommand{\UDSyncP}{\mbox{SPDOffline}\xspace} 
\newcommand{\UDPWRDropSyncP}{\mbox{SPD\textsubscript{PWR}\textsuperscript{Drop}}\xspace}
\newcommand{\UDPWRTOHACK}{\mbox{UNDEAD\textsubscript{TOH}}\xspace}

\newcommand{\SPDVectorClocks}{LW\xspace}

We evaluated our approach in an offline setting using a large dataset
of pre-recorded program traces from prior work.
These traces include fork and join events, whose straightforward treatment we do not discuss in this paper due to space limitations.

\paragraph{Test candidates.}

For experimentation, we considered the following four test candidates.
All our candidates are implemented in C++.%
\footnote{Available at \url{https://osf.io/ku9fx/files/osfstorage?view_only=b7f53d3110894fe39ad1520ed0fed4ec} (anonymous).}
\begin{description}
  \item[\UD] is the original UNDEAD implementation~\cite{conf/ase/ZhouSLCL17} adapted to work in an
    offline setting.
  \item[\UDTRW] implements TRW according to \cref{alg:lock-deps,alg:dp-trw-block}.
  \item[\UDPWR] adapts \cref{alg:lock-deps,alg:dp-trw-block} to PWR.
  \item[\SPDOfflineUD] is our C++ reimplementation of \SPDOffline~\cite{conf/pldi/TuncMPV23}.\footnote{\SPDOffline is written in Java. For a fairer comparison, we therefore use our version \SPDOfflineUD.}
\end{description}

\SPDOfflineUD employs two phases that roughly correspond to \cref{alg:lock-deps} and \cref{alg:dp-trw-block}.
In Phase~(1), \SPDOfflineUD computes \SPDVectorClocks vector clocks, which is a simplification of PWR that only satisfies Conditions~\cond{PWR-PO/-LW} (cf.\ \cref{d:pwr}).
In Phase~(2), \SPDOfflineUD makes use of \SPDVectorClocks vector clocks to check
if there is a sync-preserving instance of
a cyclic chain $(\Ald_1,\ldots,\Ald_n)$ of abstract lock dependencies (recall that sync-preserving means that the order of critical sections on the same lock is preserved).
This check is carried out by a call to function~\textsc{CompSPClosure} from~\cite[Algorithm 1]{conf/pldi/TuncMPV23}, replacing \cref{ln:dp-trw-okay} in \cref{alg:dp-trw-block} (the check for Condition \cond{DP-TRW}).
As \SPDOfflineUD does not need to check for Condition~\cond{DP-Block}, \cref{ln:dp-block-start,ln:dp-block-check,ln:dp-block-end} are dropped.

As discussed in \cref{s:sound,s:complete}, in general
\UDTRW is sound but incomplete whereas
\UDPWR is complete but unsound.
\SPDOfflineUD is sound
but only covers sync-preserving deadlocks~\cite{conf/pldi/TuncMPV23}.

\paragraph{Benchmarks and system setup.}

Our experiments are based on a large set of benchmark traces from
prior work~\cite{conf/pldi/TuncMPV23,10.1145/3503222.3507734}.%
We excluded five benchmarks (``RayTracer'', ``jigsaw'', ``Sor'', ``Swing'', ``eclipse''),
because we noticed that their traces are not well formed.
For example, locks are acquired by distinct threads with no release in between.
We suspect that inaccurate trace recording is at fault.
The issue has been confirmed by \Tunc et al.~\cite{conf/pldi/TuncMPV23}.

We conducted our experiments on an Apple M1 max CPU with 32GB of RAM
running macOS Monterey (Version 12.1).

\paragraph{Evaluation.}

\begin{table*}[p]
  \caption{
    \small
    \textbf{Deadlock warnings and running times}.
    Columns~2--5 contain the number of events, of threads,
    of memory locations, and of locks, respectively.
    Columns~6--13 show the number of deadlocks reported and running time for each candidate.
    Times include the time to both compute lock dependencies and identify deadlocks (i.e., Phases~(1) and~(2)).
    For \UD, times are in seconds, and rounded to the nearest hundreth.
    For \UDPWR, \UDTRW and \SPDOfflineUD, times are factors compared to \UD.
  }
  \label{tbl:benchmarks}
  {
  \small
  \setlength{\tabcolsep}{2.9pt} 
  \renewcommand{\arraystretch}{0.87} 
\begin{tabular}{|r|r|r|r|r||r|r||r|r||r|r||r|r|}
 \hline
1 & 2 & 3 & 4 & 5 & 6 & 7 & 8 & 9 & 10 & 11 & 12 & 13
 \\
 \hline
\multirow{2}{*}{\Benchmark} &
\multirow{2}{*}{\EE} &
\multirow{2}{*}{\TT} &
\multirow{2}{*}{\MM} &
\multirow{2}{*}{\LL} &
\multicolumn{2}{c||}{\UD} &
\multicolumn{2}{c||}{\UDPWR} &
\multicolumn{2}{c||}{\UDTRW} &
\multicolumn{2}{c|}{\SPDOfflineUD} \\ \cline{6-7} \cline{8-9} \cline{10-11} \cline{12-13}
 & & & &
& \Cycles &  \Time
& \Cycles & \Time
& \Cycles & \Time
& \Cycles & \Time
 \\
 \hline
Deadlock & 28 & 3 & 3 & 2 & 1 & 0.00 & 0 & 1x & 0 & 1x & 0 & 1x  \\   \hline
NotADeadlock & 42 & 3 & 3 & 4 & 1 & 0.00 & 0 & 1x & 0 & 1x & 0 & 1x  \\   \hline
Picklock & 46 & 3 & 5 & 5 & 2 & 0.00 & 1 & 1x & 1 & 1x & 1 & 1x  \\   \hline
Bensalem & 45 & 4 & 4 & 4 & 2 & 0.00 & 1 & 1x & 1 & 1x & 1 & 1x  \\   \hline
Transfer & 56 & 3 & 10 & 3 & 1 & 0.00 & 0 & 1x & 0 & 1x & 0 & 1x  \\   \hline
Test-Dimminux & 50 & 3 & 8 & 6 & 2 & 0.00 & 2 & 1x & 2 & 1x & 2 & 1x  \\   \hline
StringBuffer & 57 & 3 & 13 & 3 & 1 & 0.00 & 1 & 1x & 1 & 1x & 1 & 1x  \\   \hline
Test-Calfuzzer & 126 & 5 & 15 & 5 & 1 & 0.00 & 1 & 1x & 1 & 1x & 1 & 1x  \\   \hline
DiningPhil & 210 & 6 & 20 & 5 & 1 & 0.00 & 1 & 1x & 1 & 1x & 1 & 1x  \\   \hline
HashTable & 222 & 3 & 4 & 2 & 0 & 0.00 & 0 & 1x & 0 & 1x & 0 & 1x  \\   \hline
Account & 617 & 6 & 46 & 6 & 3 & 0.00 & 0 & 1x & 0 & 1x & 0 & 1x  \\   \hline
Log4j2 & 1\mbox{K} & 4 & 333 & 10 & 0 & 0.00 & 0 & 1x & 0 & 1x & 0 & 1x  \\   \hline
Dbcp1 & 2\mbox{K} & 3 & 767 & 4 & 2 & 0.00 & 1 & 1x & 1 & 1x & 1 & 1x  \\   \hline
Dbcp2 & 2\mbox{K} & 3 & 591 & 9 & 1 & 0.01 & 0 & 1x & 0 & 1x & 0 & 1x  \\   \hline
Derby2 & 3\mbox{K} & 3 & 1\mbox{K} & 3 & 0 & 0.00 & 0 & 1x & 0 & 1x & 0 & 1x  \\   \hline
elevator & 222\mbox{K} & 5 & 726 & 51 & 0 & 0.50 & 0 & 1x & 0 & 1x & 0 & 1x  \\   \hline
hedc & 410\mbox{K} & 7 & 109\mbox{K} & 7 & 0 & 0.97 & 0 & 1x & 0 & 1x & 0 & 1x  \\   \hline
JDBCMySQL-1 & 436\mbox{K} & 3 & 73\mbox{K} & 10 & 2 & 1.02 & 0 & 1x & 0 & 1x & 0 & 1x  \\   \hline
JDBCMySQL-2 & 436\mbox{K} & 3 & 73\mbox{K} & 10 & 0 & 1.00 & 0 & 1x & 0 & 1x & 0 & 1x  \\   \hline
JDBCMySQL-3 & 436\mbox{K} & 3 & 73\mbox{K} & 12 & 8 & 1.02 & 1 & 1x & 1 & 1x & 1 & 1x  \\   \hline
JDBCMySQL-4 & 437\mbox{K} & 3 & 73\mbox{K} & 13 & 6 & 1.00 & 1 & 1x & 1 & 1x & 1 & 1x  \\   \hline
cache4j & 758\mbox{K} & 2 & 46\mbox{K} & 19 & 0 & 1.77 & 0 & 1x & 0 & 1x & 0 & 1x  \\   \hline
ArrayList & 3\mbox{M} & 801 & 121\mbox{K} & 801 & 4 & 5.94 & 4 & 4x & 4 & 4x & 4 & 1x  \\   \hline
IdentityHashMap & 3\mbox{M} & 801 & 496\mbox{K} & 801 & 1 & 6.12 & 1 & 4x & 1 & 4x & 1 & 1x  \\   \hline
Stack & 3\mbox{M} & 801 & 118\mbox{K} & 2\mbox{K} & 3 & 7.66 & 3 & 5x & 3 & 5x & 3 & 1x  \\   \hline
LinkedList & 3\mbox{M} & 801 & 290\mbox{K} & 801 & 4 & 7.85 & 4 & 3x & 4 & 3x & 4 & 1x  \\   \hline
HashMap & 3\mbox{M} & 801 & 555\mbox{K} & 801 & 1 & 7.85 & 1 & 3x & 1 & 3x & 1 & 1x  \\   \hline
WeakHashMap & 3\mbox{M} & 801 & 540\mbox{K} & 801 & 1 & 7.97 & 1 & 3x & 1 & 3x & 1 & 1x  \\   \hline
Vector & 3\mbox{M} & 3 & 14 & 3 & 1 & 8.78 & 1 & 1x & 1 & 1x & 1 & 1x  \\   \hline
LinkedHashMap & 4\mbox{M} & 801 & 617\mbox{K} & 801 & 1 & 9.62 & 1 & 2x & 1 & 3x & 1 & 1x  \\   \hline
montecarlo & 8\mbox{M} & 3 & 850\mbox{K} & 2 & 0 & 18.63 & 0 & 1x & 0 & 1x & 0 & 1x  \\   \hline
TreeMap & 9\mbox{M} & 801 & 493\mbox{K} & 801 & 1 & 20.73 & 1 & 1x & 1 & 1x & 1 & 1x  \\   \hline
hsqldb & 20\mbox{M} & 46 & 945\mbox{K} & 402 & 0 & 49.82 & 0 & 1x & 0 & 1x & 0 & 1x  \\   \hline
sunflow & 21\mbox{M} & 15 & 2\mbox{M} & 11 & 0 & 53.26 & 0 & 1x & 0 & 1x & 0 & 1x  \\   \hline
jspider & 22\mbox{M} & 11 & 5\mbox{M} & 14 & 0 & 56.28 & 0 & 1x & 0 & 1x & 0 & 1x  \\   \hline
tradesoap & 42\mbox{M} & 236 & 3\mbox{M} & 6\mbox{K} & 2 & 114.17 & 0 & 1x & 0 & 1x & 0 & 1x  \\   \hline
tradebeans & 42\mbox{M} & 236 & 3\mbox{M} & 6\mbox{K} & 2 & 114.26 & 0 & 1x & 0 & 1x & 0 & 1x  \\   \hline
TestPerf & 80\mbox{M} & 50 & 598 & 8 & 0 & 173.71 & 0 & 1x & 0 & 1x & 0 & 1x  \\   \hline
Groovy2 & 120\mbox{M} & 13 & 13\mbox{M} & 10\mbox{K} & 0 & 308.10 & 0 & 1x & 0 & 1x & 0 & 1x  \\   \hline
tsp & 307\mbox{M} & 10 & 181\mbox{K} & 2 & 0 & 876.75 & 0 & 1x & 0 & 1x & 0 & 1x  \\   \hline
lusearch & 217\mbox{M} & 10 & 5\mbox{M} & 118 & 0 & 597.37 & 0 & 1x & 0 & 1x & 0 & 1x  \\   \hline
biojava & 221\mbox{M} & 6 & 121\mbox{K} & 78 & 0 & 595.36 & 0 & 1x & 0 & 1x & 0 & 1x  \\   \hline
graphchi & 216\mbox{M} & 20 & 25\mbox{M} & 60 & 0 & 615.90 & 0 & 1x & 0 & 1x & 0 & 1x  \\   \hline
 \hline
\Sum & 1354\mbox{M} & 7\mbox{K} & 61\mbox{M} & 30\mbox{K} & 55 & 3663.45 & 27 & 1x & 27 & 1x & 27 & 1x
 \\  \hline   \end{tabular}
}
\end{table*}

\Cref{tbl:benchmarks} contains all benchmark results.
It reports the running time and the number of deadlocks
reported for each candidate.

\paragraph{Precision.}

As expected, the number of deadlocks reported decreases
when comparing \UD against \UDPWR and \UDTRW.
\UD reports 55 deadlocks overall,
whereas \UDPWR, \UDTRW  and \SPDOfflineUD report 27 deadlocks overall.
\UDPWR and \UDTRW report the exact same deadlocks.
That is, each deadlock reported is the same instance resulting
from a cyclic chain of abstract lock dependencies.
Thus, we can conclude:
\begin{itemize}
   \item All deadlocks reported by \UDPWR\ are true positives, because \UDTRW is sound.
   \item \UDTRW has no false negatives, because \UDPWR is complete.
\end{itemize}

Benchmark ``Groovey2'' contains TRW-unbounded critical sections
and the benchmark ``hsqldb'' is ill nested (cf.\ \cref{s:sound:sound}).
It is straightforward to modify Phase~(1) to check for these conditions; for brevity, we omit details.
As there are no deadlocks reported for both these benchmarks, soundness of the results of \UDTRW is not affected.

For almost all benchmarks, Conditions~\cond{DP-PWR} and~\cond{DP-TRW}
are crucial for eliminating false positives.
Benchmark ``Picklock'' is the only exception:
Condition~\cond{DP-Block} was needed to eliminate a false positive.

\SPDOfflineUD and \UDTRW report similar deadlocks.
The reported deadlocks result from the same cyclic chains of abstract lock dependencies,
but some of the reports by \UDTRW are not sync-preserving.
However, this should not be viewed as evidence that \emph{most} deadlocks
in practice are sync-preserving, as the benchmark traces do not necessarily represent the full spectrum of concurrency patterns in modern programming (in fact, the benchmark traces were obtained
from programs that are several years old).

\paragraph{Performance.}

The overall running times of \UDPWR, \UDTRW and \SPDOfflineUD
are in the same range as \UD.
For seven benchmarks (``ArrayList''--``WeakHashMap'', ``LinkedHashMap'')
we encounter an increase by a factor of two to five
for \UDPWR and \UDTRW compared to \UD and \SPDOfflineUD.
This increase is due to the large number of 801~threads.
The computation of PWR and TRW vector clocks requires tracking conflicts
among critical sections in different threads.
The more threads the more conflict management takes place, leading to
some overhead in Phase~(1).
Despite a similarly large number of threads (256--801),
there is no increase for benchmarks ``TreeMap'', ``tradesoap'', and ``tradebeans''.
This is due to fewer conflicts among critical sections.

The running times of \UDPWR\ and \UDTRW are effectively the same,
except for benchmark~``LinkedHashMap''.
The difference of a factor of two (\UDPWR) versus a factor of three~(\UDTRW)
is due to the fact that \UDTRW has to additionally deal with write-write and read-write conflicts (Condition~\cond{PWR-LW} versus Condition~\cond{TRW-Conf}).
As ``LinkedHashMap'' has a large number of memory locations (617K),
the running time of \UDTRW is slightly higher due to more management and
synchronization of conflicting memory operations.

Overall, \SPDOfflineUD runs a bit faster than \UDPWR and \UDTRW.
The main reason is that in Phase~(1), \SPDOfflineUD computes
\SPDVectorClocks vector clocks that are not affected by large numbers of threads.
On the other hand, Phase~(2) of \SPDOfflineUD builds
sync-preserving closures, whereas \UDTRW only compares
vector clocks (see \cref{ln:dp-trw-okay} in \cref{alg:dp-trw-block}).
This difference entails a complexity difference of $O(\THREADNO)$ for \UDTRW versus $O(\TRACELEN)$ for \SPDOfflineUD, due to which we may expect \UDTRW to run faster than \SPDOfflineUD (usually, $\TRACELEN$ is much bigger than $\THREADNO$).
However, table 2, which splits running times into Phase~(1) and Phase~(2), show that this does not affect the benchmark traces: the running times of Phase~(2) are neglible compared to Phase~(1).

%

\begin{table*}[t]
  \caption{
    \textbf{Further details: running times of Phases~(1) and~(2) and numbers of concrete lock dependencies.}
      Columns~2---9 contain the number of deadlocks reported, of acquires that
      lead to a lock dependency (the same for \UDTRW and \SPDOfflineUD), and running time for each candidate.
      Times are rounded in seconds reported for Phases~(1) and~(2) separately.
  }%
  \label{tbl:details}%
{%
  \small%
  \setlength{\tabcolsep}{3.2pt}
\begin{tabular}{|r||r|r|r||r|r||r|r|r|}
 \hline
1 & 2 & 3 & 4 & 5 & 6 & 7 & 8 & 9
 \\
 \hline
\multirow{2}{*}{\Benchmark} &
\multicolumn{3}{c||}{\UDTRWEvict} &
\multicolumn{2}{c||}{\UDTRW} &
\multicolumn{3}{c|}{\SPDOfflineUD} \\ \cline{2-4} \cline{5-6} \cline{7-9}
& \Cycles &  \Dependencies & \Time \ (\PhaseOne+\PhaseTwo)
& \Cycles & \Time \ (\PhaseOne+\PhaseTwo)
& \Cycles & \Dependencies & \Time \ (\PhaseOne+\PhaseTwo)
 \\
 \hline
\ldots & \ldots & \ldots & \ldots  & \ldots & \ldots & \ldots & \ldots & \ldots  \\   \hline
Vector & 1 & 3 & 11 (11+0) & 1 & 10 (10+0) & 1 & 200\mbox{K} & 9 (9+0)  \\   \hline
tradesoap & 0 & 9\mbox{K} & 173 (166+6) & 0 & 174 (168+6) & 0 & 40\mbox{K} & 163 (157+6)  \\   \hline
Groovy2 & 0 & 11\mbox{K} & 380 (379+1) & 0 & 386 (385+1) & 0 & 29\mbox{K} & 372 (371+1)  \\   \hline
\Sum & 27 & 33\mbox{K} & 4533 (4518+14) & 27 & 4546 (4531+15) & 27 & 626\mbox{K} & 4348 (4333+15)
\\  \hline   \end{tabular}%
}%
\end{table*}

\paragraph{Reducing the number of concrete lock dependencies.}

As also observed in \cite{conf/pldi/TuncMPV23},
the number of concrete lock dependencies can be huge while
the number of abstract lock dependencies remains small.
Recall trace~$T_{10}$ in \cref{fig:ld-abstract}, where
subtrace $[e_{11},...,e_{14}]$ might result from the body of a loop.
In further loop iterations, further entries
will be added to the abstract lock dependency~$\LDMap{\thread{1}}{\LKB}{\{\LKA\}}$.
The number of entries (i.e., concrete lock dependency) increases,
whereas the number of abstract lock dependencies remains the same.

In our approach, we can aggressively remove `duplicates' of concrete lock dependencies.
We consider two entries as duplicates if no inter-thread synchronization took place between
processing the respective acquires.
Such an optimization is not possible for \SPDOfflineUD, because the trace order
of concrete lock dependencies matters (for sync-preservation):
removing a concrete lock dependency may cause a deadlock pattern to no longer be sync-preserving.

We implemented a variant of \UDTRW, referred to as \UDTRWEvict, where duplicates are removed.
\Cref{tbl:details} compares these variants to each other and to \SPDOfflineUD by including separately the running time of Phases~(1) and~(2);
for brevity, we only detail a few selected cases but include the totals over all benchmark traces.
The table shows that the number of concrete lock dependencies
can be reduced substantially.
For example, in case of benchmark ``Vector'' the number of concrete lock dependencies
is reduced from 200K to three.
We might expect that fewer concrete lock dependencies causes Phase~(2) of \UDTRWEvict to
run faster, because the loop in
function~\textsc{checkTRW}~(\cref{alg:dp-trw-block}) needs to consider fewer candidates.
However, the measurements in \cref{tbl:details} show that this is not the case in practice:
\UDTRW and \UDTRWEvict run equally fast.
We believe that this is due to the few and rather simple deadlocks in our benchmark suite.

\paragraph{Deadlocks reported in \cite{conf/pldi/TuncMPV23} and trace anomalies.}

The number of deadlocks reported in
\cref{tbl:benchmarks}
differ to that in~\cite[Table~1]{conf/pldi/TuncMPV23} for \SPDOffline.
For example, for ``ArrayList'',
\SPDOfflineUD reports four deadlocks while
\SPDOffline only reports three.
We believe that this is caused by a difference in handling lock-request versus lock-acquire events:

\begin{enumerate}
  \item\label{it:UD}
    Our implementations (derived from UNDEAD) treats the vector clock of a thread just before processing an acquire event to correspond to the corresponding request event, effectively allowing us to
    ignore all request events.
  \item\label{it:SPD}
    \SPDOffline assumes that every acquire event is preceded by a request
    event; based on our own knowledge and correspondence with the authors of~\cite{conf/pldi/TuncMPV23}, \SPDOffline uses these request events explicitly.
\end{enumerate}

Both approaches have advantages and disadvantages.
Some (prematurely-ended) traces may end in request events not followed by a corresponding acquire event.
Hence, Approach~\labelcref{it:SPD} may identify more cycles than Approach~\labelcref{it:UD};
this seems to be the case for ``JDBCMySQL-2'' (based on a comparison of \cref{tbl:benchmarks} and \cite[Table~1]{conf/pldi/TuncMPV23}).
However, not every acquire event is preceded by a request event, so Approach~\labelcref{it:UD} may report more cycles;
this applies to trace~``ArrayList''.

In fact, we even encountered a trace (``Log4j2'')
where there is a matching request-release pair of events without an acquire event between.
Moreover, we encountered some ill-formed traces, where a lock has been acquired by two distinct threads
without being released in between.
This explains the difference in number of deadlocks reported between \cref{tbl:benchmarks} and \cite[Table~1]{conf/pldi/TuncMPV23}.


\section{Related Work}
\label{sec:related}

\paragraph{Lockset-based dynamic resource-deadlock analysis.}

The idea of identifying deadlocks via circularity in the lock-order dependency
of threads dates back to the work by %
Havelund~\cite{10.5555/645880.672085} and
Harrow~\cite{DBLP:conf/spin/Harrow00}.
In subsequent work,
Bensalem and Havelund~\cite{10.1007/11678779_15} introduce lock-order graphs
that capture the lock-order dependency for threads,
where deadlock analysis reduces to checking for cycles.
Extensions of lock-order graphs to avoid false positives if a circular lock-order
dependency takes place within a single thread or is protected by a common lock
are discussed by
Agarwal et al.~\cite{5571951}.

Instead of lock-order graphs,
Joshi et al.~\cite{10.1145/1542476.1542489}
introduce lock dependencies on a per-thread basis.
The advantage compared to lock-order graphs is that common locks or singular threads
can be easily detected.
Several works~\cite{Samak:2014:TDD:2692916.2555262,conf/ase/ZhouSLCL17,6718069}
improve on this idea, e.g., by using an efficient representation
for lock dependencies and/or ignoring impossible cyclic chains
due to fork-join dependencies.
Our experiments show that considering fork-join dependencies
only avoids some false positives, but a significant number remains.

\paragraph{Sound dynamic resource-deadlock analysis.}

Lockset-based analysis methods are prone to false positives.
One way to eliminate false positives is to re-run the trace
(and/or the program)
to verify that a deadlock exists (e.g.~\cite{10.1007/978-3-319-23404-5_13,10.1145/1542476.1542489,Samak:2014:TDD:2692916.2555262}).


Kalhauge and Palsberg~\cite{conf/oopsla/KalhaugeP18} rely on an SMT-solver for exhaustive trace
exploration to eliminate false positives among deadlock patterns, but
the use of SMT-solving may severely impact performance~\cite{conf/pldi/TuncMPV23}.
Also, false positives may arise if
a request is guarded by a lock in another thread (cf., e.g., the example
in \cref{fig:guarded}, adapted from \Tunc et al.~\cite{conf/pldi/TuncMPV23}).
Our soundness result (\cref{thm:TRWSound}) explicitly excludes such traces.


\Cref{sec:experiments} contains an extensive discussion of the relation
to \SPDOffline~\cite{conf/pldi/TuncMPV23}.

SeqCheck~\cite{conf/fse/CaiYWQP21}
is similar to \SPDOffline, also employing a closure construction
that relies on a partial order to eliminate infeasible deadlock patterns.
The SeqCheck partial order does not impose last-write dependencies
but a weaker form of observation order,
so further checks are required during the closure construction.
According to \Tunc et al.~\cite{conf/pldi/TuncMPV23},
SeqCheck infers almost the same deadlocks as \SPDOffline,
but the running time is significantly higher.
We were not able to include SeqCheck in our own measurements, because our requests to access the artifact were not answered.

Ang and Mathur~\cite{DBLP:conf/cav/AngM24} discuss a novel approach to predictive
monitoring and its application to race and deadlock
detection, based on Mazurkiewicz's trace equivalence rather than
classical partial-order techniques.
The idea is to match critical
patterns like conflicting events or deadlock patterns against a trace
while assuming commutativity of all events except those that are
marked dependent (like those in the same thread, operating on the same
lock, etcetera).
As standard trace equivalence is unsuitable for predicting
deadlocks, they introduce two refinements:
strong trace prefixes and
strong reads-from prefixes.
With these refinements, they obtain
algorithms for detecting sync-preserving data races and
deadlocks.
Their experimental results for deadlocks match those of
previous work~\cite{conf/pldi/TuncMPV23,10.1145/3434317}, but with
substantial slowdown~\cite[Table~3]{DBLP:conf/cav/AngM24}.

\paragraph{Partial-order methods for dynamic data race prediction.}

Starting with Lamport's Happens-Before relation~\cite{lamport1978time},
the literature offers a wide range of partial orders~\cite{Smaragdakis:2012:SPR:2103621.2103702,conf/pldi/KiniMV17,Mathur:2018:HFR:3288538:3276515,10.1145/3360605}.
As discussed in detail in \cref{sec:overview},
none of these partial orders is suitable for sound deadlock prediction,
motivating the introduction of the new TRW partial order.
Like in the case of data-race prediction, establishing soundness is a non-trivial task.
Our soundness proof shares  some similarities with the WCP soundness proof
in that we argue that either (a)~requests in a deadlock pattern can be placed
next to each other, or (b)~there must be a blocking cycle.
The details differ and are specific to TRW and the deadlock-prediction setting.





\section{Conclusions}
\label{sec:conclusions}

This paper considered the application of partial-order methods
to  eliminate false positives often reported by traditional lockset-based deadlock predictors.
Inspired by dynamic data-race prediction,
key ingredients for our refined deadlock
patterns are concurrency of lock acquisitions and absence of earlier blocking deadlock patterns.
We established soundness under
our novel TRW partial order and
completeness under the slightly weaker PWR.

We implemented our approach in an offline version of
the UNDEAD deadlock predictor.
Our experimental results demonstrate the effectiveness and
precision  of our approach, based on a sizeable set of
pre-recorded program traces.
Avenues for future work are improving the efficiency of the implementation
for large numbers of threads,
as well as exploring ways to allow for
unbounded critical sections without introducing false positives.






\clearpage

\appendix

\section{Proofs}
\label{sec:proofs}

\begin{lemma}
  \label{lem:DPLt}
  Given well-formed and well-nested trace $T$, $\DPLt[T]$ is a strict partial order.
\end{lemma}

\begin{proof}
  \leavevmode
  \begin{description}
    \item[Irreflexivity]
      Take any DP $A$ in $T$.
      Towards contradiction, suppose $A \DPLt[T] A$.
      Take any $(a,q) \in A$.
      There exists $(a',q') \in A$ such that $q \TrLt[T] q'$ and $a \in \AH[T](q')$.
      Since also $a \in \AH[T](q)$, then $\thd(q) = \thd(q') = \thd(a)$.
      But if $q \ne q'$, we must have $\thd(q) \ne \thd(q')$.
      Hence, $q = q'$, contradicting $q \TrLt[T] q'$.
    \item[Asymmetry]
      Take any DPs $A,B$ in $T$ such that $A \DPLt[T] B$.
      Towards contradiction, suppose $B \DPLt[T] A$.
      Take any $(a,q) \in A$.
      There exists $(a',q') \in B$ such that $q \TrLt[T] q'$ and $a \in \AH[T](q')$.
      Since also $a \in \AH[T](q)$, then $\thd(q) = \thd(q') = \thd(a)$.
      Similarly, there exists $(a'',q'') \in A$ such that $q' \TrLt[T] q''$ and $a' \in \AH[T](q'')$.
      Hence, $\thd(q) = \thd(q') = \thd(q'') = \thd(a')$.
      Similar to the previous subproof, then $q = q''$, contradicting $q \TrLt[T] q' \TrLt[T] q''$.
    \item[Transitivity]
      Take any DPs $A,B,C$ in $T$ such that $A \DPLt[T] B \DPLt[T] C$.
      Take any $(a,q) \in A$.
      There exists $(a',q') \in B$ such that $q \TrLt[T] q'$ and $a \in \AH[T](q')$.
      Similarly, there exists $(a'',q'') \in C$ such that $q' \TrLt[T] q''$ and $a' \in \AH[T](q'')$.
      Similar to the previous subproofs, we have $\thd(a) = \thd(q) = \thd(a') = \thd(q') = \thd(a'') = \thd(q'')$.
      By Condition~\cond{WF-Req}, $a \in \AH[T](a')$, so by well nestedness, $a \in \AH[T](q'')$.
      By transitivity, $q \TrLt[T] q''$.
      Hence, $A \DPLt[T] C$.
      \qedhere
  \end{description}
\end{proof}

\begin{lemma}
  \label{l:mustTRWAcq}
  Suppose given well-formed trace $T$, and event $e \in T$ and acquire $a \in T$.
  Assume that, in any $T' \in \crp(T)$ where $e$ is the last event in its thread, $a \TrLt[T'] e$.
  Then $a \TRWLt[T] e$.
\end{lemma}

\begin{proofsketch}
  Towards contradiction, suppose $a \not\TRWLt[T] e$.
  For they would be TRW ordered otherwise, then $a$ and $e$ cannot appear in the same thread or be conflicting.
  Hence, there must be a chain of critical sections forcing $a$ to precede $e$ in any of the assume $T' \in \crp(T)$.
  We apply induction on the length of this chain.
  The inductive case is a simple extension of the base case by applying the IH.

  In the base case, there are two critical sections on the same lock, and they contain events that force an ordering.
  W.l.o.g., assume $a$ appears before some $f$, both in a CS between $a_1$ and $r_1$ on lock $l$, where $f \TRWLt[T] g$ and $g$ appears after $e$, both in another CS between $a_2$ and $r_2$ on lock $l$: we have $r_1 \TRWLt[T] g$, but not $r_1 \TRWLt[T] e$.
  However, since $g$ appears after $e$, $g$ does not appear in any of the assumed $T' \in \crp(T)$, so having to place $f$ before $g$ cannot be the reason why $a$ appears before $e$ in every such $T'$.
  Hence, the ordering must be due to $g$ appearing before $e$, so $a \TRWLt[T] e$.
\end{proofsketch}

\begin{lemma}
  \label{l:mustTRWRel}
  Suppose given well-formed and well-nested trace $T$, and event $e \in T$ and acquire~$a \in T$ with matching release $r \in T$.
  Assume that, in any $T' \in \crp(T)$ where $e$ is the last event in its thread,~$r \notin T'$.
  If $a \TRWLt[T] e$, then $e \TRWLt[T] r$.
\end{lemma}

\begin{proofsketch}
  Towards contradiction, suppose $e \not\TRWLt[T] r$.
  We apply a similar induction as in the proof of \cref{l:mustTRWAcq}, and only discuss the base case.
  Here, $r$ appears before some $f$, both in a CS between $a_1$ and $r_1$ on lock $l$, where $g \TRWLt[T] f$ and $g$ appears after $e$, both in another CS between $a_2$ and $r_2$ on lock $l$: we have $e \TRWLt[T] f$, but not $e \TRWLt[T] r$.
  By well nestedness, also $a$ in the CS between $a_1$ and $r_1$.
  Since $a \TRWLt[T] e$, in any of the assumed $T' \in \crp(T)$, we must have $r_1 \TrLt[T'] a_2$.
  But then $r \TrLt[T'] e$, i.e., $e \in T'$: contradiction.
\end{proofsketch}

\begin{proof}
  [Proof of \cref{thm:TRWSound}]
  Assume $A = \DP{(\ACQ1,\REQ1),\ldots,(\ACQ{n},\REQ{n})}$.
  We need to show that there is $T' \in \crp(T)$ such that the requests in $A$ are the last events in their threads in $T'$, witnessing that $A$ is a predictable deadlock.
  To make full use of Condition~\cond{DP-TRW}, we prove the stronger property that there is such $T' \in \TRWcrp(T)$.
  We do so by contradiction, distinguishing cases on (1)~there being such $T'$ but none in $\TRWcrp(T)$ where we show contradiction of Condition~\cond{DP-TRW}, and (2)~there being no such $T' \in \crp(T)$ at all where we show contradiction of Condition~\cond{DP-Block}.

  To be precise, let $\mathbb{T} \subseteq \crp(T)$ such that every $T' \in \mathbb{T}$ contains all $q \in A$ as the last in their respective thread, and let $\mathbb{T}^\star = \mathbb{T} \cap \TRWcrp(T)$.
  We prove the theorem by proving the stronger property that $\mathbb{T}^\star$ is not empty.
  Towards contradiction, assume $\mathbb{T}^\star$ is empty.
  The rest of the analysis depends on whether (1)~$\mathbb{T}$ is not empty or (2)~$\mathbb{T}$ is empty.

  \begin{description}

    \item[Case~(1) contradicting Condition~\cond{DP-TRW}.]
      In case~(1), $\mathbb{T}$ is not empty.
      Since $\mathbb{T}^\star$ is empty, this means that $\mathbb{T} \setminus \TRWcrp(T)$ is not empty.
      That means that in any $T' \in \mathbb{T} \setminus \TRWcrp(T)$, there are $e,f \in T'$ such that $e \TRWLt[T] f$ but $f \TrLt[T'] e$, where $e \POLt[T] q_1$ and $f \POLt[T] q_2$ for some $q_1,q_2 \in A$.
      Hence, there must be a request that leaves an early critical section open, such that a later critical section is reordered before it, and the critical sections contain TRW-ordered events.
      To be precise, there are acquires $b_1,b_2 \in T$ on the same lock with respectively matching releases $r_1,r_2$, $e,f \in T$, and $q_1,q_2 \in A$, such that (a)~$b_1 \POLt[T] e \POLt[T] q_1 \POLt[T] r_1$, (b)~$f \POLt[T] r_2 \POLt[T] q_2$, and (c)~$e \TRWLt[T] f$.
      It is not necessarily the case that $f \in \CS[T](b_2)$ (when $f \POLt[T] b_2$), but this is no problem: there is $f' \in \CS[T](b_2)$ such that $f \POLeq[T] f'$ (i.e., possibly $f = f'$).
      Since $r_1 \notin T'$, we have $r_2 \TrLt[T'] b_1$ and hence $f \TrLt[T'] e$.
      We make the following inferences by definition of TRW:
      \begin{align*}
        e &\TRWLt[T] f \tag*{((c))} \\
        b_1 &\TRWLt[T] f \tag*{((a) and Conditions~\cond{TRW-PO/-Tr})} \\
        b_1 &\TRWLt[T] f' \tag*{(Condition~\cond{TRW-Tr})} \\
        r_1 &\TRWLt[T] f' \tag*{(Condition~\cond{TRW-Rel})} \\
        q_1 &\TRWLt[T] f' \tag*{((a) and Condition~\cond{TRW-Tr})} \\
        q_1 &\TRWLt[T] q_2 \tag*{((b) and Condition~\cond{TRW-Tr})}
      \end{align*}
      But then not $q_1 \TRWConc[T] q_2$, contradicting the Condition~\cond{DP-TRW} that DP $A$ satisfies.

    \item[Case~(2).]
      In case~(2), $\mathbb{T}$ is empty entirely.
      There may yet be candidate reordering, but they are not in $\crp(T)$ because of violation of Condition~\cond{CRP-LW}.
      If there are (subcase~(2.a)), we show a straightforward contradiction of Condition~\cond{DP-TRW}.
      Otherwise (subcase~(2.b)), the problem must be in reordering critical sections, and we show contradiction of Condition~\cond{DP-Block} (i.e., we show that there is an earlier cycle).

      \begin{description}

        \item[Subcase~(2.a) contradicting Condition~\cond{DP-TRW}.]
          In subcase~(2.a), all candidate reordering are rejected (i.e., not in $\crp(T)$) only because of violation of Condition~\cond{CRP-LW}.
          Either (i)~a read before some $q_1 \in A$ has its last write after some $q_2 \in A$, or (ii)~a read is assigned the wrong last write but its own last write is in the trace.
          In case~(i), trivially $q_1 \TRWLt[T] q_2$, again contradicting Condition~\cond{DP-TRW} that DP $A$ satisfies.
          The contradiction in case~(ii) follows analogously to that in case~(i).

        \item[Subcase~(2.b) contradicting Condition~\cond{DP-Block}.]
          In subcase~(2.b), the absence of the required reorderings is not only due to violation of Condition~\cond{CRP-LW}.
          Since satisfaction of all other conditions is trivial, the problem must be in violation of Conditions~\cond{CRP-WF/WF-Acq/WF-Rel}.
          That is, there are no candidate reorderings in $\crp(T)$, because there is a pair of acquires on the same lock that cannot be ordered in any way.

          To derive a contradiction, we show that there is a cycle $B$ such that $B \DPLt[T] A$.
          We apply induction on the size of $A$ (\ih1).
          In base case and inductive case, the idea is the same: we identify the problematic pair of acquires on the same lock, and show that their respective ordering is impossible due to another pair of acquires on another lock that cannot be ordered.
          In the base case, which we detail, the size of $A$ is two.
          We show that this new pair of acquires must also precede the two requests in $A$, and hence we can construct $B$ to contradiction Condition~\cond{DP-Block}.
          In the inductive case, the new pair of acquires may involve a third request in $A$; since there are less uninvolved requests left, the contradiction follows from \ih1 (the soundness of this induction follows from the finite size of $A$, so we ought to eventually encounter a ``new'' request we have seen before).

          To find the pair of problematic acquires, we apply induction on the distance between the problematic acquires and requests in $A$ (\ih2).
          In the inductive case, assume that both acquires have matching releases that precede requests in $A$: one of the critical sections must contain another acquire that causes the issue.
          Since this acquire is closer to a request in $A$, the thesis follows from \ih2.
          If one of the acquires has a matching release that succeeds a request in~$A$, we proceed as in the base case.

          In the base case, the acquires are as close to requests as possible, so there are no further problematics acquires.
          It must then be that one of the acquires has a matching release that must succeed a request in $A$; otherwise, there could not be an issue in ordering the acquires (remember that we ruled out write-read issues).

          Thus, we have acquires~$b_1,b'_1$ on the same lock~$l_1$ that, for a prefix of $T$ to be correctly reordered, have to occur before some $q_1,q_2 \in A$ and the release $r_1$ matching $b_1$ cannot occur before $q_1$ (w.l.o.g., assume $r_1 \in T$).
          Hence, by \cref{l:mustTRWAcq,l:mustTRWRel}, $b_1 \TRWLt[T] q_1 \TRWLt[T] r_1$, and $b'_1 \TRWLt[T] q_2$.
          Note that $q_1 \ne q_2$, for otherwise we would already have $b'_1 \TrLt[T] r'_1 \TrLt[T] b_1$, where release $r'_1$ matches $b'_1$: reordering would be unnecessary and hence there is no issue.

          By TRW boundedness, $\thd(b_1) = \thd(q_1)$.
          It follows that $b_1 \in \AH[T](q_1)$ and hence $l_1 \in \LH[T](q_1)$.
          Then, by Condition~\cond{DP-Guard}, $l_1 \notin \LH[T](q_2)$ so $b'_1 \notin \AH[T](q_2)$.
          Hence, $q_2 \not\TRWLt[T] r'_1$.

          This means that $q_2$ being the last in its thread cannot prevent $r'_1$ from being placed.
          However, by assuming that $\mathbb{T}$ is empty, doing so is impossible.
          Hence, there must be acquire $b_2$ that precedes $b'_1$ ($b_2 \TRWLt[T] b'_1$ by \cref{l:mustTRWAcq}) that cannot be moved to precede some $b'_2$ on the same lock~$l_2$, because the release $r_2$ matching $b_2$ succeeds $q_2$ ($q_2 \TRWLt[T] r_2$ by \cref{l:mustTRWRel}).
          Assuming $A = \DP{ (a_1,q_1) , (a_2,q_2) }$, following similar reasoning as above, it must then be that $b'_2 \TRWLt[T] q_1 \not\TRWLt[T] r'_2$, where release $r'_2$ matches~$b'_2$.

          There are many (well-nested) ways in which $b_2$ and $b'_2$ can be arranged with respect to $b'_1$ and $b_1$, respectively.
          Some are impossible by Condition~\cond{DP-Guard}, and all others (except one) can be ordered validly.
          Hence, the only arrangement that cannot be ordered applies.
          In this arrangement, we have $b_1 \TRWLt[T] b'_2 \TRWLt[T] r'_2 \TRWLt[T] q_1 \TRWLt[T] r_1$ and $b_2 \TRWLt[T] b'_1 \TRWLt[T] r'_1 \TRWLt[T] q_2 \TRWLt[T] r_2$.
          Hence, by TRW boundedness, $\thd(b_1) = \thd(b'_2) = \thd(q_1)$ and $\thd(b_2) = \thd(b'_1) = \thd(q_2)$.

          Let $q'_1$ and~$q'_2$ be the requests requesting $b'_2$ and $b'_1$, respectively, and let $B = \DP{ (b_1,q'_1) , (b_2,q'_2) }$.
          Clearly, $B$ is a cycle (satisfying Condition~\cond{DP-Cycle}) where $b_1 \in \AH[T](q_1)$ and $q'_1 \TrLt[T] q_1$, and $b_2 \in \AH[T](q_2)$ and $q'_2 \TrLt[T] q_2$.
          Hence, $B \DPLt[T] A$ (\cref{def:ltDP}), reaching our desired contradiction that DP $A$ satisfies Condition~\cond{DP-Block}.

          %
          %

          If there are more than two requests in $A$, then $b'_2$ may also be related to some $q_3 \in A$.
          Analogous to above, an analysis of arrangements of critical sections reveals that there must be yet another pair of unorderable acquires.
          Since there are only a finite number of requests in $A$, at some point we must encounter a request we have seen before.
          Since every acquire encountered is assumed to contribute to unorderability, it must then be that we would uncover a pair of unorderable acquires that we have seen before.
          This way, we reach the cycle of requires for us to build $B$ as above.
          \qedhere
      \end{description}
  \end{description}
\end{proof}

\begin{lemma}
  \label{l:PWRTr}
  Suppose given well-formed trace $T$, and events $e,f \in T$ such that $e \PWRLt[T] f$.
  Then, for any $T' \in \crp(T)$ where $f \in T'$, we have $e \TrLt[T'] f$.
\end{lemma}

\begin{proof}
  Take any $T' \in \crp(T)$ where $f \in T'$.
  We apply induction on the derivation of $e \PWRLt[T] f$.
  The only interesting case is Condition~\cond{PWR-Rel}:
  there are acquires $a_1,a_2 \in T$ on the same lock such that (i)~$e$ is the release matching $a_1$, (ii)~$a_1 \TrLt[T] a_2$, (iii)~$f \in \CS[T](a_2)$, and (iv)~$a_1 \PWRLt[T] f$.
  By the IH on (iv), $a_1 \TrLt[T'] f$.
  By (iii), $\thd(f) = \thd(a_2)$, so $a_2 \POLt[T] f$;
  by Condition~\cond{CRP-PO}, $a_2 \TrLt[T'] f$.
  Then, by Condition~\cond{WF-Acq}, $e \TrLt[T'] f$.
\end{proof}

\begin{proof}
  [Proof of \cref{t:PWRcomplete}]
  Assume $A = \DP{(\ACQ{1},\REQ{1}),\ldots,(\ACQ{n},\REQ{n})}$.
  Since $A$ is a predictable deadlock (\cref{d:deadlock}), there is $T' \in \crp(T)$ such that every $\REQ{i}$ is the last event in its thread.

  To prove Condition~\cond{DP-PWR}, towards contradiction assume there are $1 \leq i < j \leq n$ such that $\REQ{i} \PWRLt[T] \REQ{j}$.
  By \cref{l:PWRTr}, $\REQ{i} \TrLt[T'] \REQ{j}$.
  Since they are the last in their respective thread, we can construct $T''$ from $T'$ by (1)~moving all requests to be the last events in the trace overall, and then (2)~swapping $\REQ{i}$ and $\REQ{j}$.
  The resulting $T''$ is in $\crp(T)$, because the only Condition~\cond{WF-/CRP-$\ast$} that applies to requests is Condition~\cond{WF-Req} but that condition is about the event following a request of which there are none.
  Now we have $\REQ{j} \TrLt[T''] \REQ{i}$.
  This contradicts our earlier conclusion from \cref{l:PWRTr}.

  To prove Condition~\cond{DP-Block}, towards contradiction assume that there is a cycle $B \DPLt[T] A$.
  Assume $B = \DP{(\ACQ{1}'',\REQB{1}),(\ACQ{2}'',\REQB{2})}$; the argument extends to size $j \leq n$ by induction.
  By \cref{def:ltDP}, $q'_1 \ne q_1$ and $q'_2 \ne q_2$.
  Moreover, for $i \in \{1,2\}$, $\ACQ{i}'' \in \AH(\REQ{i})$ and $q'_i \TrLt[T] q_i$.

  Take any $i \in \{1,2\}$.
  Because $\REQB{i}$ and $\REQ{i}$ have an acquire held in common, by \cref{d:CS}, $\thd(\REQB{i}) = \thd(\REQ{i})$.
  Hence, $q'_i \POLt[T] q_i$, so $q'_i \TrLt[T'] q_i$.
  Moreover, $\ACQ{i}'' \TrLt[T'] \REQB{i}$.
  Let $\ACQ{i}'$ be the acquire requested by $\REQ{i}'$.
  By Condition~\cond{WF-Req} (\cref{d:WF}), $\REQB{i} \TrLt[T'] \ACQ{i}' \TrLt[T'] \REQ{i}$.
  Let $\REL{i}'$ and $\REL{i}''$ be the releases matching $\ACQ{i}'$ and $\ACQ{i}''$, respectively.
  Since $a''_i \in \AH(q_i)$, $r''_i \in T$ implies $q_i \POLt[T] r''_i$.
  Because $\REQ{i}$ is the last event in its thread in $T'$, it follows that $\REL{i}'' \notin T'$.

  Let $j \in \{1,2\} \setminus \{i\}$.
  Analogously, $a''_j \TrLt[T'] q'_j \TrLt[T'] a'_j \TrLt[T'] q_j$, $r''_j \notin T'$, and all are in the same thread.
  By Condition~\cond{DP-Cycle}, $a''_i,q'_j,a'_j$ and $a''_j,q'_i,a'_i$ are pairwise on the same lock.
  Hence, by well formedness of $T'$, $r'_i,r'_j \in T'$ with $\REL{i}' \TrLt[T'] \ACQ{j}''$ and $\REL{j}' \TrLt[T'] \ACQ{i}''$.
  But then we obtain the following:
  $\REL{i}' \TrLt[T'] \ACQ{j}'' \TrLt[T'] \ACQ{j}' \TrLt[T'] \REL{j}' \TrLt[T'] \ACQ{i}'' \TrLt[T'] \ACQ{i}' \TrLt[T'] \REL{i}'$.
  Clearly, this cyclic ordering is impossible: a contradiction.
\end{proof}

\section{TRW for Soundness}

\begin{figure}[t]
\bda{|l|l|l|}
\hline  T_{11} & \thread{1} & \thread{2}\\ \hline
\eventE{1}  & \lockE{\LKA}&\\
\eventE{2}  & \lockE{\LKB}&\\
\eventE{3}  & \lockE{\LKC}&\\
\eventE{4}  & \writeE{\VA}&\\
\eventE{5}  & \unlockE{\LKC}&\\
\eventE{6}  & \unlockE{\LKB}&\\
\eventE{7}  & \lockE{\LKD}&\\
\eventE{8}  & \reqLockE{\LKE}&\\
\eventE{9}  & \lockE{\LKE}&\\
\eventE{10}  & \unlockE{\LKE}&\\
\eventE{11}  & \unlockE{\LKD}&\\
\eventE{12}  & \unlockE{\LKA}&\\
\eventE{13}  & &\lockE{\LKB}\\
\eventE{14}  & &\writeE{\VA}\\
\eventE{15}  & &\unlockE{\LKB}\\
\eventE{16}  & &\lockE{\LKA}\\
\eventE{17}  & &\unlockE{\LKA}\\
\eventE{18}  & &\lockE{\LKC}\\
\eventE{19}  & &\readE{\VA}\\
\eventE{20}  & &\lockE{\LKE}\\
\eventE{21}  & &\reqLockE{\LKD}\\
\eventE{22}  & &\lockE{\LKD}\\
\eventE{23}  & &\unlockE{\LKD}\\
\eventE{24}  & &\unlockE{\LKE}\\
\eventE{25}  & &\unlockE{\LKC}\\

\hline \eda{}
  \caption{Ordering conflicting write-write memory operations is critical for soundness.}
  \label{f:w-w-trw}
\end{figure}

The examples in \cref{fig:ex_3,fig:ex_5} show that
ordering conflicting write-read and read-write memory operations is critical to achieve soundness.
Trace $T_{11}$ in \cref{f:w-w-trw} shows that ordering
conflicting write-write memory operations is important as well.
We find that $A = \{(e_7,e_8), (e_{20},e_{21})\}$
satisfies Conditions~\cond{DP-Guard/-Cycle/-Block}
but \emph{not} \cond{DP-TRW} because $e_8 \TRWLt e_{21}$.
Were we to ignore the write-write dependency $e_4 \TRWLt e_{14}$, then
$e_8$ and $e_{21}$ become unordered.
However, $A$ is not a predictable deadlock: there is no reordering
where $e_8$ and $e_{21}$ are the last events in their respective threads.
This shows that ordering \emph{all} conflicting memory operations is critical for soundness.

\section{Full Table ``Number of concrete lock dependencies''}

See \cref{tbl:details-full}.

\begin{table*}[p]
    \caption{
    \textbf{Extended version of \cref{tbl:details} containing all benchmarks.}
  }
  \label{tbl:details-full}
  {
  \small
  \setlength{\tabcolsep}{4.8pt} 
  \renewcommand{\arraystretch}{0.91} 
\begin{tabular}{|r||r|r|r||r|r||r|r|r|}
 \hline
1 & 2 & 3 & 4 & 5 & 6 & 7 & 8 & 9
 \\
 \hline
\multirow{2}{*}{\Benchmark} &
\multicolumn{3}{c||}{\UDTRWEvict} &
\multicolumn{2}{c||}{\UDTRW} &
\multicolumn{3}{c|}{\SPDOfflineUD} \\ \cline{2-4} \cline{5-6} \cline{7-9}
& \Cycles &  \Dependencies & \Time \ (\PhaseOne+\PhaseTwo)
& \Cycles & \Time \ (\PhaseOne+\PhaseTwo)
& \Cycles & \Dependencies & \Time \ (\PhaseOne+\PhaseTwo)
 \\
 \hline
Deadlock & 0 & 2 & 0 (0+0) & 0 & 0 (0+0) & 0 & 1 & 0 (0+0)  \\   \hline
NotADeadlock & 0 & 2 & 0 (0+0) & 0 & 0 (0+0) & 0 & 1 & 0 (0+0)  \\   \hline
Picklock & 1 & 6 & 0 (0+0) & 1 & 0 (0+0) & 1 & 5 & 0 (0+0)  \\   \hline
Bensalem & 1 & 6 & 0 (0+0) & 1 & 0 (0+0) & 1 & 6 & 0 (0+0)  \\   \hline
Transfer & 0 & 2 & 0 (0+0) & 0 & 0 (0+0) & 0 & 1 & 0 (0+0)  \\   \hline
Test-Dimminux & 2 & 7 & 0 (0+0) & 2 & 0 (0+0) & 2 & 7 & 0 (0+0)  \\   \hline
StringBuffer & 1 & 2 & 0 (0+0) & 1 & 0 (0+0) & 1 & 3 & 0 (0+0)  \\   \hline
Test-Calfuzzer & 1 & 5 & 0 (0+0) & 1 & 0 (0+0) & 1 & 5 & 0 (0+0)  \\   \hline
DiningPhil & 1 & 5 & 0 (0+0) & 1 & 0 (0+0) & 1 & 25 & 0 (0+0)  \\   \hline
HashTable & 0 & 1 & 0 (0+0) & 0 & 0 (0+0) & 0 & 42 & 0 (0+0)  \\   \hline
Account & 0 & 9 & 0 (0+0) & 0 & 0 (0+0) & 0 & 9 & 0 (0+0)  \\   \hline
Log4j2 & 0 & 3 & 0 (0+0) & 0 & 0 (0+0) & 0 & 3 & 0 (0+0)  \\   \hline
Dbcp1 & 1 & 5 & 0 (0+0) & 1 & 0 (0+0) & 1 & 5 & 0 (0+0)  \\   \hline
Dbcp2 & 0 & 10 & 0 (0+0) & 0 & 0 (0+0) & 0 & 17 & 0 (0+0)  \\   \hline
Derby2 & 0 & 0 & 0 (0+0) & 0 & 0 (0+0) & 0 & 0 & 0 (0+0)  \\   \hline
elevator & 0 & 0 & 1 (1+0) & 0 & 1 (1+0) & 0 & 0 & 1 (1+0)  \\   \hline
hedc & 0 & 4 & 1 (1+0) & 0 & 1 (1+0) & 0 & 4 & 1 (1+0)  \\   \hline
JDBCMySQL-1 & 0 & 23 & 1 (1+0) & 0 & 1 (1+0) & 0 & 3\mbox{K} & 1 (1+0)  \\   \hline
JDBCMySQL-2 & 0 & 24 & 1 (1+0) & 0 & 1 (1+0) & 0 & 3\mbox{K} & 1 (1+0)  \\   \hline
JDBCMySQL-3 & 1 & 28 & 1 (1+0) & 1 & 1 (1+0) & 1 & 3\mbox{K} & 1 (1+0)  \\   \hline
JDBCMySQL-4 & 1 & 30 & 1 (1+0) & 1 & 1 (1+0) & 1 & 3\mbox{K} & 1 (1+0)  \\   \hline
cache4j & 0 & 31 & 2 (2+0) & 0 & 2 (2+0) & 0 & 10\mbox{K} & 2 (2+0)  \\   \hline
ArrayList & 4 & 123 & 24 (24+0) & 4 & 24 (24+0) & 4 & 8\mbox{K} & 7 (7+0)  \\   \hline
IdentityHashMap & 1 & 42 & 25 (25+0) & 1 & 25 (25+0) & 1 & 79 & 8 (8+0)  \\   \hline
Stack & 3 & 2\mbox{K} & 44 (44+0) & 3 & 43 (43+0) & 3 & 95\mbox{K} & 8 (8+1)  \\   \hline
LinkedList & 4 & 118 & 27 (27+0) & 4 & 27 (27+0) & 4 & 7\mbox{K} & 9 (9+0)  \\   \hline
HashMap & 1 & 40 & 27 (27+0) & 1 & 27 (27+0) & 1 & 4\mbox{K} & 10 (10+0)  \\   \hline
WeakHashMap & 1 & 40 & 27 (27+0) & 1 & 28 (28+0) & 1 & 4\mbox{K} & 10 (10+0)  \\   \hline
Vector & 1 & 3 & 11 (11+0) & 1 & 10 (10+0) & 1 & 200\mbox{K} & 9 (9+0)  \\   \hline
LinkedHashMap & 1 & 40 & 29 (29+0) & 1 & 30 (30+0) & 1 & 4\mbox{K} & 12 (12+0)  \\   \hline
montecarlo & 0 & 0 & 23 (23+0) & 0 & 23 (23+0) & 0 & 0 & 23 (23+0)  \\   \hline
TreeMap & 1 & 40 & 41 (41+0) & 1 & 41 (41+0) & 1 & 4\mbox{K} & 24 (24+0)  \\   \hline
hsqldb & 0 & 2\mbox{K} & 57 (57+0) & 0 & 57 (56+0) & 0 & 125\mbox{K} & 56 (55+0)  \\   \hline
sunflow & 0 & 45 & 67 (67+0) & 0 & 68 (68+0) & 0 & 248 & 66 (66+0)  \\   \hline
jspider & 0 & 158 & 72 (72+0) & 0 & 73 (73+0) & 0 & 2\mbox{K} & 72 (72+0)  \\   \hline
tradesoap & 0 & 9\mbox{K} & 173 (166+6) & 0 & 174 (168+6) & 0 & 40\mbox{K} & 163 (157+6)  \\   \hline
tradebeans & 0 & 9\mbox{K} & 177 (171+6) & 0 & 177 (171+6) & 0 & 40\mbox{K} & 166 (160+6)  \\   \hline
TestPerf & 0 & 0 & 196 (196+0) & 0 & 199 (199+0) & 0 & 0 & 192 (192+0)  \\   \hline
Groovy2 & 0 & 11\mbox{K} & 380 (379+1) & 0 & 386 (385+1) & 0 & 29\mbox{K} & 372 (371+1)  \\   \hline
tsp & 0 & 0 & 992 (992+0) & 0 & 989 (989+0) & 0 & 0 & 997 (997+0)  \\   \hline
lusearch & 0 & 87 & 719 (719+0) & 0 & 723 (723+0) & 0 & 41\mbox{K} & 719 (719+0)  \\   \hline
biojava & 0 & 89 & 656 (656+0) & 0 & 656 (656+0) & 0 & 545 & 661 (661+0)  \\   \hline
graphchi & 0 & 50 & 756 (756+0) & 0 & 756 (756+0) & 0 & 82 & 754 (754+0)  \\   \hline
 \hline
\Sum & 27 & 33\mbox{K} & 4533 (4518+14) & 27 & 4545 (4531+14) & 27 & 626\mbox{K} & 4348 (4333+15)
 \\  \hline   \end{tabular}
}
\end{table*}

\section{TRW-Boundedness Check}

\begin{algorithm*}[t!]
  \caption{\cref{alg:lock-deps} extended with TRW-boundedness check.}
  \label{alg:lock-deps-guard-check}

  {\small
    \begin{algorithmic}[1]
      \Function{computeTRWLockDeps}{$T$}
      \label{ln:g-cLDs}
        \State $\forall t \colon \threadVC{t} = [\bar{0}]; \incC{\threadVC{t}}{t}$
        \Comment{Vector clock $\threadVC{t}$ of thread $t$}
        \label{ln:g-thvc}
        \State $\forall x \colon \lastWriteVC{x} = [\bar{0}]; \lastReadVC{x} = [\bar{0}]$
        \Comment{Vector clocks $\lastWriteVC{x},\lastReadVC{x}$ of most recent $\writeE{x},\readE{x}$}
        \label{ln:g-lwlr}
        \State $\forall l \colon \acqVC{l} = [\bar{0}]$
        \Comment{Vector clock $\acqVC{l}$ of most recent $\acqE{l}$}
        \label{ln:g-acqv}
        \State $\forall l \colon \Hist{l} = []$
        \Comment{History $\Hist{l}$ of acquire-release pairs $(\Vacq,\Vrel)$ for lock $l$}
        \label{ln:g-hist}
        \State $\forall t \colon \AcqHeld(t) = []$
        \Comment{Sequence $\AcqHeld(t)$ of acquires held by thread $t$}
        \label{ln:g-acqhd}
        \State $\LDMapSym = \emptyset$
        \Comment{Map  with keys $(t,l,ls)$,
                 list values with elements $(i,V,\{a_1,\ldots,a_n\})$}
        \label{ln:g-ld}
        \State $\GlobalLS = \emptyset$
        \Comment{Global lockset across all threads}
        \label{ln:g-gls}
        \State $\GCMapSym = \emptyset$
        \Comment{Map  with keys $(t,l)$ and vector clock values}
        \label{ln:g-gmap}
        \ForDo {$e$ in $T$} {\Call{process}{$e$}}
        \State \Return \LDt
      \EndFunction
      \algstore{cld}
    \end{algorithmic}

    \begin{minipage}[t]{.54\textwidth}
      \begin{algorithmic}[1]
        \algrestore{cld}
        \Procedure{process}{$(\alpha,t,acq(l))$}
          \label{ln:g-procAcq}
          \If {$\AcqHeld(t) \not = []$}
          \label{ln:g-acq-non-empty}
          \State $ls = \{ l' \mid \acqE{l'} \in \AcqHeld(t) \}$
          \label{ln:g-ld-ls}
          \State $\LDMap{t}{l}{ls}.pushBack(\alpha,\threadVC{t},\AcqHeld(t))$
            \label{ln:g-ld-add}
          \EndIf
          \State $\AcqHeld(t) = \AcqHeld(t) \cup \{ (\alpha,t,acq(l)) \}$
          \label{ln:g-acq-push}
          \State $\GlobalLS = \GlobalLS \cup \{ l \}$
          \label{ln:g-global-lockset-push}
          \For{$l \in \GlobalLS, acq(l) \not\in \AcqHeld(t), \acqVC{l} < \threadVC{t}$}
          \label{ln:g-guard-candidate}
          \If{$\GCMapSym[l][t].empty()$}
          \State $\GCMapSym[l][t] = \threadVC{t}$
          \label{ln:g-first-entry}
          \EndIf
          \EndFor
          \State $\threadVC{t} = \Call{syncCS}{\threadVC{t},\AcqHeld(t)}$
          \label{ln:g-acq-sync}
          \State $\acqVC{l} = \threadVC{t}$
          \label{ln:g-acq-store}
          \State $\incC{\threadVC{t}}{t}$
          \label{ln:g-acq-vc-inc}
        \EndProcedure
        \algstore{eacq}
      \end{algorithmic}

      \begin{algorithmic}[1]
        \algrestore{eacq}
        \Procedure{process}{$(\dontCare,t,rel(l))$}
        \label{ln:g-procRel}
        \For{$t', \neg \GCMapSym[l][t'].empty()$}
        \If{$\GCMapSym[l][t'] < \threadVC{t}$}
        \label{ln:g-guard-check}
        \State ``TRW-boundedness violated''
        \EndIf
        \EndFor
        \State $\GCMapSym[l] = \emptyset$
        \label{ln:g-init-map}
          \State $\AcqHeld(t) = \{\dontCare,\dontCare,acq(l')) \in \AcqHeld(t) \mid l' \not= l \}$
          \label{ln:g-acq-pop}
          \State $\GlobalLS = \GlobalLS - \{ l \}$
          \label{ln:g-global-lockset-pop}
          \State $\Hist{l} = \Hist{l} \cup \{(\Acq{l}, \threadVC{t})\}$
          \label{ln:g-hist-add}
          \State $\incC{\threadVC{t}}{t}$
          \label{ln:g-rel-vc-inc}
        \EndProcedure
        \algstore{erel}
      \end{algorithmic}


    \end{minipage}%
    \hfill%
    \begin{minipage}[t]{.46\textwidth}
      \begin{algorithmic}[1]
        \algrestore{erel}
        \Procedure{process}{$(\dontCare,t,wr(x))$}
          \label{ln:g-procWr}
          \State $\threadVC{t} = \threadVC{t} \sqcup \lastWriteVC{x}$
          \label{ln:g-ww-sync}
          \State $\threadVC{t} = \threadVC{t} \sqcup \lastReadVC{x}$
          \label{ln:g-rw-sync}
          \State $\threadVC{t} = \Call{syncCS}{\threadVC{t},\AcqHeld(t)}$
          \label{ln:g-csw-sync}
          \State $\lastWriteVC{x} = \threadVC{t}$
          \label{ln:g-lw-store}
          \State $\incC{\threadVC{t}}{t}$
          \label{ln:g-wr-vc-inc}
        \EndProcedure
        \algstore{ewr}
      \end{algorithmic}

      \begin{algorithmic}[1]
        \algrestore{ewr}
        \Procedure{process}{$(\dontCare,t,rd(x))$}
          \label{ln:g-procRd}
          \State $\threadVC{t} = \threadVC{t} \sqcup \lastWriteVC{x}$
          \label{ln:g-wr-sync}
          \State $\threadVC{t} = \Call{syncCS}{\threadVC{t},\AcqHeld(t)}$
          \label{ln:g-csr-sync}
            \State $\lastReadVC{x} = \threadVC{t}$ 
          \label{ln:g-lr-store}
          \State $\incC{\threadVC{t}}{t}$
          \label{ln:g-rd-vc-inc}
        \EndProcedure
        \algstore{erd}
      \end{algorithmic}

      \begin{algorithmic}[1]
        \algrestore{erd}
        \Function{syncCS}{$V, A$}
          \label{ln:g-syncCS}
          \For {$\acqE{l} \in A , (\Vacq, \Vrel) \in \Hist{l}$}
          \label{ln:g-sync-hist-entry}
            \IfThen {$\Vacq < V$} {$V = V \sqcup \Vrel$}
            \label{ln:g-ro-sync}
          \EndFor
        \State \Return V
        \EndFunction
      \end{algorithmic}
    \end{minipage}
  }
\end{algorithm*}

In terms of vector clocks, TWR boundedness is violated
if there is a pair of acquire-release vector clocks $(V_{acq}, V_{rel})$
and a request vector clock $V'$
such that $V_{acq} < V' < V_{rel}$ and the request results from
a thread other than the acquire-release pair.

\cref{alg:lock-deps-guard-check} integrates
the TRW-boundedness check in Phase~(1).
We additionally maintain the set $\GlobalLS$ of locks
held across all threads and
use map $\GCMapSym$ to check for potential `guard' locks
that violate the TRW-boundedness condition.

For every acquire event, we check for a potential guard lock that
has been acquired by some other thread and that is TRW ordered
(\cref{ln:g-guard-candidate} in \cref{alg:lock-deps-guard-check}).
For each potential guard lock~$l$ and thread~$t$, we record
the vector clock $\threadVC{t}$ (corresponding to the vector clock of the request).
There may be multiple requests with vector clocks $V_1,...,V_n$
that satisfy \cref{ln:g-guard-candidate}.
For every such $V_i$, we check whether $V_i < V_{rel(l)}$
where $V_{rel(l)}$ is the vector clock of the release event
that belongs to the acquire event $\acqVC{l}$.
If there is such $V_i$,
then TRW boundedness is violated.
Because $V_1 < \ldots < V_n$, it sufficies to only
check the ``earlier'' vector clock $V_1$.
Storage of $V_1$ is done on \cref{ln:g-first-entry}.
The check $V_1 < V_{rel(l)}$ is carried out on \cref{ln:g-guard-check}.

Access to the map $\GCMapSym$ takes constant time, and
the set of guard locks and threads to consider can be treated
as a constant.
Hence, the main cost factors of our TRW-boundedness check are vector-clock operations.
The number of vector-clock operations remains linear in the number of events.
We conclude that the time complexity of Phase~(1) is unaffected.
Measurements show that the TRW-boundedness check
does not cause any additional overhead.
See \cref{tbl:guard}, where \UDTRWGuardCheck
is the variant that includes the TRW-boundedness check described
above.

\begin{table*}[t]
    \caption{
      \textbf{Impact of TWR-boundedness check.}
          Columns~2--5 contain the number of events, of threads,
    of memory locations, and of locks, respectively.
    Columns~6--9 contain the number of deadlocks reported and running time for each candidate.
      Times are rounded to the nearest second, and reported for Phases~(1) and~(2) separately.
  }
    \label{tbl:guard}
  {
  \small
  \setlength{\tabcolsep}{4.7pt} 
  \renewcommand{\arraystretch}{0.91} 
\begin{tabular}{|r|r|r|r|r||r|r||r|r|}
 \hline
1 & 2 & 3 & 4 & 5 & 6 & 7 & 8 & 9
 \\
 \hline
\multirow{2}{*}{\Benchmark} &
\multirow{2}{*}{\EE} &
\multirow{2}{*}{\TT} &
\multirow{2}{*}{\MM} &
\multirow{2}{*}{\LL} &
\multicolumn{2}{c||}{\UDTRW} &
\multicolumn{2}{c|}{\UDTRWGuardCheck} \\ \cline{6-7} \cline{8-9}
 & & & &
& \Cycles &  \Time \ (\PhaseOne+\PhaseTwo)
& \Cycles & \Time \ (\PhaseOne+\PhaseTwo)
 \\
 \hline
Deadlock & 28 & 3 & 3 & 2 & 0 & 0 (0+0) & 0 & 0 (0+0)  \\   \hline
NotADeadlock & 42 & 3 & 3 & 4 & 0 & 0 (0+0) & 0 & 0 (0+0)  \\   \hline
Picklock & 46 & 3 & 5 & 5 & 1 & 0 (0+0) & 1 & 0 (0+0)  \\   \hline
Bensalem & 45 & 4 & 4 & 4 & 1 & 0 (0+0) & 1 & 0 (0+0)  \\   \hline
Transfer & 56 & 3 & 10 & 3 & 0 & 0 (0+0) & 0 & 0 (0+0)  \\   \hline
Test-Dimminux & 50 & 3 & 8 & 6 & 2 & 0 (0+0) & 2 & 0 (0+0)  \\   \hline
StringBuffer & 57 & 3 & 13 & 3 & 1 & 0 (0+0) & 1 & 0 (0+0)  \\   \hline
Test-Calfuzzer & 126 & 5 & 15 & 5 & 1 & 0 (0+0) & 1 & 0 (0+0)  \\   \hline
DiningPhil & 210 & 6 & 20 & 5 & 1 & 0 (0+0) & 1 & 0 (0+0)  \\   \hline
HashTable & 222 & 3 & 4 & 2 & 0 & 0 (0+0) & 0 & 0 (0+0)  \\   \hline
Account & 617 & 6 & 46 & 6 & 0 & 0 (0+0) & 0 & 0 (0+0)  \\   \hline
Log4j2 & 1\mbox{K} & 4 & 333 & 10 & 0 & 0 (0+0) & 0 & 0 (0+0)  \\   \hline
Dbcp1 & 2\mbox{K} & 3 & 767 & 4 & 1 & 0 (0+0) & 1 & 0 (0+0)  \\   \hline
Dbcp2 & 2\mbox{K} & 3 & 591 & 9 & 0 & 0 (0+0) & 0 & 0 (0+0)  \\   \hline
Derby2 & 3\mbox{K} & 3 & 1\mbox{K} & 3 & 0 & 0 (0+0) & 0 & 0 (0+0)  \\   \hline
elevator & 222\mbox{K} & 5 & 726 & 51 & 0 & 1 (1+0) & 0 & 1 (1+0)  \\   \hline
hedc & 410\mbox{K} & 7 & 109\mbox{K} & 7 & 0 & 1 (1+0) & 0 & 1 (1+0)  \\   \hline
JDBCMySQL-1 & 436\mbox{K} & 3 & 73\mbox{K} & 10 & 0 & 1 (1+0) & 0 & 1 (1+0)  \\   \hline
JDBCMySQL-2 & 436\mbox{K} & 3 & 73\mbox{K} & 10 & 0 & 1 (1+0) & 0 & 1 (1+0)  \\   \hline
JDBCMySQL-3 & 436\mbox{K} & 3 & 73\mbox{K} & 12 & 1 & 1 (1+0) & 1 & 1 (1+0)  \\   \hline
JDBCMySQL-4 & 437\mbox{K} & 3 & 73\mbox{K} & 13 & 1 & 1 (1+0) & 1 & 1 (1+0)  \\   \hline
cache4j & 758\mbox{K} & 2 & 46\mbox{K} & 19 & 0 & 2 (2+0) & 0 & 2 (2+0)  \\   \hline
ArrayList & 3\mbox{M} & 801 & 121\mbox{K} & 801 & 4 & 24 (24+0) & 4 & 24 (24+0)  \\   \hline
IdentityHashMap & 3\mbox{M} & 801 & 496\mbox{K} & 801 & 1 & 25 (25+0) & 1 & 25 (25+0)  \\   \hline
Stack & 3\mbox{M} & 801 & 118\mbox{K} & 2\mbox{K} & 3 & 43 (43+0) & 3 & 44 (43+0)  \\   \hline
LinkedList & 3\mbox{M} & 801 & 290\mbox{K} & 801 & 4 & 27 (27+0) & 4 & 27 (27+0)  \\   \hline
HashMap & 3\mbox{M} & 801 & 555\mbox{K} & 801 & 1 & 27 (27+0) & 1 & 27 (27+0)  \\   \hline
WeakHashMap & 3\mbox{M} & 801 & 540\mbox{K} & 801 & 1 & 28 (28+0) & 1 & 27 (27+0)  \\   \hline
Vector & 3\mbox{M} & 3 & 14 & 3 & 1 & 10 (10+0) & 1 & 10 (10+0)  \\   \hline
LinkedHashMap & 4\mbox{M} & 801 & 617\mbox{K} & 801 & 1 & 30 (30+0) & 1 & 29 (29+0)  \\   \hline
montecarlo & 8\mbox{M} & 3 & 850\mbox{K} & 2 & 0 & 23 (23+0) & 0 & 23 (23+0)  \\   \hline
TreeMap & 9\mbox{M} & 801 & 493\mbox{K} & 801 & 1 & 41 (41+0) & 1 & 41 (41+0)  \\   \hline
hsqldb & 20\mbox{M} & 46 & 945\mbox{K} & 402 & 0 & 57 (56+0) & 0 & 56 (56+0)  \\   \hline
sunflow & 21\mbox{M} & 15 & 2\mbox{M} & 11 & 0 & 68 (68+0) & 0 & 67 (67+0)  \\   \hline
jspider & 22\mbox{M} & 11 & 5\mbox{M} & 14 & 0 & 73 (73+0) & 0 & 72 (72+0)  \\   \hline
tradesoap & 42\mbox{M} & 236 & 3\mbox{M} & 6\mbox{K} & 0 & 174 (168+6) & 0 & 171 (165+6)  \\   \hline
tradebeans & 42\mbox{M} & 236 & 3\mbox{M} & 6\mbox{K} & 0 & 177 (171+6) & 0 & 175 (169+6)  \\   \hline
TestPerf & 80\mbox{M} & 50 & 598 & 8 & 0 & 199 (199+0) & 0 & 195 (195+0)  \\   \hline
Groovy2 & 120\mbox{M} & 13 & 13\mbox{M} & 10\mbox{K} & 0 & 386 (385+1) & 0 & 384 (383+1)  \\   \hline
tsp & 307\mbox{M} & 10 & 181\mbox{K} & 2 & 0 & 989 (989+0) & 0 & 991 (991+0)  \\   \hline
lusearch & 217\mbox{M} & 10 & 5\mbox{M} & 118 & 0 & 723 (723+0) & 0 & 718 (718+0)  \\   \hline
biojava & 221\mbox{M} & 6 & 121\mbox{K} & 78 & 0 & 656 (656+0) & 0 & 659 (659+0)  \\   \hline
graphchi & 216\mbox{M} & 20 & 25\mbox{M} & 60 & 0 & 756 (756+0) & 0 & 775 (775+0)  \\   \hline
 \hline
\Sum & 1354\mbox{M} & 7\mbox{K} & 61\mbox{M} & 30\mbox{K} & 27 & 4545 (4531+14) & 27 & 4546 (4532+14)
\\  \hline   \end{tabular}
}
\end{table*}

\section{Preliminary Access to our Implementation and Experiment Setup}

We plan to our implementation as an artifact; preliminary access is available at \url{https://osf.io/ku9fx/files/osfstorage?view_only=b7f53d3110894fe39ad1520ed0fed4ec} (anonymized link).
Reviewers are welcome to confirm our results by downloading our implementation and benchmark traces.
The repository contains a `README.md' with build and execution instructions.
Note that the repository formats MarkDown poorly, so we strongly advise to download `README.md' and view it locally.


\end{document}


\maketitle

\begin{abstract}
This corrigendum revisits the soundness claims made in \emph{Partial Orders for Precise and Efficient Dynamic Deadlock Prediction}. We show that the proposed partial-order–based prediction is unsound by constructing concrete counterexamples that violate the claimed guarantees. In contrast, the completeness results and performance evaluation reported in the original paper remain valid. We analyze the root causes of the unsoundness and outline possible directions toward restoring soundness in future work.
\end{abstract}

\section{Introduction}

After dissemination of the original manuscript~\cite{DBLP:journals/corr/abs-2502-20070}, a closer examination of the soundness proof revealed flaws that cannot be repaired within the proposed framework. This note documents these issues to prevent propagation of incorrect claims and to clarify which results remain valid.
This note is authored independently and focuses exclusively on issues discovered after the dissemination of the original manuscript.

\paragraph{Outline}

\cref{sec:background} repeats all relevant background material.
\cref{sec:soundness-claim} repeats the soundness claim
and \cref{sec:counterexample} presents a counterexample
that disproves the claim.
\cref{sec:analysis} provides a more detailed analysis of the unsoundness issue.
\cref{sec:unaffected} lists results stated in the original manuscript~\cite{DBLP:journals/corr/abs-2502-20070} that are not affected by the unsoundness issue.
\cref{sec:repair-dir} considers possible ways how to fix the (un)soundness issue.

\sectionfont{\titlecap}
\section{Background}
\label{sec:background}

\subsection{Events and traces}

We use $l_1,l_2,\ldots$ and $x,y,z$ for lock and shared variables, respectively.

\begin{definition}[Events and Traces]
  \label{def:run-time-traces-events}
  \bda{@{}l@{}c@{}ll@{\hskip3em}l@{}c@{}llr@{}}
    \alpha, \beta,\delta & {} ::= {} & 1 \mid 2 \mid \ldots & \mbox{(unique event ids)}
    & e & {} ::= {} & (\alpha, t,op) & \mbox{(events)}
    \\ t,s,u & {} ::= {} & \THD1 \mid \THD2 \mid \ldots & \mbox{(thread ids)}
    & T & {} ::= {} & [] \mid e : T   & \mbox{(traces)}
    \\ op & {} ::= {} & \readE{x} \mid \writeE{x}
    \mid \reqE{l} \mid \acqE{l} \mid \relE{l}
    \span\span\span\span
    & \mbox{(operations)}
  \eda
\end{definition}

\noindent
A trace $T$ is a list of events reflecting a single execution of a
concurrent program under the sequential consistency memory
model~\cite{Adve:1996:SMC:619013.620590}.
We write $[o_1,\dots,o_n]$
for a list of objects  
and use the operator ``$\conc$'' for list concatenation.

An event $e$ is represented by a triple $(\alpha, t,op)$
where
$\alpha$ is a unique event identifier,
$op$ is an operation, and
$t$ is the thread id in which the operation took place.
The main thread has thread id~$\THD1$.
The unique event identifier unambiguously identifies events under trace reordering.

The operations $\readE{x}$ and $\writeE{x}$ denote reading of and writing to a shared variable~$x$, respectively.
Operations $\acqE{l}$ and $\relE{l}$ denote acquiring and releasing a lock $l$, respectively.
We also include lock requests $\reqE{l}$ to denote the (possibly unfulfillable) attempt at acquiring a lock $l$; this allows for smoother definitions and is crucial for completeness.

We often omit identifier and/or thread when denoting events, and omit parentheses if only the operation remains, writing $e = (t,op)$ or $e = op$ instead of $e = (\alpha,t,op)$.
Notation $\thd(e)$ extracts the thread id from an event.

We write $e \in T$ to indicate that $T = [e_1, \dots, e_n]$ and $e = e_k$ for $1 \le k \le n$, defining also $\posP{T}{e} = k$.
The set of events in a trace is then $\evts(T) = \{ e \mid e \in T \}$, and the set of thread ids in a trace is $\thds(T) = \{ \thd(e) \mid e \in T \}$.
For trace $T$ and events $e, f \in \evts(T)$, we define trace order: $e \TrLt[T] f$ if $\posP{T}{e} < \posP{T}{f}$.
We then define program order $\POLt[T]$ as the restriction of $\TrLt[T]$ to events with the same thread id.

\subsection{Well-formedness and correct reorderings}

Traces must be well-formed, following the standard sequential consistency conditions for concurrent objects \cite{Adve:1996:SMC:619013.620590,Said:2011:GDR:1986308.1986334,Huang:2014:MSP:2666356.2594315}.

\begin{definition}[Well-Formedness]
\label{d:WF}
  A trace $T$ is \emph{well-formed} if all the following conditions are satisfied:
  \begin{description}
    \item[\cond{WF-Acq}]
      For every $\evtAcc = (t,\acqE{l}), \evtAccA = (s,\acqE{l}) \in T$ where $\evtAcc \TrLt[T] \evtAccA$, there exists $\evtRel = (t,\relE{l}) \in T$ such that $\evtAcc \TrLt[T] \evtAccA$.
      We say that $a$ and $r$ \emph{match}.

   \item[\cond{WF-Rel}]
      For every $\evtRel = (t,\relE{l}) \in T$, there exists $\evtAcc = (t,\acqE{l}) \in T$ such that $a \TrLt[T] r$ and there is no $\evtRelA = (s,\relE{l}) \in T$ with $\evtAcc \TrLt[T] \evtRelA \TrLt[T] \evtRel$.
      We say that $a$ and $r$ \emph{match}.

    \item[\cond{WF-Req}]\label{cond:lock-3}
      For every $\ACQ{} = (t,\acqE{l}) \in T$, there exists $\REQ{} = (t,\reqE{l}) \in T$ such that $\REQ{} \TrLt[T] \ACQ{}$, and for every $\REQ{} = (t,\reqE{l}) \in T$ and $e = (t,op) \in T$ such that $\REQ{} \TrLt[T] e$, if there is no $f = (t,op')$ such that $\REQ{} \TrLt[T] f \TrLt[T] e$ then $e = (t,\acqE{l})$.
      We say that $\REQ{}$ \emph{requests} $e$.

  \end{description}
\end{definition}

\noindent
Condition~\cond{WF-Acq} states that a previously acquired lock can only be acquired after it has been released.
Similarly, Condition~\cond{WF-Rel} states that a lock can only be released after is has been acquired but not yet released.
Note that these conditions require matching acquires and releases to occur in the same thread.
Condition~\cond{WF-Req} states that all lock acquires must have been requested in the same thread immediately before.
We often omit requests from example traces, but we assume that they precede each acquire implicitly.
At the end of the trace, acquired locks do not have to be released and requests do not have to be fulfilled.

A trace represents only one possible interleaving of events: there can be as many interleavings as there are permutations of the original trace.
However, not all such reorderings are feasible in the sense of being reproducible by executing the program with a different schedule.
In addition to well-formedness, a reordering must guarantee that (a)~program order and (b)~last writes are maintained.
A reordering maintains program order if, in any thread, the order of events is unchanged.
For guarantee~(b), note that every read observes some write on the same variable: the last preceding such write.
Last writes are then maintained if the write observed by any read is unchanged.
Guarantee~(b) is particularly important to ensure that the control flow of the traced program is unaffected by the reordering, e.g., when the read is used in a conditional statement.

Reorderings do not have to run to completion, so we consider reordered \emph{prefixes} of traces.

\begin{definition}[Correctly Reordered Prefix]
  \label{def:crp}
  The \emph{projection} of $T$ onto thread $t$, denoted $\proj(t,T)$, restricts $T$ to events $e$ with $\thd(e) = t$.
  That is,
  $\evtAA \in proj(t,T)$ if and only if $e \in T$ and $\thd(e) = t$, and
  $\evtAA \TrLt[proj(t,T)] \evtBB$ implies $\evtAA \TrLt[T] \evtBB$.

  Take $\evtAA = \readE{x},\evtBB = \writeE{x} \in T$.
  We say that $\evtBB$ is the \emph{last write} of $\evtAA$ w.r.t.\ $T$ if $\evtBB$ appears before $\evtAA$ with no other write on $x$ in between.
  That is,
  $\evtBB \TrLt[T] \evtAA$, and
  there is no $\evtCC = \writeE{x} \in T$ such that $\evtBB \TrLt[T] \evtCC \TrLt[T] \evtAA$.

  Trace $T'$ is a \emph{correctly reordered prefix} of $T$ if the following
  conditions are satisfied:
  \begin{description}
    \item[\cond{CRP-WF}] $T'$ is well-formed and $\evts(T') \subseteq \evts(T)$.
    \item[\cond{CRP-PO}] For every $t \in \thds(T')$, $\proj(t,T')$ prefixes $\proj(t,T)$.
    \item[\cond{CRP-LW}] For every $\readE{x} \in T'$ with last write $f$ w.r.t.\ $T$, it must be $f$ is also the last write for $e$ w.r.t.\ $T'$.\footnote{Unique event identifiers are crucial for this condition.}
  \end{description}
  We write $\crp(T)$ to denote the set of correctly reordered prefixes of $T$.
\end{definition}

\subsection{Resource Deadlocks}
\label{s:prelim:deadlock}

We consider a standard definition of (resource) deadlock, where multiple threads are simultaneously requesting to acquire a lock, but every such lock is held by another thread: none of the requests can be fulfilled, and so the trace is stuck.
Our goal is to predict that a trace can be rescheduled to deadlock, and so we identify traces with \emph{predictable} deadlocks in terms of correctly reordered prefixes with deadlocks.
We first identify these deadlocks using \emph{deadlock patterns}.

Deadlock patterns codify cycles of requests that cannot be fulfilled because the requested lock is held by the next request.
Hence, we first define what it means for an event to hold a lock, in terms of \emph{critical sections} of events bordered by lock acquires and releases.

\begin{definition}[Critical Sections]
  \label{d:CS}
  Given $a = (t,\acqE{l}) \in T$, we define the \emph{critical section acquired by $a$} as $\CS[T](a) \subseteq \evts(T)$, where $e \in \CS[T](a)$ if all the following conditions hold:
  \begin{description}
    \item[\cond{CS-PO}]
      $e = a$ or $a \POLt[T] e$.
    \item[\cond{CS-Held}]
      There is no $r = \relE{l} \in T$ that matches $a$ with $a \POLt[T] r \POLt[T] e$.
  \end{description}

  The set of \emph{acquires held} by $e \in T$ is defined as $\AH[T](e) = \{ a \in T \mid e \in \CS[T](a) \}$.
  Accordingly, the set of \emph{locks held} is defined as $\LH[T](e) = \{ l \mid \exists \acqE{l} \in \AH[T](e) \}$.
\end{definition}

\noindent
Thus, an event is in a critical section if it is or appears after the acquire in the same thread (Condition~\cond{CS-PO}), as long as the critical section has not yet been closed (Condition~\cond{CS-Held}).
Note that Condition~\cond{CS-Held} considers a lock held if it is never released at the end of the trace.

\begin{definition}[Deadlock Patterns (DPs)]
  \label{d:dp}
  Given trace $T$, let $A = \DP{(\ACQ1,\REQ1),\ldots,(\ACQ{n},\REQ{n})} \subseteq \evts(T) \times \evts(T)$ for $n \ge 2$, where every $(\ACQ{i},\REQ{i}) \in A$ is an acquire-request tuple, and $\thd(\REQ{i}) \ne \thd(\REQ{j})$ for every $1 \le i < j \le n$.
  We define the following properties for $A$:
  \begin{description}
    \item[\cond{DP-Cycle}]
      For every $1 \le i \le n$, $q_i = \reqE{l} \in \CS[T](a_i)$ and $a_{(i \mod n) + 1} = \acqE{l}$.
    \item[\cond{DP-Guard}]
      $\LH[T](\ACQ{i}) \cap \LH[T](\ACQ{j}) = \emptyset$ for $1 \leq i < j \leq n$.
  \end{description}
  Condition~\cond{DP-Cycle} identifies $A$ as a \emph{cycle in~$T$}.
  If $A$ is a cycle in~$T$ that satisfies Condition~\cond{DP-Guard}, then $A$ is a \emph{deadlock pattern (DP)} in~$T$.
\end{definition}

\noindent
Thus, we consider sets of acquire-request pairs, where the request holds the acquire.
Condition~\cond{DP-Cycle} ensures that the set forms a cycle of requests for locks held by other requests.
DPs satisfy an additional condition to eliminate false alarms: Condition~\cond{DP-Guard} ensures that no two requests hold a common lock, which would make it impossible to schedule the deadlock.

Finally, we define predictable deadlocks in terms of DPs.

\begin{definition}[Predictable Deadlock]
  \label{d:deadlock}
  Let $A$ be a DP in well-formed trace $T$.
  We say $A$ is a \emph{predictable deadlock} if there exists $T' \in \crp(T)$ such that, for every $(\ACQ{},\REQ{}) \in A$, $\REQ{} \in T'$ and there is no $e \in T'$ with $\REQ{} \POLt[T'] f$.
\end{definition}

\subsection{Partially-ordered requests}

A DP is a predictable deadlock in $T$ when there is a $T' \in \crp(T)$ in which its requests are the last events in their respective threads.
Because two requests never directly depend on each other, this implies that the requests can occur in $T'$ in any order, i.e., they are concurrent.
Hence, the first situation occurs when some requests are not concurrent, making it impossible to find such a $T'$.
We rule out this situation by giving a partial order that is appropriate for determining the concurrency of the requests in a DP.

\begin{definition}[Total Read-Write Order (TRW)]
  \label{d:trw}
  We say events $e$ and $f$ are \emph{conflicting}, denoted~$e \conf f$, if $e$ and $f$ are reads/writes on the same variable and at least one is a write.

  Given trace $T$, we define the \emph{Total-Read-Write order (TRW)} as a relation $\TRWLt[T]$ on $\evts(T)$, where $e \TRWLt[T] f$ if either of the following conditions holds:
  \begin{description}
    \item[\cond{TRW-PO}]
      $e \POLt[T] f$.
    \item[\cond{TRW-Conf}]
      $e \conf f$ and $e \TrLt[T] f$.
    \item[\cond{TRW-Rel}]
      There are acquires $a_1,a_2 \in T$ on the same lock such that (i)~$e$ is the release matching $a_1$, (ii)~$a_1 \TrLt[T] a_2$, (iii)~$f \in \CS[T](a_2)$, and (iv)~$a_1 \TRWLt[T] f$.
    \item[\cond{TRW-Tr}]
      There exists $g \in T$ such that $e \TRWLt[T] g \TRWLt[T] f$.
  \end{description}

  A DP $A$ in well-formed trace $T$ satisfies Condition~\cond{DP-TRW} if neither $\REQ{} \TRWLt[T] \REQ{}'$ nor $\REQ{}' \TRWLt[T] \REQ{}$ (denoted $\REQ{} \TRWConc[T] \REQ{}'$) for every distinct $(\ACQ{},\REQ{}),(\ACQ{}',\REQ{}') \in A$.
\end{definition}

\subsection{Partially-ordered deadlock patterns}

\begin{definition}
  \label{def:ltDP}
  Given trace $T$, we define \emph{cycle order} as a relation $\DPLt[T]$ on cycles in $T$, where $A \DPLt[T] B$ if, for every $(\ACQ{},\REQ{}) \in A$, there exists $(\ACQ{}',\REQ{}') \in B$ with $q \TrLt[T] q'$ and $a \in \AH[T](\REQ{}')$.

  A DP $A$ in well-formed trace $T$ satisfies Condition~\cond{DP-Block} if there is no cycle $B$ in $T$ such that $B \DPLt[T] A$.
\end{definition}

\subsection{TRW bounded and well-nested traces}

We assume that trace $T$ is \emph{TRW bounded}:
    for every request $q \in T$ and acquire $a \in T$ with matching release $r$, if $a \TRWLt[T] q$ and $r \in T$ implies $q \TRWLt[T] r$, then $\thd(q) = \thd(a)$.

We only consider traces that are \emph{well-nested}, meaning that any thread can only release an acquired lock once all further acquired locks have been released.

\sectionfont{\titlecap}
\section{Restated soundness claim}
\label{sec:soundness-claim}

We repeat the soundness claim from~\cite{DBLP:journals/corr/abs-2502-20070}.

\begin{theorem}[Soundness]
  \label{thm:TRWSound}
  If well-formed trace $T$ is TRW bounded and well-nested, and DP $A$ in $T$ satisfies Conditions~\cond{DP-TRW} and \cond{DP-Block}, then $A$ is a predictable deadlock.
\end{theorem}

\section{Counterexample}
\label{sec:counterexample}

\begin{figure}[t]

\bda{|l|l|l|l|l|}
\hline T_1  & \thread{1} & \thread{2} & \thread{3} & \thread{4}\\ \hline
\eventE{1}  & \forkE{\thread{2}}&&&\\
\eventE{2}  & \forkE{\thread{3}}&&&\\
\eventE{3}  & \forkE{\thread{4}}&&&\\
\eventE{4}  & \lockE{\LKA}&&&\\
\eventE{5}  & \writeE{\VA}&&&\\
\eventE{6}  & \lockE{\LKD}&&&\\
\eventE{7}  & \reqLockE{\LKE}&&&\\
\eventE{8}  & \lockE{\LKE}&&&\\
\eventE{9}  & \unlockE{\LKE}&&&\\
\eventE{10}  & \unlockE{\LKD}&&&\\
\eventE{11}  & \unlockE{\LKA}&&&\\
\eventE{12}  & &\lockE{\LKA}&&\\
\eventE{13}  & &\lockE{\LKC}&&\\
\eventE{14}  & &\writeE{\VB}&&\\
\eventE{15}  & &\unlockE{\LKC}&&\\
\eventE{16}  & &\readE{\VA}&&\\
\eventE{17}  & &\unlockE{\LKA}&&\\
\eventE{18}  & &&\lockE{\LKB}&\\
\eventE{19}  & &&\writeE{\VC}&\\
\eventE{20}  & &&\lockE{\LKC}&\\
\eventE{21}  & &&\readE{\VB}&\\
\eventE{22}  & &&\unlockE{\LKC}&\\
\eventE{23}  & &&\unlockE{\LKB}&\\
\eventE{24}  & &&&\lockE{\LKB}\\
\eventE{25}  & &&&\readE{\VC}\\
\eventE{26}  & &&&\lockE{\LKE}\\
\eventE{27}  & &&&\reqLockE{\LKD}\\
\eventE{28}  & &&&\lockE{\LKD}\\
\eventE{29}  & &&&\unlockE{\LKD}\\
\eventE{30}  & &&&\unlockE{\LKE}\\
\eventE{31}  & &&&\unlockE{\LKB}\\

 \hline \eda{}

   \caption{Counterexample.}
  \label{f:counterex}
\end{figure}

\cref{f:counterex} defines trace $T_1$ where
we use a tabular representation with a separate column for each thread and one
event per row. We write $\eventE{i}$ to refer to the $i$th event in the trace.
For brevity, we omit request events for acquires
$\eventE{4},
\eventE{6},
\eventE{12},
\eventE{13},
\eventE{18},
\eventE{20},
\eventE{24},
\eventE{26}$.

Trace $T_1$ is TRW bounded and well-nested.
We find that $A = \{(\eventE{6},\eventE{7}), (\eventE{26},\eventE{27})\}$
is a deadlock pattern (DP) in $T_1$ and satisfies Conditions~\cond{DP-TRW} and \cond{DP-Block}.
However, $A$ is not a predictable deadlock because
there is no correctly reordered prefix $T'$ in which requests $\eventE{7}$
and $\eventE{27}$ are the last events in their respective threads.
This contradicts~\cref{thm:TRWSound}.

We consider the example in more detail.
The TRW relations derived from $T_1$ are as follows.
\begin{enumerate}

\item  $\eventE{5} \TRWLt[T_1] \eventE{16}$ via \cond{TRW-Conf}.

\item  $\eventE{14} \TRWLt[T_1] \eventE{21}$ via \cond{TRW-Conf}.

\item  $\eventE{19} \TRWLt[T_1] \eventE{25}$ via \cond{TRW-Conf}.

\item  $\eventE{4} \TRWLt[T_1] \eventE{16}$ from above, via \cond{TRW-PO} and \cond{TRW-Tr}.
\item  $\eventE{13} \TRWLt[T_1] \eventE{21}$ from above, via \cond{TRW-PO} and \cond{TRW-Tr}.

\item  $\eventE{18} \TRWLt[T_1] \eventE{25}$ from above, via \cond{TRW-PO} and \cond{TRW-Tr}.

\item $\eventE{11} \TRWLt[T_1] \eventE{16}$ from above and via \cond{TRW-Rel}.

\item $\eventE{15} \TRWLt[T_1] \eventE{21}$ from above and via \cond{TRW-Rel}.

\item $\eventE{23} \TRWLt[T_1] \eventE{25}$ from above and via \cond{TRW-Rel}.

\end{enumerate}

For brevity, we omit some further TRW relations
that can be derived via \cond{TRW-PO} and \cond{TRW-Tr}.

We find that requests $\eventE{7}$ and $\eventE{27}$ are not ordered under TRW.
Hence, \cond{DP-TRW} holds for $A$.
There are no further cycles in $T_1$. Hence, \cond{DP-Block} holds as well.

Next, we show that there exists no correctly reordered prefix $T'$ of trace $T_1$
where requests $\eventE{7}$ and $\eventE{27}$ are the last events
in their respective thread.

Suppose there exists such a correctly reordered prefix $T'$.
We attempt to construct $T'$ by adding events that must be included
due to the conditions formulated in \cref{d:WF} and \cref{def:crp}.
Initially, we assume that $T'$ consists of events $\eventE{7}$ and $\eventE{27}$.

\begin{enumerate}

  \item Due to \cond{CRP-PO}, events $\eventE{24},\eventE{25},\eventE{26}$
must be in $T'$ and must appear before $\eventE{27}$.

 \item Similarily, $\eventE{1},\eventE{2},\eventE{3},\eventE{4},\eventE{5}, \eventE{6}$
appear in $T'$ before $\eventE{7}$.

 \item Due to \cond{CRP-LW}, event $\eventE{19}$ appears before $\eventE{25}$.

 \item Due to \cond{CRP-PO}, event $\eventE{18}$ appears before $\eventE{19}$.

 \item We have two acquires $\eventE{18}$ and $\eventE{24}$ on the same lock in $T'$.
       We distinguish between the following two cases.

       Case $\eventE{18}$ appears before $\eventE{24}$ in $T'$:
       Due to \cond{CRP-WF} and \cond{WF-Acq}, the release $\eventE{23}$
       must appear in $T'$ before $\eventE{24}$.

       The other case is impossible.
       If $\eventE{24}$ would appear before $\eventE{18}$, we would
       need to include the release event $\eventE{31}$.
       But this immediately leads to a contradiction because
       $\eventE{27}$ would no longer be the last event in thread $\thread{4}$.

 \item Due to \cond{CRP-PO}, we also find events $\eventE{20},\eventE{21},\eventE{22}$ in $T'$.

 \item Due to \cond{CRP-LW}, event $\eventE{14}$ must appear before $\eventE{21}$ in $T'$.

 \item Due to \cond{CRP-PO}, events $\eventE{12},\eventE{13}$ must appear
before $\eventE{14}$ in $T'$.

\item We have two acquires $\eventE{4}$ and $\eventE{12}$ on the same lock in $T'$.
  We distinguish between the following two cases.

      Case $\eventE{4}$ appears before $\eventE{12}$ in $T'$:
      Due to \cond{CRP-WF} and \cond{WF-Acq}, the release $\eventE{11}$
      must appear in $T'$ before $\eventE{12}$.

      The other case is impossible.
      If $\eventE{12}$ appears before $\eventE{4}$,
      then events $\eventE{12},...,\eventE{17}$ must appear before $\eventE{4}$.
      But this contradicts \cond{CRP-LW}.

 \item Because of $\eventE{11}$ in $T'$ and \cond{CRP-PO} we encounter $\eventE{7}$.
   This leads to a contradiction because
   our strategy so far is to add events to $T'$ that must appear before $\eventE{7}$.
\end{enumerate}

We conclude that there is no correctly reordered prefix
where $\eventE{7}$ and $\eventE{27}$ are last in their respective threads.

\sectionfont{\titlecap}
\section{Analysis of unsoundness}
\label{sec:analysis}


The soundness proof in~\cite{DBLP:journals/corr/abs-2502-20070} tries to
establish the following result:

\begin{quote}
  If
(1) requests are part of a cyclic lock dependency as characterized by \cond{DP-Cycle}, and
(2) requests are not protected by a common lock as characterized by \cond{DP-Guard}, and
(3) requests are unordered under TRW, and
(4) there is no earlier deadlock
then (5) there is a reordering where the requests appear right next to each other.
\end{quote}

The example in \cref{f:counterex} shows that case (3) is not strong enough to establish (5).
In fact, there is another issue connected to case (2).
We elaborate below.

\begin{figure}[ht]

\bda{|l|l|l|l|l|l|l|}
\hline  T_2 & \thread{1} & \thread{2} & \thread{3} & \thread{4} & \thread{5} & \thread{6}\\ \hline
\eventE{1}  & \forkE{\thread{2}}&&&&&\\
\eventE{2}  & \forkE{\thread{3}}&&&&&\\
\eventE{3}  & \forkE{\thread{4}}&&&&&\\
\eventE{4}  & \forkE{\thread{5}}&&&&&\\
\eventE{5}  & \forkE{\thread{6}}&&&&&\\
\eventE{6}  & &\lockE{\LKA}&&&&\\
\eventE{7}  & &\writeE{\VE}&&&&\\
\eventE{8}  & \lockE{\LKB}&&&&&\\
\eventE{9}  & \writeE{\VA}&&&&&\\
\eventE{10}  & &\readE{\VA}&&&&\\
\eventE{11}  & &\unlockE{\LKA}&&&&\\
\eventE{12}  & &&\lockE{\LKA}&&&\\
\eventE{13}  & &&\writeE{\VB}&&&\\
\eventE{14}  & &&&\readE{\VB}&&\\
\eventE{15}  & &&&\readE{\VE}&&\\
\eventE{16}  & &&&\writeE{\VF}&&\\
\eventE{17}  & &&&\writeE{\VC}&&\\
\eventE{18}  & &&\readE{\VC}&&&\\
\eventE{19}  & &&\unlockE{\LKA}&&&\\
\eventE{20}  & &&&&\readE{\VF}&\\
\eventE{21}  & &&&&\lockE{\LKC}&\\
\eventE{22}  & &&&&\reqLockE{\LKD}&\\
\eventE{23}  & &&&&\lockE{\LKD}&\\
\eventE{24}  & &&&&\unlockE{\LKD}&\\
\eventE{25}  & &&&&\unlockE{\LKC}&\\
\eventE{26}  & &&&&\writeE{\VD}&\\
\eventE{27}  & \readE{\VD}&&&&&\\
\eventE{28}  & \unlockE{\LKB}&&&&&\\
\eventE{29}  & &&&&&\lockE{\LKB}\\
\eventE{30}  & &&&&&\lockE{\LKD}\\
\eventE{31}  & &&&&&\reqLockE{\LKC}\\
\eventE{32}  & &&&&&\lockE{\LKC}\\
\eventE{33}  & &&&&&\unlockE{\LKC}\\
\eventE{34}  & &&&&&\unlockE{\LKD}\\
\eventE{35}  & &&&&&\unlockE{\LKB}\\

 \hline \eda{}

   \caption{Counterexample 2}
  \label{f:trw-1}
\end{figure}

Case (2) ensures that requests are not protected
by a common lock. If this would be the case,
requests cannot appear right next to each other in a trace,
thus, ruling out that there is a predictable deadlock.
To check for common locks we rely on the
standard locks held (aka lock set) construction.
See Definition 4 in the original manuscript~\cite{DBLP:journals/corr/abs-2502-20070}.
The issue is that the standard locks held construction
does not consider inter-thread dependencies where
a lock acquired and released in one thread protects
an event in some other thread.
If such inter-thread dependencies arise, the standard
lock set construction becomes \emph{inaccurate}.

For example, consider trace $T_2$ in~\cref{f:trw-1}.
For brevity, we omit most request events in trace~$T_2$.
There are also conflicting write/read events, e.g., consider $\eventE{13}$ and $\eventE{14}$.
For brevity, we omit additional lock events to resolve these conflicts.
In trace $T_2$ we find that the acquire event $\eventE{8}$ protects event $\eventE{22}$
because for any reordering of $T_2$ we find
that $\eventE{8}$ precedes $\eventE{22}$.
Hence, we can argue that event $\eventE{22}$
is protected by lock $\LKB$;
the underlying lock of the acquire~$\eventE{8}$.
However, $\eventE{8}$ and $\eventE{22}$ arise in distinct
threads and therefore the standard locks held construction
does not include $\LKB$ in the lock set of~$\eventE{22}$.
This shows that the standard lock set construction is inaccurate for trace~$T_2$.

This is a serious issue
because  $B = \{(\eventE{21},\eventE{22}), (\eventE{30},\eventE{31}) \}$
is a DP in $T_2$ (Conditions~\cond{DP-Cycle},\cond{DP-Guard} hold) and also satisfies Conditions~\cond{DP-TRW} and \cond{DP-Block}
and therefore we report $B$ as a deadlock.
This is a false positive because $\eventE{22}$ and $\eventE{31}$
are protected by the common lock $\LKB$.
We wish to exclude such traces by demanding that traces
are TRW bounded.
Below, we repeat the TRW bounded assumption
which is introduced in Section 4.3
in the original manuscript~\cite{DBLP:journals/corr/abs-2502-20070}.

\begin{itemize}
  \item
    We only consider traces in which locks only protect requests in the same thread.
    That is, we assume that trace $T$ is \emph{TRW bounded}:
    for every request $q \in T$ and acquire $a \in T$ with matching release $r$, if $a \TRWLt[T] q$ and $r \in T$ implies $q \TRWLt[T] r$, then $\thd(q) = \thd(a)$.
\end{itemize}

However, the TRW bounded assumption is too weak to avoid inaccuracies
of the standard lock set construction.
All acquires in~$T_2$ satisfy the TRW bound assumption.~\footnote{If the acquire is TRW-bounded so is the corresponding request.}
As shown above, we (falsely) report a deadlock
because condition~\cond{DP-Guard} relies on the standard lock set construction
and misses the fact that
requests $\eventE{22}$ and $\eventE{31}$ are protected by the same lock~$\LKB$

We summarize.
\begin{itemize}
    \item Condition \cond{DP-TRW} is meant to guarantee that requests appear last in their
      respective thread.
         \cond{DP-TRW} is too weak as shown by the example in~\cref{f:counterex}.

       \item The TRW bounded assumption is meant to avoid inaccuracies
          of the standard lock set construction.
     The TRW bounded assumption is too weak as shown by the example in~\cref{f:trw-1}.
\end{itemize}


\sectionfont{\titlecap}
\section{Unaffected results}
\label{sec:unaffected}

The completeness result stated by Theorem~2~\cite{DBLP:journals/corr/abs-2502-20070} for a weakened TRW variant, referred to as PWR, still holds.

\begin{definition}[Program-Write-Release Order (PWR)]
  \label{d:pwr}
  Given trace $T$, we define the Program-Write-Release order (PWR) as a relation $\PWRLt[T]$ on $\evts(T)$, where $e \PWRLt[T] f$ if either of the following conditions hold:
  \begin{description}
    \item[\cond{PWR-PO}]
      $e \POLt[T] f$.
    \item[\cond{PWR-LW}]
      $e$ is the last write of read $f$ w.r.t.\ $T$.
    \item[\cond{PWR-Rel}]
      There are acquires $a_1,a_2 \in T$ on the same lock such that (i)~$e$ is the release matching $a_1$, (ii)~$a_1 \TrLt[T] a_2$, (iii)~$f \in \CS[T](a_2)$, and (iv)~$a_1 \PWRLt[T] f$.
    \item[\cond{PWR-Tr}]
      There exists $g \in T$ such that $e \PWRLt[T] g \PWRLt[T] f$.
  \end{description}
  A DP $A$ in well-formed trace $T$ satisfies Condition~\cond{DP-PWR} if neither $q \PWRLt[T] q'$ nor $q' \PWRLt[T] q$ (denoted $q \PWRConc[T] q'$) for every distinct $(a,q),(a',q') \in A$.
\end{definition}

\noindent
Thus, PWR is defined exactly as TRW (\cref{d:trw}), except for Condition~\cond{PWR-LW}.

The empirical measurements show that while TRW and PRW are unsound
in theory, this does not seem to be an issue in practice.
For an extensive set of benchmarks, TRW and PRW report the same deadlock patterns
and all are confirmed as true positives.
See Table~1~\cite{DBLP:journals/corr/abs-2502-20070}.

\sectionfont{\titlecap}
\section{Possible repair directions}
\label{sec:repair-dir}

We see two possible directions to establish soundness.
One direction is to impose even stricter conditions on the partial order.
The partial order must guarantee that unordered events can appear
right next to each other in some reordering.
Likely, this leads to a high number of false negatives.
Another direction is to apply additional ``vindication'' checks.
For example, expose the deadlock either via some symbolic reasoning
at the level of traces or force reexecution of the program~\cite{conf/oopsla/KalhaugeP18,conf/spin/Sorrentino15,conf/pldi/JoshiPSN09,Samak:2014:TDD:2692916.2555262,conf/fse/CaiYWQP21,conf/pldi/TuncMPV23}.
This means that to establish soundness the computational cost increases.
We believe that the discussion and counterexamples presented here precisely describe
the boundary of what can be achieved using partial orders alone.